\documentclass{LMCS}
\pdfoutput=1

\def\doi{8(3:3)2012}
\lmcsheading%
{\doi}
{1--44}
{}
{}
{Feb.~29, 2012}
{Oct.~10, 2012}
{}
 
\usepackage{amssymb}
\usepackage{amsmath}
\usepackage{stmaryrd}
\usepackage{bbm}
\usepackage{amsfonts}
\usepackage{amsthm}
\usepackage{enumerate}
\usepackage{hyperref}
\usepackage{upgreek} 
\usepackage{cmll}
\usepackage{float}
  \floatstyle{boxed} 
  \restylefloat{figure}
\usepackage{subfigure}
\usepackage{proof}
\usepackage{tikz}
\usepackage[all]{xy}
\usepackage{enumerate,hyperref}

\newcommand{\rel}[1]{\ \texttt{#1}\ }
\newcommand{\nak}[1]{\tilde{#1}}
\newcommand{\Bag}[1]{[#1]} 
\newcommand{\Proj}[1]{\pi_{#1}}
\newcommand{\eval}{\mathrm{ev}}
\newcommand{\restr}{\hspace{-4pt}\upharpoonright}
\newcommand{\dom}{\mathrm{dom}}
\newcommand{\curry}{\Uplambda}
\newcommand{\size}[1]{\mathrm{size}(#1)}
\newcommand\msize[1]{\mathrm{size_m}(#1)}
\newcommand{\mge}{>_{\mathrm{m}}}
\newcommand{\nat}{\mathbf{N}}

\newcommand{\one}{\mathbf{1}}
\newcommand{\bool}{\mathbf{2}}
\newcommand{\st}{ \mid }
\newcommand{\perm}[1]{\mathfrak{S}_{#1}}
\newcommand{\card}[1]{\# #1}
\newcommand{\bag}[1]{[#1]}
\newcommand{\Omegatuple}[1]{\Mfin{#1}^{(\omega)}}
\newcommand{\MRel}{\bold{M\hspace{-1pt}Rel}}
\newcommand{\Pow}[1]{\cP(#1)}
\newcommand{\Powf}[1]{\cP_{\mathrm{f}}(#1)}
\newcommand{\Id}[1]{\mathrm{Id}_{#1}}
\newcommand{\comp}{\circ}
\newcommand{\With}[2]{{#1}\with{#2}}
\newcommand{\Termobj}{\mathbbm{1}}
\newcommand{\App}{\mathrm{Ap}}
\newcommand{\Abs}{\uplambda}
\newcommand{\Funint}[2]{[{#1}\!\imp\!{#2}]}

\newcommand{\dlam}{\ensuremath{\partial\lambda}}
\newcommand{\dzlam}{\ensuremath{\partial_0\lambda}}
\newcommand{\lam}{\ensuremath{\lambda}}
\newcommand{\bang}{\oc}

\renewcommand{\hole}[1]{\llparenthesis #1\rrparenthesis}
\newcommand{\paral}{\vert}
\newcommand{\FV}{\mathrm{FV}}
\newcommand{\Set}[1]{\Lambda^{#1}}
\newcommand{\FSet}[1]{\Lambda^{#1}_{\bang}}
\newcommand{\ContSet}{\Set{\gt}_{\hole{\cdot}}}
\newcommand{\FContSet}{\Set{\gt\bang}_{\hole{\cdot}}}
\newcommand{\sums}[1]{\bool\langle\Set{#1}\rangle}
\newcommand{\Fsums}[1]{\bool\langle\FSet{#1}\rangle}
\newcommand{\msto}{\twoheadrightarrow}
\newcommand{\toh}{\to_{h}}
\newcommand{\mstoh}{\msto_{h}}
\newcommand{\lsubst}[2]{\langle #2 / #1 \rangle}	

\newcommand{\subst}[2]{\{ #2 / #1 \}}	
\newcommand{\dg}[2]{\mathrm{deg}_{#1}(#2)} 
\newcommand{\obsle}{\sqsubseteq_{\mathcal{O}}}
\newcommand{\obseq}{\approx_{\mathcal{O}}}
\newcommand{\Tobsle}{\sqsubseteq^{\gt}_{\mathcal{O}}}
\newcommand{\Tobseq}{\approx^{\gt}_{\mathcal{O}}}
\newcommand{\Fobsle}{\sqsubseteq^{\gt\bang}_{\mathcal{O}}}
\newcommand{\Fobseq}{\approx^{\gt\bang}_{\mathcal{O}}}
\newcommand{\TE}[1]{#1^{\circ}} 
\newcommand{\at}{\!::\!}
\newcommand{\rank}{\mathrm{rk}}
\newcommand{\len}{\ell}
\newcommand{\Int}[1]{\llbracket #1\rrbracket} 
\newcommand{\trm}[1]{#1^{\textrm{--}}}
\newcommand{\cont}[2]{#1^{+}\hole{#2}}
\newcommand{\Mfin}[1]{\mathcal{M}_{\mathrm{f}}(#1)}
\newcommand{\mcup}{\uplus}

\newcommand{\Pair}[2]{\langle{#1},{#2}\rangle}

\newcommand{\gramm}{\mathrel{::=}}
\newcommand{\ass}{:=}
\renewcommand{\iff}{\Leftrightarrow}
\newcommand{\seq}[1]{\vec{#1}}
\newcommand{\imp}{\Rightarrow}

\newcommand{\sA}{\mathbb{A}}
\newcommand{\sB}{\mathbb{B}}

\newcommand{\sM}{\mathbb{M}}
\newcommand{\sN}{\mathbb{N}}
\newcommand{\sL}{\mathbb{L}}
\newcommand{\sH}{\mathbb{H}}
\newcommand{\sP}{\mathbb{P}}
\newcommand{\sQ}{\mathbb{Q}}

\newcommand{\sV}{\mathbb{V}}
\newcommand{\sW}{\mathbb{W}}

\newcommand{\ga}{\alpha}
\newcommand{\gb}{\beta}

\newcommand{\gt}{\ensuremath{\tau}}

\newcommand{\gto}{\ensuremath{\bar\tau}}

\newcommand\lab{\mathrm{lab}}

\newcommand{\cA}{\mathcal{A}}

\newcommand{\scC}{\mathcal{C}}
\newcommand{\cD}{\mathcal{D}}
\newcommand{\cE}{\mathcal{E}}
\newcommand{\cF}{\mathcal{F}}
\newcommand{\cG}{\mathcal{G}}

\newcommand{\cI}{\mathcal{I}}

\newcommand{\cL}{\mathcal{L}}

\newcommand{\cN}{\mathcal{N}}

\newcommand{\cP}{\mathcal{P}}

\newcommand{\cR}{\mathcal{R}}

\newcommand{\cV}{\mathcal{V}}

\newcommand{\pts}{.\,.\,}
\newcommand{\cOnv}[1]{#1\!\downarrow} 
\newcommand{\module}[1]{\bool\langle #1 \rangle}

\begin{document}

\title[Full Abstraction for Resource Calculi]{Full Abstraction for the Resource Lambda Calculus with Tests, through Taylor Expansion\rsuper*}

\author[A.~Bucciarelli]{Antonio Bucciarelli\rsuper a}	
\address{{\lsuper a}Univ Paris Diderot, Sorbonne Paris Cit\'e, PPS, UMR 7126, CNRS, F-75205 Paris, France}	
\email{antonio.bucciarelli@pps.jussieu.fr}  

\author[A.~Carrano]{Alberto Carraro\rsuper b}	
\address{{\lsuper b}Universit\`a Ca'Foscari,  Via Torino 155, 30172 Mestre, Venice, Italia}	
\email{acarraro@dsi.unive.it}  

\author[T.~Ehrhard]{Thomas Ehrhard\rsuper c}	
\address{{\lsuper c}CNRS, PPS, UMR 7126, Univ Paris Diderot, Sorbonne Paris Cit\'e, F-75205 Paris, France}	
\email{thomas.ehrhard@pps.jussieu.fr}  

\author[G.~Manzonetto]{Giulio Manzonetto\rsuper d}	
\address{{\lsuper d}Univ Paris 13, Sorbonne Paris Cit\'e, LIPN, UMR 7030, CNRS, F-93430 Villetaneuse, France}	
\email{giulio.manzonetto@lipn.univ-paris13.fr}  


\titlecomment{{\lsuper*}This paper is an extended version of ``Full Abstraction for Resource Calculus with Tests'' appeared in the proceedings of the 20\textsuperscript{th} Conference on Computer Science Logic (CSL'11) \cite{BucciarelliCEM11}.
Work partially supported by the international ANR-NSFC Project \emph{Locali}.}


\begin{abstract}
 \noindent  We study the semantics of a resource-sensitive extension of the lambda calculus in a canonical reflexive object of a category of sets and relations, a relational version of Scott's
original model of the pure lambda calculus. This calculus is related to Boudol's resource calculus and is derived from Ehrhard and Regnier's differential extension of Linear Logic and of the lambda calculus.  We extend it with new constructions, to be understood as implementing a very simple exception mechanism, and with a ``must'' parallel composition. These new operations allow to associate a context of this calculus with any point of the model and to prove full abstraction for the finite sub-calculus where ordinary lambda calculus application is not allowed. The result is then extended to the full calculus by means of a Taylor Expansion formula.
As an intermediate result we prove that the exception mechanism is not essential in the finite sub-calculus.
\end{abstract}

\maketitle

\section{Introduction}

In concurrent calculi like CCS \cite{Milner80}, guarded processes are resources
that can be used only once by other processes. This fundamental 
linearity of resources leads naturally to non-determinism,
since several agents  (senders and receivers) 
can interact on the same channel. In general, 
various synchronization scenarios are possible, giving  rise to 
different behaviours. On the other hand in the $\lambda$-calculus \cite{Bare}, a
function (receiver) can duplicate its argument (sender) arbitrarily. 
Thanks to this asymmetry, the $\lambda$-calculus enjoys a
strong determinism (namely, the Church-Rosser Theorem), but for the same reason it 
lacks any form of control on resource handling. 
\medskip

{\bf Resource Lambda Calculi.} Resource $\lambda$-calculi stem from an attempt to combine 
the functionality of the $\lambda$-calculus and 
the resource-sensitivity of process calculi.
Boudol has been the first to design a resource-conscious functional programming language, 
the \emph{resource $\lambda$-calculus}  \cite{Boudol93}, extending the usual one 
along two directions. 
First, a function is not necessarily applied to a single argument, but can also be
applied to a multiset of arguments called \emph{resources}.
Second, a resource can be either linear (it must be used exactly once) or reusable (it can be used \emph{ad libitum}). 
In this context, the evaluation of a function applied to a multiset of resources gives rise to several possible choices, 
corresponding to the different possibilities of distributing the resources in the multiset among the occurrences of the formal parameter.
From the viewpoint of concurrent programming, this  was a natural step to take since one
of the main features of this programming setting is the consumption of
resources which cannot be copied. Milner's
$\pi$-calculus \cite{Milner93} features this phenomenon in great generality, and Boudol's
calculus keeps track of it in a functional setting.

Together with Regnier, Ehrhard observed that this idea of resource consumption can be
understood as resulting from a \emph{differential} extension of
$\lambda$-calculus (and of Linear Logic)~\cite{EhrhardR03}. 
Instead of considering two kinds of resources, they defined two kinds of applications: the \emph{ordinary} application and a \emph{linear} one.
In a simply typed setting, linear application of a term $M:A\to B$ to a multiset made of $n$ terms
$N_1,\dots,N_n:A$, combined with ordinary application to a term $N:A$,
corresponds to computing $M^{(n)}(N)(N_1,\dots,N_n)$, where $M^{(n)}$ is the
$n$-th derivative of $M$ which is of type $A\to(A^n\to B)$ and associates a
symmetric $n$-linear map with any element of $A$. The symmetry of this
multilinear map corresponds to Schwarz's Theorem and
is implemented in the resource $\lambda$-calculus by the use of multisets for
representing linear applications. 
A notable advantage of this approach is that it allows to apply powerful methods from differential calculus in the context of \lam-calculus.
For instance, iterated differentiation yields very naturally a Taylor expansion formula, which consists in expanding the ordinary application into
several linear applications of the differential \lam-calculus. 
More precisely, if $M : A\to B$ and $N : A$ are \lam-terms, then the Taylor expansion of $MN$ is given by
$$
	\TE{(MN)} =\sum_{n=0}^\infty \frac{1}{n!} M^{(n)}(0)(\underbrace{N,\dots,N}_{n\textrm{ times}})
$$
in analogy with the standard Taylor formula of the entire functions.
The Taylor expansion has been studied in \cite{EhrhardR08} where the authors relate it to the B\"ohm tree of a \lam-term, giving the intuition that the former is a resource conscious improvement of the latter.

The main difference between Boudol's resource $\lambda$-calculus and Ehrhard and Regnier's
differential $\lambda$-calculus is that the first is lazy --- this means that in many cases linear substitutions must be delayed. 
To that effect, the calculus features a linear explicit substitution mechanism. 
Moreover, it implements a fixed reduction strategy similar to linear head reduction.
Therefore, Boudol's calculus is not an extension of the ordinary 
\lam-calculus.
Also, the resource \lam-calculus is rather affine than linear, since depletable resources cannot be duplicated but can be erased.
Another difference lies in the respective origins of these calculi: 
the resource \lam-calculus originates from syntactical considerations related to the theory 
of concurrent processes, while the differential one arises from denotational models of linear logic where the existence of differential
operations has been observed. These models are based on the well-known relational model
of Linear Logic \cite{Girard88}, and the interpretation of the new differential constructions
is as natural and simple as the interpretation of the ordinary Linear Logic
constructions. 

In this paper we work with a resource-sensitive \lam-calculus because our techniques depend on the linear logic structure underlying the calculus and on the presence
 of a Taylor expansion formula. Two main syntaxes have been proposed for the differential $\lambda$-calculus: Ehrhard and Regnier's original one \cite{EhrhardR03},
 simplified by Vaux in~\cite{Vaux07}, and Tranquilli's \emph{resource calculus} of \cite{Tranquilli10} whose syntax  is close to Boudol's one.
 These calculi share a common semantical backbone as well as similar connections with differential Linear Logic and proof nets. 
We adopt roughly Tranquilli's syntax and call our calculus  \emph{\dlam-calculus}. To avoid the problem of handling the coefficients introduced by the Taylor formula 
 we conveniently suppose that the formal sum in the calculus is idempotent; this amounts to saying that we only check whether a term appears in a result, not how many
 times it appears. This is very reasonable when studying convergency properties since $M + M$ converges exactly when $M$ does.

\medskip
{\bf Full Abstraction.} 
A natural problem when a new calculus is introduced is to characterize when two programs are operationally equivalent, 
namely when one can be replaced by the other in every context without noticing any difference with respect to a given observational equivalence. 
In this paper we prove a full abstraction result (a semantical characterization of operational equivalence)
for the \dlam-calculus in the spirit of~\cite{BoudolCL99}. As in that paper, we extend the 
language with a convergence testing mechanism. 
Implicitly, this extension already appears in~\cite{EhrhardL10}, in a differential linear logic setting: it corresponds to the
$0$-ary tensor and par cells. To implement the corresponding extension of the 
$\lambda$-calculus, we
introduce two sorts of expressions: the \emph{terms} (variable, application,
abstraction, ``throw'' $\gto(V)$ where  $V$ is a test) and the \emph{tests}
(empty test, parallel composition of tests and 
``catch'' $\tau(M)$ where  $M$ is a term).
Parallel composition allows to combine tests in such a way that
the combination succeeds if and only if each test succeeds. 
Outcomes of tests (convergence or divergence) are the only observations
allowed in our calculus, and the corresponding contextual equivalence and preorder on terms
constitute our main object of study.

This extended \dlam-calculus, that we call \emph{\dlam-calculus with tests}, has a natural denotational
interpretation in a model of the pure $\lambda$-calculus introduced by
Bucciarelli, Ehrhard and Manzonetto in~\cite{BucciarelliEM07},
which is indeed a denotational model of the differential pure nets
of~\cite{EhrhardL10} as one can check easily. This model is a reflexive
object $\cD$   in the Kleisli category of the linear logic model of sets and relations where $!X$
is the set of all finite multisets over $X$. An element of $\cD$ can be described as a finite
tree which alternates two kinds of layers: \emph{multiplicative layers} where
subtrees are indexed by natural numbers and \emph{exponential layers} where
subtrees are organized as non-empty multisets. To be more precise, $\parr-?$
(negative) pairs of layers alternate with $\otimes-!$ (positive) pairs,
respecting a strict polarity discipline very much in the spirit of
Ludics~\cite{Girard03}. The empty positive multiplicative tree corresponds to the
empty tensor cell and the negative one to the empty par cell.  The
corresponding constructions $\tau$, $\gto$ are therefore quite easy to
interpret.

We use this logical interpretation to turn the elements of $\cD$ into 
\dlam-calculus terms with tests. 
More precisely, with each element $\alpha$ of $\cD$, we
associate a test $\cont{\alpha}{\cdot}$ with a hole $\hole{\cdot}$ for a term, and we show that
$\alpha$ belongs to the interpretation of a (closed) term $M$ iff the test
$\cont{\alpha}{M}$ converges. From this fact, we derive  a full abstraction
result for the fragment of the \dlam-calculus with tests in
which all ordinary applications are trivial, that we call \dzlam-calculus with tests.
To extend this
result to the \dlam-calculus with tests, we
use the Taylor formula introduced in~\cite{EhrhardR03} which
allows to turn any ordinary application into a sum of infinitely many linear
applications of all possible arities. One exploits then the fact that the
Taylor formula holds in the model, as well as a simulation lemma
which relates the head reduction of a term with the head reduction of its
Taylor expansion.

\medskip \textbf{Contributions.}
In Section~\ref{sec:thomas} we provide the abstract categorical framework which is 
needed to interpret the $\dlam$-calculus and its extension with tests.
The syntax and operational semantics of the $\dzlam$-calculus with tests (which is the promotion-free fragment)
are presented in Section~\ref{sec:dzlam-calculus+tests}, while its relational model $\cD$ is described concretely in Section~\ref{sec:rela-sem}.
The definability of the elements of $\cD$ in the \dzlam-calculus with tests is
the main conceptual contribution of this paper --- it shows that, in this setting, the standard syntax \emph{versus} semantics dichotomy is
essentially meaningless.  
From definability it follows easily that the relational model is fully abstract for the $\dzlam$-calculus with tests, as shown in Section~\ref{sec:FA-dzlam-with-Tests}.
This result is analyzed further in Section~\ref{sec:FA-notests-nobang}, where it is proved that in the absence of promotion the test operators do not add any discriminatory 
power to the contexts, thus showing that $\cD$ is also fully abstract for the $\dzlam$-calculus without tests. 

We then focus on the full $\dlam$-calculus with tests.
Section~\ref{sec:FullRC} is devoted to present its syntax, operational
semantics and relational semantics.  In Section~\ref{sec:TaylorExp} we
consider the use of Taylor expansions to reduce the full abstraction
problem for $\dlam$ to its ``\dzlam'' version, thus introducing an
original and promising reduction technique.

\renewcommand{\phi}{\varphi}
\renewcommand\epsilon{\varepsilon}
\renewcommand\sharp{\#}
\newcommand\Eqref[1]{(\ref{#1})}

\newcommand{\Iff}{\quad\hbox{iff}\quad}
\newcommand{\Implies}{\Rightarrow}
\newcommand\Equiv{\Leftrightarrow}
\newcommand{\St}{\mid}

\newcommand{\Ro}{\circ}
\newcommand{\Inf}{\bigwedge}
\newcommand{\Infi}{\wedge}
\newcommand{\Sup}{\bigvee}
\newcommand{\Supi}{\vee}

\newcommand{\arrow}{\rightarrow}
\renewcommand{\Bot}{{\mathord{\perp}}}
\newcommand{\Top}{\top}

\newcommand\Seqempty{\langle\rangle}

\newcommand\Fini{{\mathrm{fin}}}

\def\frownsmile{%
\mathrel{\vbox{\hbox{${\frown}$}\vspace{-2ex}\hbox{${\smile}$}\vspace{-.5ex}}}}
\def\smilefrown{%
\mathrel{\vbox{\hbox{${\smile}$}\vspace{-2ex}\hbox{${\frown}$}\vspace{-.5ex}}}}

\newcommand\Part[1]{{\cal P}({#1})}

\newcommand\Union{\bigcup}

\newcommand{\Linarrow}{\multimap}

\def\frownsmile{%
\mathrel{\vbox{\hbox{${\frown}$}\vspace{-2ex}\hbox{${\smile}$}\vspace{-.5ex}}}}
\def\smilefrown{%
\mathrel{\vbox{\hbox{${\smile}$}\vspace{-2ex}\hbox{${\frown}$}\vspace{-.5ex}}}}

\newcommand\CScoh[3]{{#2}\mathrel{\frownsmile_{{#1}}}{#3}}
\newcommand\CScohs[3]{{#2}\mathrel{{\frown}_{#1}}{#3}}
\newcommand\CScohstr[3]{\CScohs{#1}{#2}{#3}}
\newcommand\CSincoh[3]{{#2}\mathrel{\smilefrown_{{#1}}}{#3}}
\newcommand\CSincohs[3]{{#2}\mathrel{{\smile}_{#1}}{#3}}
\newcommand\CSeq[3]{{#2}\mathrel{{=}_{#1}}{#3}}

\newcommand\Myleft{}
\newcommand\Myright{}

\newcommand\Web[1]{\Myleft|{#1}\Myright|}
\newcommand\Supp[1]{\Myleft|{#1}\Myright|}

\newcommand\Emptymset{[]}
\newcommand\Mset[1]{[{#1}]}
\newcommand\Sset[1]{\langle{#1}\rangle}

\newcommand\MonP{P}
\newcommand\MonPZ{\MonP_0}

\newcommand\Cl[1]{\mbox{\textrm{Cl}}({#1})}
\newcommand\ClP[1]{\mbox{\textrm{Cl}}_{\MonP}({#1})}

\newcommand\Star{\star}
\newcommand\CohName{\mbox{\textbf{Coh}}}
\newcommand\COH[1]{\CohName(#1)}
\newcommand\COHR[2]{\CohName(#1,#2)}

\newcommand\Par[2]{{#1}\parr{#2}}
\newcommand\ITens{\mathrel\otimes}
\newcommand\Tens[2]{{#1}\ITens{#2}}
\newcommand\Tensp[2]{({#1}\ITens{#2})}
\newcommand\Tensexp[2]{{#1}^{\mathord\otimes{#2}}}
\newcommand\IWith{\mathrel{\&}}
\newcommand\BWith{\bigwith}
\newcommand\IPlus{\oplus}
\newcommand\Plus[2]{{#1}\IPlus{#2}}
\newcommand\ILfun{\multimap}
\newcommand\Lfun[2]{{#1}\ILfun{#2}}
\newcommand\Orth[1]{{#1}^\Bot}
\newcommand\Orthr[2]{{#1}^{\Bot(#2)}}
\newcommand\Biorthr[3]{{#1}^{\Bot(#2)\Bot(#3)}}
\newcommand\Triorthr[4]{{#1}^{\Bot(#2)\Bot(#3)\Bot(#4)}}
\newcommand\Orthw[1]{{#1}^\Bot}

\newcommand\Scal[2]{\langle{#1}\mid{#2}\rangle}

\newcommand\Inj{\lambda}

\newcommand\GlobalIndex{I}
\newcommand\Index{\GlobalIndex}
\newcommand\Relbot{\Bot}
\newcommand\MonPZI{{\MonPZ}^\Index}
\newcommand\ClPI[1]{{\ClP{{#1}}}^\Index}
\newcommand\WebI[1]{\Web{#1}^\Index}

\newcommand\Scalb[2]{\Scal{\Bar{#1}}{\Bar{#2}}}

\newcommand\Ortho[2]{{#1}\mathrel{\bot}{#2}}
\newcommand\Orthob[2]{\Bar{#1}\mathrel{\bot}\Bar{#2}}

\newcommand\Biorth[1]{{#1}^{\Bot\Bot}}
\newcommand\Triorth[1]{{#1}^{\Bot\Bot\Bot}}

\newcommand\Relpretens[2]{\cR_{#1}\bar\ITens\cR_{#2}}
\newcommand\Relpreplus[2]{\cR_{#1}\bar\IPlus\cR_{#2}}


\newcommand\RelP[1]{\widetilde{#1}}

\newcommand\Eqw[2]{\delta({#1},{#2})}
\newcommand\Eqwb[2]{\Eqw{\Bar{#1}}{\Bar{#2}}}

\newcommand\PFacts[1]{\cF({#1})}

\newcommand\Facts{\cF(\MonPZI)}

\newcommand\RelL[1]{\overline{#1}}

\newcommand\PRel[1]{R_{#1}}

\newcommand\PFamb[2]{[\Bar{#1},\Bar{#2}]}

\newcommand\Fplus[2]{\Bar{#1}+\Bar{#2}}

\newcommand\Char[1]{\epsilon_{#1}}

\newcommand\Fproj[2]{\pi_{#1}(\Bar{#2})}

\newcommand\One{1}

\newcommand\Pbot[1]{\Bot_{#1}}
\newcommand\PBot[1]{\Bot_{#1}}
\newcommand\PRBot[1]{\Bot_{#1}}
\newcommand\PROne[1]{1_{#1}}

\newcommand\Pproj[1]{\pi_{#1}}

\newcommand\Zext[1]{\zeta_{\Bar{#1}}}
\newcommand\Aext[1]{\bar\zeta_{\Bar{#1}}}

\newcommand\Mall{\hbox{\textsf{MALL}}}
\newcommand\RMall{\Mall(\Index)}
\newcommand\RMallr[1]{\Mall({#1})}
\newcommand\FDom[1]{d({#1})}

\newcommand\RBot[1]{\Bot_{#1}}
\newcommand\ROne[1]{\One_{#1}}

\newcommand\Seq[1]{\vdash{#1}}
\newcommand\Seqpolr[1]{\vdash_\POLS{#1}}
\newcommand\Seqv[1]{\Seq{#1}}
\newcommand\Seqt[3]{\vdash_{\mathsf{t}}{#1}:{#2}\mid{#3}}
\newcommand\Seqn[2]{\vdash_{\mathsf{n}}{#1\mid#2}}
\newcommand\Seqnfree[2]{\vdash_{\mathsf{1}}{#1\mid#2}}
\newcommand\Seql[2]{\vdash_{\mathsf{0}}{#1\mid#2}}
\newcommand\Seqnp[3]{\vdash_{\mathsf{p}}{#1:#2\mid#3}}
\newcommand\Seqtfree[1]{\vdash_1{#1}}
\newcommand\Ccup{,}

\newcommand\RSeq[2]{\vdash_{#1}{#2}}

\newcommand\Restr[2]{{#1}|_{#2}}
\newcommand\FRestr[2]{{#1}|_{#2}}

\newcommand\FSem[1]{{#1}^{*}}
\newcommand\PSem[1]{{#1}^{*}}

\newcommand\FFamb[2]{{#1}\langle\Bar{#2}\rangle}

\newcommand\Premskip{\hskip1cm}

\newcommand\Forg[1]{\underline{#1}}

\newcommand\Phase[1]{{#1}^\bullet}

\newcommand\Punit[1]{1^{#1}}

\newcommand\Reg[1]{R_{#1}}

\newcommand\Cont[1]{{#1}^\circ}

\newcommand\Neutral{e}
\newcommand\RNeutral[1]{\Neutral_{#1}}

\newcommand\POne{1}

\newcommand\relstack[2]{\underset{#2}{#1}}
\newcommand\RPlus[4]{{#3}\relstack\IPlus{{#1},{#2}}{#4}}
\newcommand\RWith[4]{{#3}\relstack\IWith{{#1},{#2}}{#4}}

\newcommand\PReq[1]{=_{#1}}

\newcommand\Fam[2]{\hbox{\textrm{Fam}}_{#1}(#2)}
\newcommand\Famfunc[1]{\hbox{\textrm{Fam}}_{#1}}

\newcommand\Pval[1]{\rho_{#1}}

\newcommand\Fcoh{\mathsf C}
\newcommand\Fincoh{\overline\Fcoh}
\newcommand\Feq{\mathsf E}

\newcommand\RZero{\Zero}
\newcommand\RTop{\Top}
\newcommand\EmptyFam{\emptyset}
\newcommand\Partial[1]{\Omega_{#1}}

\newcommand\PMall{\hbox{\textsf{MALL}}_\Omega}
\newcommand\PRMall{\Mall_\Omega(\Index)}
\newcommand\Homo[1]{H(#1)}
\newcommand\RProd[2]{\Lambda_{#1}{({#2})}}
\newcommand\FProd[1]{\widetilde{#1}}

\newcommand\CltoPCl[1]{{#1}^P}

\newcommand\Partfin[1]{{\cP_\Fini}({#1})}

\newcommand\RLL{\LL(\Index)}
\newcommand\RLLP{\LL^+(\Index)}
\newcommand\RLLext{\LL^\mathrm{ext}(\Index)}

\newcommand\IExcl{{\mathord{!}}}

\newcommand\RExcl[1]{\IExcl_{#1}}
\newcommand\RInt[1]{\IInt_{#1}}

\newcommand\Fempty[1]{0_{#1}}

\newcommand\FAct[1]{{#1}_*}

\newcommand\Excl[1]{\IExcl{#1}}
\newcommand\Excls[1]{\IExcl_{\mathsf s}{#1}}

\newcommand\Noindex[1]{\langle{#1}\rangle}

\newcommand\Card[1]{\#{#1}}
\newcommand\Multi[1]{\mathsf m({#1})}

\newcommand\FamRestr[2]{{#1}|_{#2}}

\newcommand\Pact[1]{{#1}_*}

\newcommand\Pinj[1]{\zeta_{#1}}

\newcommand\Locun[1]{1^J}

\newcommand\Isom\simeq

\newcommand\FGraph[1]{\mathsf{g}({#1})}
\newcommand\GPFact[1]{\mathsf f_0({#1})}
\newcommand\GFact[1]{\mathsf f({#1})}

\newcommand\NUCS{non-uniform coherence space}
\newcommand\NUCSb{non-uniform coherence space }
\newcommand\NUCSs{non-uniform coherence spaces}
\newcommand\NUCSsb{non-uniform coherence spaces }
\newcommand\Nucs{\mathbf{nuCS}}

\newcommand\Comp{\mathrel\circ}

\newcommand\Funinv[1]{#1^{-1}}

\newcommand\Reindex[1]{{#1}^*}

\newcommand\Locbot[1]{\Bot({#1})}
\newcommand\Locone[1]{1({#1})}
\newcommand\LocBot{\Locbot}
\newcommand\LocOne{\Locone}

\newcommand\INJ[1]{\mathcal I({#1})}
\newcommand\COHP[1]{\mathcal C({#1})}
\newcommand\FuncFam[1]{\mathrm{Fam}_{#1}}
\newcommand\SET{\mathbf{Set}}
\newcommand\Locmod[2]{{#1}({#2})}
\newcommand\FuncPhase[1]{\mathcal{F}_{#1}}
\newcommand\Trans[2][]{{\widehat{#2}}_{#1}}
\newcommand\Useful[1]{\Web{#1}_{\mathrm u}}

\newcommand\Limpl[2]{{#1}\Linarrow{#2}}

\newcommand\ConstFam[2]{{#1}^{#2}}

\newcommand\Center[1]{\mathcal C({#1})}

\newcommand\Derel[1]{\mathrm{d}_{#1}}

\newcommand\Myid{\operatorname{Id}}

\newcommand\Mspace[2]{{#2}^*_{#1}}

\newcommand\CStoMS[1]{{#1}^+}
\newcommand\MStoCS[1]{{#1}^-}

\newcommand\PhaseCoh[1]{{\mathbf{Coh}}_{#1}}

\newcommand\Nat{{\mathbb{N}}}
\newcommand\Natnz{{\Nat^+}}

\newcommand\Rien{\Omega}

\newcommand\DenEq[2]{{#1}\sim{#2}}

\newcommand\VMon[2]{\langle{#1},{#2}\rangle}
\newcommand\VMonc[3]{\langle{#1},{#2},{#3}\rangle}

\newcommand\Multind[1]{\begin{substack}#1\end{substack}}
\newcommand\Biind[2]{\genfrac{}{}{0pt}{1}{#1}{#2}}

\newcommand\Adapt[1]{\cA_{#1}}

\newcommand\Rllp[1]{\widetilde{#1}}

\newcommand\Emptyfun[1]{0_{#1}}

\newcommand\Multisetfin[1]{\mathcal{M}_{\mathrm{fin}}({#1})}

\newcommand\Myparag[1]{\noindent\textbf{#1.}\ }

\newcommand\Webinf[1]{\Web{#1}_\infty}
\newcommand\Fin[1]{\mathsf{F}({#1})}

\newcommand\Phasefin{\mathbb{F}}

\newcommand\Fspace[1]{\mathrm f{#1}}
\newcommand\Space[1]{\mathrm m{#1}}

\newcommand\Prom[1]{{#1}^!}
\newcommand\Promp[1]{{(#1)}^!}
\newcommand\Promm[1]{{#1}^{!!}}

\newcommand\Ssupp[1]{\Supp{#1}}

\newcommand\FIN{\mathbf{Fin}}
\newcommand\FINH[2]{\FIN({#1},{#2})}
\newcommand\FINF[1]{\FIN[{#1}]}
\newcommand\SETINJ{\mathbf{Inj}}
\newcommand\EMB{\FIN_{\mathrm e}}

\newcommand\SNat{\mathsf N}
\newcommand\Snat{\mathsf N}

\newcommand\Iter[1]{\mathsf{It}_{#1}}
\newcommand\Piter[2]{\Iter{#1}^{(#2)}}
\newcommand\Case[1]{\mathsf{Case}_{#1}}
\newcommand\Fix[1]{\mathsf{Y}_{#1}}
\newcommand\Pfix[2]{\mathsf{Y}_{#1}^{(#2)}}
\newcommand\Ifthenelse[3]{\mathtt{if}\,{#1}\,\mathtt{then}\,{#2}\,%
\mathtt{else}\,{#3}}

\newcommand\Trace[1]{\mathsf{Tr}{#1}}
\newcommand\Tracelin[1]{\mathsf{tr}{#1}}
\newcommand\Ptrace[1]{\|{#1}\|}
\newcommand\Enat[1]{\overline{#1}}
\newcommand\FA[1]{\forall{#1}}
\newcommand\Trcl[2]{\langle {#1},{#2}\rangle}
\newcommand\Faproj[1]{\epsilon^{#1}}
\newcommand\Faintro[3]{\Lambda^{{#1},{#2}}({#3})}

\newcommand\Bool{\mathbf{Bool}}
\newcommand\True{\mathbf t}
\newcommand\False{\mathbf f}

\newcommand\Tarrow{\arrow}
\newcommand\Diffsymb{\mathsf D}

\newcommand\Diff[3]{\mathrm D_{#1}{#2}\cdot{#3}}
\newcommand\Diffexp[4]{\mathrm D_{#1}^{#2}{#3}\cdot{#4}}
\newcommand\Diffvar[3]{\mathrm D_{#1}{#2}\cdot{#3}}
\newcommand\Diffterm[3]{\mathrm D_{#1}{#2}\cdot{#3}}
\newcommand\Zero{0}
\newcommand\Appv[2]{\App{#1}{\Tofval{#2}}}
\newcommand\Appp[2]{({#1}\,{#2})}
\newcommand\Apppv[2]{\Appp{#1}{\Tofval{#2}}}
\newcommand\Applp[2]{({#1}){#2}}
\newcommand\Applpv[2]{\Applp{#1}{\Tofval{#2}}}
\newcommand\Apprp[2]{{#1}({#2})}
\newcommand\Apprpv[2]{\Apprp{#1}{\Tofval{#2}}}
\newcommand\Applrp[2]{({#1})({#2})}
\newcommand\Diffp[3]{\frac{\partial{#1}}{\partial{#2}}\cdot{#3}}
\newcommand\Derp[3]{\frac{\partial{#1}}{\partial{#2}}\cdot{#3}}
\newcommand\Derplist[4]{%
\frac{\partial^{#4}{#1}}{\partial#2_1\cdots\partial#2_{#4}}%
\cdot\left(#3_1,\dots,#3_{#4}\right)}
\newcommand\Derplistexpl[6]{%
\frac{\partial^{#6}{#1}}{\partial{#2}\cdots\partial{#3}}%
\cdot\left({#4},\dots,{#5}\right)}
\newcommand\Derpmult[4]{%
\frac{\partial^{#4}{#1}}{\partial{#2}^{#4}}%
\cdot\left(#3_1,\dots,#3_{#4}\right)}
\newcommand\Derpm[4]{%
\frac{\partial^{#4}{#1}}{\partial{#2}^{#4}}%
\cdot\left({#3}\right)}
\newcommand\Derpmultgros[4]{%
\frac{\partial^{#4}}{\partial{#2}^{#4}}\left({#1}\right)%
\cdot\left(#3_1,\dots,#3_{#4}\right)}
\newcommand\Derpmultbis[4]{%
\frac{\partial^{#4}{#1}}{\partial{#2}^{#4}}%
\cdot{#3}}
\newcommand\Derpmultbisgros[4]{%
\frac{\partial^{#4}}{\partial{#2}^{#4}}\left({#1}\right)%
\cdot{#3}}
\newcommand\Derppar[3]{\left(\Derp{#1}{#2}{#3}\right)}
\newcommand\Derpgros[3]{\frac{\partial}{\partial{#2}}\Bigl({#1}\Bigr)%
\cdot{#3}}
\newcommand\Derpdeux[4]{\frac{\partial^2{#1}}{\partial{#2}\partial{#3}}%
\cdot{#4}}
\newcommand\Derpdeuxbis[5]{\frac{\partial^{#5}{#1}}{\partial{#2}\partial{#3}}%
\cdot{#4}}
\newcommand\Diffpv[3]{\frac{\partial}{\partial{#2}}{#1}\cdot{#3}}
\newcommand\Diffpd[5]{\frac{\partial^2{#1}}{\partial{#2}\partial{#3}}%
\cdot\left({#4},{#5}\right)}

\newcommand\Paragraph[1]{\smallbreak\noindent\textbf{#1}}

\newcommand\Redamone{\mathrel{\beta^{0,1}}}
\newcommand\Redonecan{\mathrel{{\bar\beta}^1_{\mathrm D}}}
\newcommand\Redpar{\mathrel{\rho}}
\newcommand\Redparcan{\mathrel{\bar\rho}}

\newcommand\Sat[1]{#1^*}

\newcommand\Can[1]{\left\langle{#1}\right\rangle}
\newcommand\Candiff[3]{\Delta_{#1}{#2}\cdot{#3}}


\newcommand\List[3]{#1_{#2},\dots,#1_{#3}}
\newcommand\Listbis[3]{#1_{#2}\dots #1_{#3}}
\newcommand\Listc[4]{{#1}_{#2},\dots,{#4},\dots,{#1}_{#3}}
\newcommand\Listbisc[4]{{#1}_{#2}\dots{#4}\dots{#1}_{#3}}
\newcommand\Absent[1]{\widehat{#1}}
\newcommand\Kronecker[2]{\delta_{{#1},{#2}}}

\newcommand\Eqindent{\quad}

\newcommand\Subst[3]{{#1}\left[{#2}/{#3}\right]}
\newcommand\Substz[2]{{#1}\left[0/{#2}\right]}
\newcommand\Substpar[3]{\left({#1}\right)\left[{#2}/{#3}\right]}
\newcommand\Substbis[2]{{#1}\left[{#2}\right]}

\newcommand\Span[1]{\overline{#1}}
\newcommand\SN{\mathcal{N}}
\newcommand\WN{\mathcal{N}}
\newcommand\Extred[1]{\mathop{\mathrm R}^{\textrm{ext}}{(#1)}}
\newcommand\Onered[1]{\mathop{\mathrm R}^1{(#1)}}
\newcommand\NO[1]{\mathop{\mathrm N}({#1})}
\newcommand\NOD[1]{\mathop{\mathrm N_0}({#1})}

\newcommand\Freemod[2]{{#1}\left\langle{#2}\right\rangle}
\newcommand\Mofl[1]{{#1}^\#}
\newcommand\Mlext[1]{\widetilde{#1}}
\newcommand\Difflamb{\Lambda_D}
\newcommand\Terms{\Lambda_{\mathrm D}}
\newcommand\Diffmod{\Freemod R\Difflamb}

\newcommand\Factor[1]{{#1}!}
\newcommand\Binom[2]{\left({{#1}\atop{#2}}\right)}
\newcommand\Multinom[2]{\left[{{#1}\atop{#2}}\right]}
\newcommand\Suite[1]{\bar#1}
\newcommand\Head[1]{\mathrel{\tau^{#1}}}

\newcommand\Headlen[2]{{\mathrm{L}}({#1},{#2})}
\newcommand\Betaeq{\mathrel{\mathord\simeq_\beta}}
\newcommand\Betadeq{\mathrel{\mathord\simeq_{\beta_{\mathrm D}}}}

\newcommand\Vspace[1]{E_{#1}}
\newcommand\Real{\mathbf{R}}

\newcommand\Ring{R}
\newcommand\Linapp[2]{\left\langle{#1}\right\rangle{#2}}
\newcommand\Fmod[2]{{#1}\langle{#2}\rangle}
\newcommand\Fmodr[1]{{\Ring}\langle{#1}\rangle}
\newcommand\Imod[2]{{#1}\langle{#2}\rangle_\infty}
\newcommand\Imodr[1]{{\Ring}\langle{#1}\rangle_\infty}
\newcommand\Res[2]{\langle{#1},{#2}\rangle}
\newcommand\Funofmat[1]{\widehat{#1}}
\newcommand\Transp[1]{{}^{\mathrm{t}}\!{#1}}
\newcommand\Idmat{\operatorname{I}}
\newcommand\FINMOD[1]{\operatorname{\mathbf{Fin}\,}({#1})}
\newcommand\Bcanon[1]{e_{#1}}

\newcommand\Mfinc[2]{\mathcal M_{#1}({#2})}
\newcommand\Expt[2]{\operatorname{exp}_{#1}({#2})}
\newcommand\Dexpt[2]{\operatorname{exp}'_{#1}({#2})}

\newcommand\Semtype[1]{{#1}^*}
\newcommand\Semterm[2]{{#1}^*_{#2}}

\newcommand\Elem[1]{\operatorname{Elem}({#1})}
\newcommand\Fcard[1]{\operatorname{Card}({#1})}
\newcommand\Linhom[2]{\cL({#1},{#2})}
\newcommand\Linhombil[2]{\cL^2({#1},{#2})}
\newcommand\Linhommulti[3]{\cL^{#1}({#2},{#3})}
\newcommand\Hom[2]{\operatorname{Hom}(#1,#2)}
\newcommand\Compana{\circ_\cA}

\newcommand\Der[1]{\mathsf d_{#1}}
\newcommand\Dernat{\mathsf d}
\newcommand\Derm[2]{\mathsf d_{#1}^{(#2)}}
\newcommand\Digg[1]{\mathsf p_{#1}}
\newcommand\Diggnat{\mathsf p}
\newcommand\Coder[1]{{\overline{\mathsf d}}_{#1}}
\newcommand\Codernat{{\overline{\mathsf d}}}
\newcommand\Coderm[2]{{\overline{\mathsf d}}_{#1}^{(#2)}}
\newcommand\Contr[1]{\mathsf c_{#1}}
\newcommand\Contrnat{\mathsf c}
\newcommand\Contrm[2]{\mathsf c_{#1}^{(#2)}}
\newcommand\Weakm[2]{\mathsf w_{#1}^{(#2)}}
\newcommand\Weak[1]{\mathsf w_{#1}}
\newcommand\Weaknat{\mathsf w}
\newcommand\Cocontr[1]{{\overline{\mathsf{c}}}_{#1}}
\newcommand\Cocontrnat{{\overline{\mathsf{c}}}}
\newcommand\Cocontrm[2]{{\overline{\mathsf{c}}}_{#1}^{(#2)}}
\newcommand\Coweak[1]{{\overline{\mathsf{w}}}_{#1}}
\newcommand\Coweaknat{{\overline{\mathsf{w}}}}

\newcommand\Conv[1]{\operatorname{c}^{#1}}
\newcommand\IConv{\mathrel{*}}
\newcommand\Exclun[1]{\operatorname{u}^{#1}}
\newcommand\Derzero[1]{\partial^{#1}_0}
\newcommand\Dermorph[1]{\partial^{#1}}
\newcommand\Dirac[1]{\delta_{#1}}

\newcommand\Lintop[1]{\lambda({#1})}
\newcommand\Neigh[1]{\mathsf V(#1)}
\newcommand\Bnd[1]{\mathsf D(#1)}
\newcommand\Matrix[1]{\mathsf{M}({#1})}

\newcommand\Ev{\operatorname{ev}}
\newcommand\Evlin{\mathsf{ev}}
\newcommand\REL{\operatorname{\mathbf{Rel}}}
\newcommand\RELI{\REL^{\mathord\subseteq}}
\newcommand\POLR{\operatorname{\mathbf{Pol}}}
\newcommand\POLRI{\POLR^{\mathord\subseteq}}
\newcommand\POP{\operatorname{\mathbf{Pop}}}
\newcommand\POPI{\POP^{\mathord\subseteq}}
\newcommand\SCOTTLIN{\mathbf{ScottL}}
\newcommand\RELS{\mathsf R}
\newcommand\POLS{\mathsf S}
\newcommand\POPS{\mathsf P}

\newcommand\Diag[1]{\Delta_{#1}}
\newcommand\Codiag[1]{\overline\Delta_{#1}}
\newcommand\Final[1]{\mathsf t_{#1}}
\newcommand\Initial[1]{\overline{\mathsf t}_{#1}}

\newcommand\Norm[1]{\|{#1}\|}

\newcommand\Tpower[2]{{#1}^{\otimes{#2}}}

\newcommand\Termty[3]{{#1}\vdash{#2}:{#3}}
\newcommand\Polyty[3]{{#1}\vdash_!{#2}:{#3}}

\newcommand\Ruleskip{\quad\quad\quad\quad}

\newcommand\Sterms{\Delta}
\newcommand\Pterms{\Delta^!}
\newcommand\Nsterms{\Delta_0}

\newcommand\Relspan[1]{\overline{#1}}
\newcommand\Rel[1]{\mathrel{#1}}

\newcommand\Redone{\leadsto^1}
\newcommand\Redzone{\leadsto}
\newcommand\Red{\leadsto}
\newcommand\Redalt{\leadsto^*_0}
\newcommand\Redgen{\leadsto_{\mathsf{g}}}
\newcommand\Redpart{\leadsto_{\mathsf{w}}}

\newcommand\Multn{\mathrm{m}}
\newcommand\Shape[1]{\mathcal{T}(#1)}
\newcommand\Tay[1]{{#1}^*}

\newcommand\Deg[1]{\mathrm{deg}_{#1}}
\newcommand\Linearize[2]{\mathcal L^{#1}_{#2}}
\newcommand\Fmodrel{\Fmod}

\newcommand\Redeq{=_\Delta}
\newcommand\Kriv{\mathsf K}
\newcommand\Dom{\operatorname{\mathsf{D}}}
\newcommand\Domp[1]{\Dom(#1)}

\newcommand\Codom{\operatorname{\mathsf{Codom}}}
\newcommand\Cons{::}
\newcommand\Addtofun[3]{{#1}[{#2}\mapsto{#3}]}
\newcommand\Tofclos{\mathsf T}
\newcommand\Tofstate{\mathsf T}
\newcommand\BT{\operatorname{\mathsf{BT}}}
\newcommand\NF{\operatorname{\mathsf{NF}}}

\newcommand\Msubst[3]{\partial_{#3}(#1,#2)}
\newcommand\Symgrp[1]{\frak S_{#1}}

\newcommand\Tcoh{\mathrel{\frownsmile}}
\newcommand\Tcohs{\mathrel{\frown}}
\newcommand\Size[1]{|{#1}|}
\newcommand\Psize[1]{|{#1}|_{\mathrm P}}
\newcommand\Ssize[1]{|{#1}|_{\mathrm S}}

\newcommand\Symbofcell[1]{\mathsf l({#1})}
\newcommand\Portsofwire{\partial}
\newcommand\V{\vdots}

\newcommand\Oppofwire[1]{{#1}^*}

\newcommand\Exclo[1]{\Int{\Orth A}}
\newcommand\Typeq[1]{\cE_{#1}}

\newcommand\Orthat[1]{\bar{#1}}

\newcommand\Dedeq{\mathrel{\mathord{\vdash}_{\operatorname{eq}}}}

\newcommand\Typing[1]{\vdash{#1}\mid}

\newcommand\Eq{\simeq}

\newcommand\Grsw[2]{\cG({#1},{#2})}

\newcommand\Struct[1]{\mathsf{S}({#1})}

\newcommand\Meas[1]{\mathsf{M}({#1})}

\newcommand\In{\iota}
\newcommand\Out{o}
\newcommand\Diffn[3]{{#1}^*_{#2,#3}}

\newcommand\Exp[1]{\exp_{#1}}

\newcommand\Coname[1]{\overline{#1}}
\newcommand\Namesof[1]{\textsf N({#1})}
\newcommand\ProcEmpty{*}
\newcommand\ProcPar[2]{{#1}\mid{#2}}
\newcommand\ProcParI{\mid}
\newcommand\ProcNu[2]{\nu{#1}\cdot{#2}}
\newcommand\ProcIn[3]{{#1}({#2})\cdot{#3}}
\newcommand\ProcInLab[4]{{#1}^{#4}({#2})\cdot{#3}}
\newcommand\ProcOut[3]{\overline{#1}\langle{#2}\rangle\cdot{#3}}
\newcommand\ProcOutLab[4]{\overline{{#1}^{#4}}\langle{#2}\rangle\cdot{#3}}
\newcommand\SoloOut[2]{\overline{#1}\langle{#2}\rangle}
\newcommand\FreeNames{\operatorname{\mathsf{FV}}}
\newcommand\ProcCong{\sim}
\newcommand\ProcRed{\mathrel\leadsto}
\newcommand\StateRed{\mathrel\leadsto}
\newcommand\NetRed{\leadsto}
\newcommand\StateRedCan{\leadsto_{\mathsf{can}}}
\newcommand\StateCan{\operatorname{\mathsf{Can}}}
\newcommand\ProcTrad[2]{[{#1}]_{#2}}
\newcommand\ClosTrad[2]{[{#1}]_{#2}}
\newcommand\SoupTrad[2]{[{#1}]_{#2}}
\newcommand\StateTrad[3]{[{#1,#2}]_{#3}}
\newcommand\StateTradOpen[2]{[{#1}]^o_{#2}}
\newcommand\StateOfProc{\operatorname{\mathsf{St}}}
\newcommand\ProcOfState{\operatorname{\mathsf{Pr}}}

\newcommand\PortIn[1]{#1^+}
\newcommand\PortOut[1]{#1^-}
\newcommand\PortPair[1]{\PortIn{#1},\PortOut{#1}}
\newcommand\NetBroadcast[1]{\mathsf{Br}_{#1}}

\newcommand\Into[1]{\Int{\Orth{#1}}}

\newcommand\FreePortP[1]{{#1}^+}
\newcommand\FreePortN[1]{{#1}^-}

\newcommand\Soup{\operatorname{\mathsf{Soup}}}
\newcommand\Private{\operatorname{\mathsf{Priv}}}

\newcommand\Labels{\cL}
\newcommand\NameOfLabel{\operatorname{\mathsf{Name}}}
\newcommand\LabelsOfState{\operatorname{\Labels}}
\newcommand\LabelsOfProc{\operatorname{\Labels}}
\newcommand\StateLTS{\mathbb{S}_{\Labels}}
\newcommand\StateTrans[2]{\mathrel{\mathop\rightarrow^{#1/\overline{#2}}}}

\newcommand\LabelsOfNet{\operatorname{\Labels}}
\newcommand\NetLTS{\mathbb{D}_{\Labels}}
\newcommand\NetTrans[2]{\mathrel{\mathop\rightarrow^{#1/\overline{#2}}}}
\newcommand\NetCom{\operatorname{\mathsf{Com}}}

\newcommand\LTStrad{\Phi}
\newcommand\RedCom[2]{D_{#1,#2}}

\newcommand\PCNF{\Delta}
\newcommand\PCNFreducible[2]{\PCNF_{#1,#2}}
\newcommand\PCNFreduced[2]{\operatorname{\mathsf{red}}_{#1,#2}}

\newcommand\NetNorm[1]{\|{#1}\|}

\newcommand\Num{\operatorname{\mathsf{N}}}
\newcommand\StrNum[1]{\Num_{\mathrm{str}}({#1})}
\newcommand\NDNum[1]{\Num_{\mathrm{ND}}({#1})}
\newcommand\SNDNum[1]{\Num({#1})}

\newcommand\FormSem[2]{[{#1}]_{#2}}
\newcommand\FormPureSem[1]{[{#1}]}
\newcommand\ProofSem[2]{[{#1}]_{#2}}

\newcommand\MsetAdd{+}

\newcommand\NetOfProof[2]{{#1}^*_{#2}}
\newcommand\Experiments[2]{\operatorname{\mathsf{exper}}_{#2}({#1})}
\newcommand\ExpRes[1]{\operatorname{\mathsf{res}}({#1})}
\newcommand\NetSem[3]{[{#1}]_{#2}^{#3}}
\newcommand\ProcSem[2]{[{#1}]_{#2}}

\newcommand\Trees{\mathsf{D}}

\newcommand\NetCombine[2]{{#1}\cdot{#2}}
\newcommand\NetDual[2]{{#1}\perp{#2}}

\newcommand\NetRedEq{\mathrel{\mathord\leadsto^{\mathord=}}}
\newcommand\NetRedC{\mathord\leadsto_{\mathrm{c}}}
\newcommand\NetRedND{\mathord\leadsto_{\mathrm{nd}}}
\newcommand\NetRedS{\mathord\leadsto_{\mathrm{s}}}
\newcommand\NetRedSND{\mathord\leadsto_{\mathrm{snd}}}
\newcommand\NetRedSNDW{\mathord\leadsto_{\mathrm{sndw}}}
\newcommand\NetRedNSNDW{\mathord\leadsto_{\mathrm{nw}}}
\newcommand\NetRedSNDMax{\mathord\leadsto^{\mathrm{max}}_{\mathrm{snd}}}
\newcommand\NetRedStr{\mathord\leadsto^{\mathord+}}
\newcommand\NetRedN{\mathord\leadsto_{\mathrm{n}}}
\newcommand\NetRedCan{\mathord\leadsto_{\mathrm{Can}}}

\newcommand\NetRedM{\mathord\leadsto_{\mathrm{m}}}
\newcommand\NetRedA{\mathord\leadsto_{\mathrm{a}}}
\newcommand\NetRedBox{\mathord\leadsto_{\mathrm{bx}}}
\newcommand\NetRedBB{\mathord\leadsto_{\mathrm{bb}}}

\newcommand\NetEqAC{\mathrel{\mathord\sim_{\mathrm{ac}}}}

\newcommand\GenContrCell{?^*}
\newcommand\GenCocontrCell{!^*}

\newcommand\ComRedPart[3]{\operatorname{\textsf{Red}}_{#1,#2}({#3})}
\newcommand\ComReduced[3]{\operatorname{\textsf{red}}_{#1,#2}({#3})}

\newcommand\TransCl[1]{#1^*}
\newcommand\ReflCl[1]{#1^-}
\newcommand\TransClStr[1]{#1^+}

\newcommand\Atoms{\cA}
\newcommand\Atorth[1]{\overline{#1}}
\newcommand\Forall[2]{\forall{#1}\,{#2}}
\newcommand\Exists[2]{\exists{#1}\,{#2}}

\newcommand\Rulename[1]{\quad(${#1})$}
\newcommand\Proofseparation{\quad\quad\quad\quad}

\newcommand\Subproof[2]{%
\noLine%
\AxiomC{\raggedleft{\hbox{\begin{tabular}{cc}$\vdots$&{#1}\hspace{-2em}\ \end{tabular}}}}%
\UnaryInfC{#2}}

\newcommand\Psupp[2]{\textsf{supp}^{#2}({#1})}
\newcommand\Quant[1]{\mathcal{T}{#1}}
\newcommand\Symgroup[1]{\mathsf{S}_{#1}}
\newcommand\Symgroupf[1]{\mathsf{S}^{\mathrm{fin}}_{#1}}
\newcommand\Orbit[3]{(#1)^{#2}_{#3}}
\newcommand\Vvec[1]{\vec{#1}}

\newcommand\Satovt[1]{\widehat{#1}}

\newcommand\Trovt[1]{\Lambda({#1})}
\newcommand\Absso[1]{\Lambda({#1})}
\newcommand\Soproj[1]{\epsilon^{#1}}
\newcommand\Projso[1]{\epsilon^{#1}}
\newcommand\Coprojso[1]{{\bar\epsilon}^{#1}}

\newcommand\RELSO[1]{\mathbf{Rel}^{(#1)}}
\newcommand\RELACTION[1]{\mathbf{Rel}[{#1}]}
\newcommand\FINSO[1]{\FIN^{(n)}}

\newcommand\Substso[1]{{#1}^*}

\newcommand\Tupleid[1]{\vec{#1}}

\newcommand\Invar[2]{{#1}:{#2}}
\newcommand\Moved[2]{{#1}_{#2}}

\newcommand\Melange[1]{\eta^{#1}}
\newcommand\Act[2]{{#1}\cdot{#2}}
\newcommand\Forallf[1]{\forall\,{#1}}
\newcommand\Existsf[1]{\exists\,{#1}}
\newcommand\Tsem[2]{[{#1}]^{#2}}
\newcommand\Tsempol[2]{[{#1}]^{#2}_\POLS}
\newcommand\Tsempop[2]{[{#1}]^{#2}_\POPS}
\newcommand\Tsemrel[2]{[{#1}]^{#2}_\RELS}
\newcommand\Psem[2]{[{#1}]^{#2}}

\newcommand\Psemp[1]{[{#1}]}

\newcommand\Sirpring{\mathsf S}

\newcommand\Sig{\Sigma}
\newcommand\Arity{\mathop{\mathsf{ar}}}
\newcommand\Celltype{\mathop{\mathsf{symb}}}
\newcommand\Wireports{\partial}

\newcommand\Plusl{\mathop{\mathord\oplus\mathsf l}}
\newcommand\Plusr{\mathop{\mathord\oplus\mathsf r}}

\newcommand\Putfig[2]{%
\begin{minipage}[c]{#1\textwidth}\input{#2}\end{minipage}}

\newcommand\Putfigs[3]{%
\begin{minipage}[c]{#1\textwidth}\scalebox{#2}{\input{#3}}\end{minipage}}

\newcommand\Figspace[1]{\hspace{#1em}}

\newcommand\Emptynet{\epsilon}

\newcommand\ARG{s}

\newcommand\Mark{\vbox{\hbox{$\bullet$}\vspace{-4mm}\hbox{}}}

\newcommand\Forallcell[1]{\forall{#1}}
\newcommand\Existscell[2]{\exists{#1},{#2}}
\newcommand\Promcell[2]{{#1}^!,{#2}}
\newcommand\Wirerev[1]{{#1}^*}

\newcommand\Sexcl{\scriptstyle{\mathord!}}
\newcommand\Sint{\scriptstyle{\mathord?}}

\newcommand\Finsub{\sqsubseteq}

\newcommand\Field{\mathbf k}
\newcommand\Simple[1]{\hat{#1}}

\newcommand\Cut[2]{\langle{#1}\mid{#2}\rangle}

\newcommand\CONTR{\mathsf c}
\newcommand\DER{\mathsf d}
\newcommand\WEAK{\mathsf w}

\newcommand\Ttens[2]{\otimes(#1,#2)}
\newcommand\Tpar[2]{\IPar(#1,#2)}
\newcommand\Tcocontr[2]{\CONTR_\oc(#1,#2)}
\newcommand\Tcontr[2]{\CONTR_\wn(#1,#2)}
\newcommand\Tcoder[1]{\DER_\oc(#1)}
\newcommand\Tder[1]{\DER_\wn(#1)}
\newcommand\Tweak{\WEAK_\wn}
\newcommand\Tcoweak{\WEAK_\oc}
\newcommand\Tprom[1]{{#1}^\oc}
\newcommand\Tpromn[2]{{#1}^\oc_{#2}}
\newcommand\Netofseq[1]{{#1}^\bullet}

\newcommand\Net[2]{[#1,#2]}

\newcommand\Vecb[2]{\vec{#1}_{<#2}}
\newcommand\Veca[2]{\vec{#1}_{>#2}}

\newcommand\Ann{\mathop{\mathsf{ann}}}
\newcommand\Canbasis[1]{\mathsf e_{#1}}

\newcommand\Op[1]{{#1}^{\mathsf{op}}}

\newcommand\Derc[1]{\partial_{#1}}
\newcommand\Dercm[2]{\partial_{#1}^{#2}}
\newcommand\Coderc[1]{\overline\partial_{#1}}
\newcommand\Codercm[2]{\overline\partial_{#1}^{#2}}
\newcommand\Partexcl[2]{\oc_{#1}#2}

\newcommand\Shuffle{\mathsf{Shuffle}}
\newcommand\Funofmatm{\theta}

\newcommand\Primint[1]{I_{#1}}
\newcommand\Stermsh[2]{\Sterms_{#1}^{(#2)}}

\newcommand\Primmor[1]{I_{#1}}

\newcommand\Curlin[1]{\mathsf{cur}(#1)}

\newcommand\Taym[2]{\mathrm T^{#2}_{#1}}

\newcommand\Named[2]{#1\cdot#2}
\newcommand\Names{\cN}
\newcommand\Truenames{\cN}
\newcommand\Noname{\tau}
\newcommand\Interface[1]{\mathsf{Int}(#1)}
\newcommand\Freeinterface[1]{\mathsf{Int}^-(#1)}
\newcommand\Foliage[1]{\mathsf{Fol}(#1)}
\newcommand\Freefoliage[1]{\mathsf{Fol}^{-}(#1)}
\newcommand\Freecontext[1]{{#1}^-}
\newcommand\Freeset[1]{{#1}^-}
\newcommand\Cutnet[1]{{#1}_{\textsf c}}
\newcommand\Freenet[1]{{#1}_{\textsf f}}

\newcommand\Rolam{\mathsf{lam}}
\newcommand\Roapp{\mathsf{app}}

\newcommand\Tenspow[2]{{#1}^{\mathord\otimes #2}}
\newcommand\Complin{\,}
\newcommand\Compmat{\,}
\newcommand\Compl{\Complin}
\newcommand\Beta{\beta}
\newcommand\Dbeta{\delta}

\newcommand\Dapp[2]{\Diffsymb{#1}\cdot{#2}}
\newcommand\Dappm[3]{\Diffsymb^{#1}{#2}\cdot{(#3)}}
\newcommand\Dsubst[3]{\frac{\partial #1}{\partial #3}\cdot{#2}} 
\newcommand\Combapp[4]{\App{\Dappm{#1}{#2}{#3}}{#4}} 

\newcommand\FINV{\FIN(\Field)}
\newcommand\FINVK{\FIN_!(\Field)}
\newcommand\Fun[1]{\mathsf{Fun}(#1)}

\newcommand\Monoidal{\mu}
\newcommand\Monoidaln[1]{\mu^{(#1)}}
\newcommand\Seely{\mathsf m}

\newcommand\Polyhom[2]{\mathbf{Pol}_\Field(#1,#2)}
\newcommand\Anahom[2]{\widetilde{\mathbf{Pol}}_\Field(#1,#2)}

\newcommand\LL{\textsf{LL}}
\newcommand\DILL{\textsf{DiLL}}
\newcommand\MLL{\textsf{MLL}}
\newcommand\MALL{\textsf{MALL}}

\newcommand\Excllab{\scriptstyle\oc}
\newcommand\Intlab{\scriptscriptstyle\wn}
\newcommand\Parlab{\scriptscriptstyle\IPar}
\newcommand\Tenslab{\scriptscriptstyle\ITens}

\newcommand\Catnet[2]{[#1,#2]}

\newcommand\Qcoeff{\mathbb B}

\newcommand\Kl[1]{#1_{\oc}}

\newcommand\Sem[1]{[#1]}

\newcommand\CBVLAM{call-by-value lambda-calculus}
\newcommand\CBV{call-by-value}
\newcommand\Let[3]{\App{\Abs{#1}{#3}}{\Tofval{#2}}}
\newcommand\Letbis[3]{\textsf{let }#1=#2\textsf{ in }#3}

\newcommand\Valsymb{\mathsf V}

\newcommand\Redv{\beta_{\Valsymb}}
\newcommand\Redvz{{\hat\beta}_{\Valsymb}}
\newcommand\Redvtr{\beta^*_{\Valsymb}}
\newcommand\Redvztr{{\hat\beta}^*_{\Valsymb}}
\newcommand\Redrv{\delta}
\newcommand\Rednsrv{\tilde\delta}
\newcommand\Redrvtr{\beta^*_{\Valsymb}}

\newcommand\Parv{\rho_{\Valsymb}}
\newcommand\Parvtr{\rho^*_{\Valsymb}}
\newcommand\Lamval{\cV}
\newcommand\Lamvall{\cV_\lambda}

\newcommand\Supertens[2]{#1^{\mathord\otimes #2}}
\newcommand\Supertensup[2]{#1^{\mathord\otimes(#2)}}
\newcommand\Supertensp[2]{(#1)^{\mathord\otimes #2}}
\newcommand\Supertenspup[2]{(#1)^{\mathord\otimes(#2)}}

\newcommand\Tofval[1]{\left\langle#1\right\rangle}
\newcommand\Tofvalr[1]{\left\langle#1\right\rangle}
\newcommand\Tofvalrs[1]{\left\langle#1\right\rangle}
\newcommand\Tredv{\cN}
\newcommand\Simpl[2]{{#1}\multimap{#2}}

\newcommand\Softype[1]{#1^\bullet}

\newcommand\Cow{1}
\newcommand\Cbunch[2]{[#1]\cdot#2}
\newcommand\Cbunchp[2]{([#1]\cdot#2)}

\newcommand\Init[1]{\cI(#1)}
\newcommand\Downcl[1]{\mathord\downarrow#1}
\newcommand\Upcl[1]{\mathord\uparrow#1}
\newcommand\Downclr[2]{\mathord\downarrow_{#1}#2}
\newcommand\Upclr[2]{\mathord\uparrow_{#1}#2}

\newcommand\Substr{\sqsubseteq}

\newcommand\ScottU{U_{\mathsf S}}

\newcommand\Intint[1]{\underline{#1}}

\newcommand\Dsubstv[3]{\partial_{#3}(#1;#2)} 

\newcommand\Proofvskip{\\[2mm]}
\newcommand\Resofder[1]{{#1}^\circ}

\newcommand\Tlam{\Lambda_{\textsf t}}
\newcommand\Vlam{\Lambda_{\textsf v}}
\newcommand\Elam{\Lambda_{\textsf e}}
\newcommand\Qlam{\Lambda_{\textsf q}}

\newcommand\Canadd{\mathsf a}
\newcommand\Setofint[1]{\overline{#1}}
\newcommand\Permoflist[1]{\left\langle#1\right\rangle}

\newcommand\Termm{\mathsf t}
\newcommand\Family[3]{(#1)_{#2\in #3}}

\newcommand\Extiso{\mathsf e}

\newcommand\Poweb[1]{\langle#1\rangle}
\newcommand\Wpoweb[1]{\Web{\!\Poweb{#1}\!}}
\newcommand\Realize[1]{\Vdash_{#1}}

\newcommand\Relcbv{\mathcal U}
\newcommand\Polrcbv{\mathcal U_{\mathsf S}}

\newcommand\Relofpop{\rho}
\newcommand\Polofpop{\sigma}

\newcommand\Appro{\mathsf{app}}
\newcommand\Lamro{\mathsf{lam}}

\newcommand\Mix{\mathsf{mix}}
\newcommand\Mixn[1]{\Mix^{(#1)}}

\section{Categorical semantics of linear logic}\label{sec:thomas}
Before introducing the syntax of our resource \lam-calculus with tests, we
describe the general categorical structures needed to interpret this
calculus. Our goal here is to give general motivations for our syntactic
constructs. In the sequel, we consider a particular model, based on the category
of sets and relations, and it is not hard to check that this particular
category is an instance of the general setting we present here. In
Section~\ref{sec:rela-sem}, we shall present this relational
interpretation concretely in order to avoid the admittedly heavy categorical
formalism.

Our main reference for categorical models of linear logic (LL) is~\cite{Mellies09}.
We denote by $\nat$ the set of natural numbers.

Let $\scC$ be a Seely category. We recall briefly that such a structure consists
of a category $\scC$, whose morphisms should be thought of as linear maps,
equipped with a symmetric monoidal structure for which it is closed and
$*$-autonomous with respect to a dualizing object $\Bot$. The monoidal product, called
tensor product, is denoted as $\ITens$, the linear function space object from
$X$ to $Y$ is denoted as $\Limpl XY$, the composition of morphisms in $\scC$ is simply denoted 
as juxtaposition. 
We use $\Evlin\in\scC(\Tens{(\Limpl
  XY)}X,Y)$ for the linear evaluation morphism and $\Curlin f\in\scC(Z,\Limpl
XY)$ for the ``linear currying'' of a morphism $f\in\scC(\Tens ZX,Y)$. The
dual object $\Limpl X\Bot$ is denoted as $\Orth X$. 

We also assume that $\scC$ is cartesian, with a cartesian product denoted as
$\IWith$ and a terminal object $\Top$. By $*$-autonomy, this implies that $\scC$
is also cocartesian; we use $\IPlus$ for the coproduct and $\Zero$ for the
initial object. In any cartesian and cocartesian category, there is a canonical
morphism $\Canadd\in\scC(\Zero,\Top)$ and a canonical natural transformation
$\Canadd_{X,Y}\in\scC(\Plus XY,\With XY)$. One says that the category is
\emph{additive} if these morphisms are isomorphisms. In that case, each homset
$\scC(X,Y)$ is equipped with a structure of commutative monoid, and all
operations defined so far (composition, tensor product, linear currying)
are linear with respect to this structure.

If $\scC$ has cartesian products of all countable families $(X_i)_{i\in I}$ of
objects, we say that it is \emph{countably cartesian}, and in that case, $\scC$ is also
\emph{countably cocartesian}. If the canonical morphism $\Canadd_{\Family
 {X_i}iI}\in\scC(\bigoplus_{i\in I}X_i,\BWith_{i\in I}X_i)$ is an isomorphism, we
say that $\scC$ is \emph{countably additive}. In that case, homsets have countable
sums and composition as well as all monoidal operations commute with these sums.

Last, we assume that $\scC$ is equipped with an endofunctor $\Excl\_$ which has
a structure of comonad (unit $\Der X\in\scC(\Excl X,X)$ called
\emph{dereliction}, multiplication $\Digg X\in\scC(\Excl X,\Excl{\Excl X})$
called \emph{digging}). Moreover, this functor must be equipped with a monoidal
structure which turns it into a symmetric monoidal functor from the symmetric
monoidal category $(\scC,\IWith)$ to the symmetric monoidal category
$(\scC,\ITens)$: the corresponding isomorphisms $\Seely:\Excl\Top\to\One$ and
$\Seely_{X,Y}:\Excl{(\With XY)}\to\Tens{\Excl X}{\Excl Y}$ are often called
\emph{Seely isomorphisms}. The following diagram is moreover required to be
commutative.
\begin{equation*}
\xymatrix @R=1.5em @C=3pc
  {\Tens{\Excl X}{\Excl Y}
    \ar[r]^{\Seely_{X,Y}}
    \ar[dd]^{\Tens{\Digg X}{\Digg Y}}
    & \Excl{(\With XY)}\ar[d]^{\Digg{\With XY}}\\
    &\Excl{\Excl{(\With XY)}}
    \ar[d]^{\Excl{\Pair{\Excl{\Proj 1}}{\Excl{\Proj 2}}}}\\
    \Tens{\Excl{\Excl X}}{\Excl{\Excl Y}}\ar[r]^{\Seely_{\Excl X,\Excl Y}}
    &\Excl{(\With{\Excl X}{\Excl Y})}
  }
  \end{equation*}
  Using this monoidal structure, we can equip the $\oc$ functor with a lax
  symmetric monoidal structure from the symmetric monoidal category (SMC)
  $(\scC,\One,\ITens)$ to itself. In other words, one can define a morphism
  $\Monoidal:\One\to\Excl\One$ and a natural transformation
  $\Monoidal_{X,Y}:\Tens{\Excl X}{\Excl Y}\to\Excl{(\Tens XY)}$ which satisfy
  compatibility conditions with respect to the structure isomorphisms of the SMC
  $(\scC,\One,\ITens)$. Given an object $X$ of $\scC$ and $k\in\nat$, this allows
  to define a morphism $\Monoidaln k:\Supertensp{\Excl X}k\to\Excl{(\Supertens
    Xk)}$ which is essential in the interpretation of \lam-terms.

\subsection{Structural natural transformations}
Using these structures, we can define a \emph{weakening} natural transformation
$\Weak X\in\scC(\Excl X,\One)$ and a \emph{contraction} natural transformation
$\Contr X\in\scC(\Excl X,\Tens{\Excl X}{\Excl X})$ as follows. Since $\Top$ is
terminal, there is a canonical morphism $\Termm_X\in\scC(X,\Top)$ and we set
$\Weak X=\Seely\Complin\Excl{\Termm_X}$. Similarly, we have a diagonal
natural transformation $\Diag X\in\scC(X,\With XX)$ and we set $\Contr
X={\Seely_{X,X}}\Complin\Excl{\Diag X}$.

This contraction morphism $\Contr X:\Excl X\to\Tens{\Excl X}{\Excl X}$ is
associative, and therefore can be generalized to a unique morphism $\Contrm
Xn:\Excl X\to\Supertensp{\Excl X}{n}$. We have $\Contrm X0=\Weak X$, $\Contrm
X1=\Id{\Excl X}$ and $\Contrm X2=\Contr X$.

More generally we can define a morphism
$\Contrm X{k,n}:
{\Supertensp{\Excl X}k}\to\Supertensp{\Supertensp{\Excl X}k}{n}$ for
the generalized contraction morphism which is defined as the following
composition
\begin{equation*}
  \xymatrix @R=0.6em @C=4pc
  {\Supertensp{\Excl X}{k}\ar[r]^-{\Supertensp{\Contrm Xn}k}
    & \Supertensp{\Supertensp{\Excl X}{n}}k\ar[r]^-{\sigma}
    & \Supertensp{\Supertensp{\Excl X}{k}}n
  }
\end{equation*}
where 
$\sigma$ is the obvious isomorphism, defined using associativity and symmetry
of $\ITens$.

Similarly, we define a generalized weakening morphism $\Weakm Xk$ as the
composition
\begin{equation*}
  \xymatrix @R=0.8em @C=3pc
  {\Supertensp{\Excl X}{k}\ar[r]^-{\Supertensp{\Weak X}k}
    & \Supertensp{\One}k\ar[r]^-\lambda
    & \One
  }
\end{equation*}
where $\lambda$ is the unique canonical isomorphism induced by the monoidal
structure. 

As usual the (co)Kleisli category $\Kl\scC$ of the comonad $\Excl\_$ is defined as the category that has the same objects as $\scC$ and $\Kl\scC(X,Y) = \scC(\Excl X,Y)$, with composition denoted as $\circ$ and defined using the comonad. One can prove $\Kl\scC$ is cartesian closed, with $\IWith$ as cartesian product and $\Limpl{\Excl X}{Y}$ as function
 space object: this is a categorical version of Girard's translation of intuitionistic logic into linear logic.

Given $f\in\scC(\Supertensp{\Excl X}k,Y)$, it is standard to define $\Prom
f\in\scC(\Supertensp{\Excl X}k,\Excl Y)$, this operation is usually called
\emph{promotion} in linear logic. This morphism is defined as the following
composition:
\begin{equation*}
  \xymatrix @R=0.6em @C=3pc
  {\Supertensp{\Excl X}k\ar[r]^{\Supertensp{\Digg X}k}
    & \Supertensp{\Excl{\Excl X}}{k}\ar[r]^-{\Monoidaln k_X}
    & \Excl{(\Supertensp{\Excl X}{k})}\ar[r]^-{\Excl f}
    & \Excl Y
  }  
\end{equation*}



\subsection{Differential LL models}
The notion of categorical model recalled above allows to interpret standard
classical linear logic. If one wishes to interpret differential constructs as
well (in the spirit of the differential \lam-calculus or of differential
linear logic), more structure and hypotheses are required. Basically, we need that:
\begin{iteMize}{$\bullet$}
\item the cartesian and cocartesian category $\scC$ is additive, and
\item the model is equipped with a \emph{codereliction} natural
  transformation $\Coder X\in\scC(X,\Excl X)$ such that $\Der X\Complin\Coder
  X=\Myid_X$. 
\end{iteMize}
More conditions are required if one wants to interpret the full differential
\lam-calculus of~\cite{EhrhardR03} or full differential linear logic as presented
in e.g.~\cite{Pagani09}: these conditions represent a categorical axiomatization of
the usual chain rule of calculus and are well explained in~\cite{Fiore07}. When these conditions, that we give explicitely now, hold,
we say that \emph{the chain rule holds in $\scC$}.

The first condition is the following commutation.
\begin{equation*}\label{eq:der-coder-monoidal}
\vcenter{\hbox{\xymatrix @R=3em @C=3pc
{
  \Tens{X}{\Excl Y}\ar[r]^{\Tens{\Coder X}{\Excl Y}}\ar[d]^{\Tens{X}{\Der Y}}
  & \Tens{\Excl X}{\Excl Y}\ar[d]^{\Monoidal_{X,Y}}\\
  \Tens XY\ar[r]^{\Coder{\Tens XY}}
  & \Excl{(\Tens XY)}
}}}
\end{equation*}
It would be interesting to know if this condition can be reduced to a more
primitive one, involving $\Coder X$ and the isomorphism $\Seely$ (of course,
one can replace $\Monoidal$ by its expression in terms of $\Seely$ in the
diagram above, so that this diagram is actually a condition on
$\Seely$, but we would like to find a simpler and more elegant commuting
diagram involving $\Seely$).

Last we have to provide a commutation relating $\Coder X$ and $\Digg X$. We
have of course $\Coder{\Excl X}\Compl\Coder X:X\to\Excl{\Excl X}$. Also,
$\Monoidal_\One:\One\to\Excl\One$ and therefore $\Excl{\Coweak
  X}\Compl\Monoidal_\One:\One\to\Excl{\Excl X}$. Keeping implicit the
isomorphism $\Tens X\One\Isom X$, we get $\Tens{(\Coder{\Excl X}\Compl\Coder
  X)}{(\Excl{\Coweak X}\Compl\Monoidal_\One)}:X\to\Tens{\Excl{\Excl
    X}}{\Excl{\Excl X}}$, and we require the following diagram to commute:
\begin{equation*}\label{eq:coder-digg}
\vcenter{\hbox{\xymatrix @R=3em @C=3pc
{
  X\ar[r]^{\Coder X}\ar[d]_{\Tens{(\Coder{\Excl X}\Compl\Coder
    X)}{(\Excl{\Coweak X}\Compl\Monoidal_\One)}}
  & \Excl X\ar[d]^{\Digg X}\\
  \Tens{\Excl{\Excl X}}{\Excl{\Excl X}}\ar[r]^{\Cocontr{\Excl X}}
  & \Excl{\Excl X}
}}}
\end{equation*}

If $\scC$ is a weak differential LL model, we can define a coweakening morphism
$\Coweak X\in\scC(\One,\Excl X)$ and a cocontraction morphism $\Cocontr
X\in\scC(\Tens{\Excl X}{\Excl X},\Excl X)$ as we did for $\Weak X$ and $\Contr
X$. Similarly we also define $\Cocontrm Xn\in\scC(\Supertens{(\Excl X)}n,\Excl
X)$. Due to the naturality of $\Coder X$ we have $\Weak X\Compl\Coder X=0$ and
$\Contr X\Complin\Coder X=\Tens{\Coder X}{\Coweak X}+\Tens{\Coweak X}{\Coder
  X}$. We also define $\Derm Xn=\Supertens{\Der X}n\Complin\Contrm
Xn\in\scC(\Excl X,\Supertens Xn)$ and $\Coderm Xn=\Cocontrm
Xn\Complin\Supertens{\Coder X}n\in\scC(\Supertens Xn,\Excl X)$.

\subsection{The Taylor formula}
Let $\scC$ be a weak differential LL model which is countably additive. Remember
that each homset $\scC(X,Y)$ is endowed with a canonical structure of
commutative monoid in which countable families are summable. We assume moreover
that these monoids are idempotent. This means that, if $f\in\scC(X,Y)$, then
$f+f=f$. 

We say that the Taylor formula holds in $\scC$ if, for any morphism
$f\in\scC(X,Y)$, we have
\begin{equation*}
  \Excl f=\sum_{n=0}^\infty \Coderm Yn\Compl\Supertens fn\Compl\Derm Xn
\end{equation*}
\begin{rem}
  If the idempotency condition does not hold, one has to require the
  homsets to have a rig structure over the non-negative real numbers, and the
  Taylor condition must be written in the more familiar way $\Excl
  f=\sum_{n=0}^\infty \frac 1{\Factor n}\Coderm Yn\Compl\Supertens
  fn\Compl\Derm Xn$. To give a precise meaning to this kind of expressions, we
  need of course more structure on homsets: they need to have some completeness
  properties, typically expressible in topological terms. 
\end{rem}
\begin{rem}
  If the chain rule holds in $\scC$, the Taylor condition reduces to the
  particular case of identity morphisms: one has just to require that
  $\Excl{\Myid_X}=\sum_{n=0}^\infty\Coderm Xn\Derm Xn$ (in the idempotent
  setting).
\end{rem}

\subsection{Models of the pure differential \lam-calculus}
A model of the pure differential \lam-calculus of~\cite{EhrhardR03} or of the
$\dlam$-calculus to be introduced below, is simply a reflexive object in
$\Kl\scC$, where $\scC$ is a model of differential linear logic in which the
chain rule holds. More precisely, it consists of such a category $\scC$ and of a
triple $(U,\Roapp,\Rolam)$ where $U$ is an object of $\scC$ and
$\Roapp\in\scC(U,\Limpl{\Excl U}U)$ and $\Lamro\in\scC(\Limpl{\Excl U}U,U)$
satisfy $\Appro\Comp\Lamro=\Myid_{\Limpl{\Excl U}U}$ in $\scC$. It is crucial to
take $\Appro$ and $\Lamro$ in the ``linear'' category $\scC$ and not in $\Kl\scC$.

In the present paper, we concentrate on the case where $U$ satisfies a stronger
condition. We assume that $\scC$ is countably cartesian, and, given an object
$X$, we denote as $X^\nat$ the cartesian product $\BWith_{i\in\nat}X_i$ where
$X_i=X$ for each $i\in\nat$. We consider an object $U$ of $\scC$ together with
an isomorphism $\phi\in\scC(U,\Orth{(\Excl{U^\nat})})$. We have clearly
$\Orth{(\Excl{U^\nat})}\Isom\Orth{(\Excl{(U\IWith U^\nat)})}$, hence
$\Orth{(\Excl{U^\nat})}\Isom\Orth{(\Tens{\Excl U}{\Excl{U^\nat}})}$ by the
Seely isomorphism and therefore $\Orth{(\Excl{U^\nat})}\Isom\Limpl{\Excl
  U}{\Orth{(\Excl{U^\nat})}}$ by $*$-autonomy of $\scC$. Using $\phi$, we get
finally that $U\Isom\Limpl{\Excl U}{U}$ and we define $\Roapp$ and $\Rolam$
using this isomorphism.

We also assume that $\scC$ is a model of the MIX rule of linear logic
(see~\cite{Girard87}). This means that $\Bot$ is equipped with a structure of
commutative monoid in the SMC $\scC$. We use $\Mixn n$ to denote the
corresponding morphism $\Supertens\Bot n\to\Bot$ so that in particular $\Mixn
0:\One\to\Bot$ and $\Mixn 1=\Id\Bot$.

The interpretation of the calculi presented in this paper is based on the
following toolbox. The first constructions we give deal with ``terms'', which
are represented here by morphisms $\Supertens{(\Excl U)}k\to U$ (the number
$k\in\nat$ corresponds intuitively to the number of variables on which the term
depends).
\begin{iteMize}{$\bullet$}
\item Given a family of terms $\List f1n:\Supertens{(\Excl U)}k\to U$, we can
  define a morphism $\Mset{\List f1n}:\Supertens{(\Excl U)}k\to\Excl U$ as
  $\Mset{\List f1n}=\Coderm Un\Comp(f_1\ITens\cdots\ITens f_n)\Comp\Contrm
  U{k,n}$ (a morphism of this type will be called a ``bag'').
\item Let $f:\Supertens{(\Excl U)}k\to U$ be a further term. Remember that
  we have defined the promotion of $f$, which is a bag $\Prom
  f:\Supertens{(\Excl U)}k\to\Excl U$. Therefore we can define
  $\Tensp{\Mset{\List f1n}}{\Prom f}\Comp\Contrm U{k,2}:\Supertens{(\Excl
    U)}k\to\Tens{\Excl U}{\Excl U}$ and we introduce a new bag
  \begin{equation*}
    \Mset{\List f1n,\Prom f}=
    \Cocontr U\Comp\Tensp{\Mset{\List f1n}}{\Prom f}\Comp\Contrm
    U{k,2}:\Supertens{(\Excl U)}k\to\Excl U\,.
  \end{equation*}
\item Let $f:\Supertens{(\Excl U)}k\to U$ and $g:\Supertens{(\Excl U)}k\to
  \Excl U$. Then $\Roapp\Comp f:\Supertens{(\Excl U)}k\to(\Limpl{\Excl U}U)$
  and hence $\Tensp{(\Roapp\Comp f)}g\Comp\Contrm U{k,2}:\Supertens{(\Excl
    U)}k\to\Tens{(\Limpl{\Excl U}U)}{U}$. Finally we define the application of
  $f$ to $g$ as the term $f\,g=\Evlin\Comp\Tensp{(\Roapp\Comp
    f)}g\Comp\Contrm U{k,2}:\Supertens{(\Excl U)}k\to U$.
\item Let $f:\Tens{\Supertens{(\Excl U)}k}{\Excl U}\to U$, we define the
  abstraction of $f$ (with respect to its last parameter) as the term
  $\lambda(f)=\Rolam\Comp\Curlin f:\Supertens{(\Excl U)}k\to U$.
\end{iteMize}

A ``test'' (with $k$ variables) is a morphism $\Supertens{(\Excl
  U)}k\to\Bot$. Here we present the categorical constructions required for dealing
with such tests.
\begin{iteMize}{$\bullet$}
\item Let $\List h1n:\Supertens{(\Excl U)}k\to\Bot$ be tests. Then we can
  define their parallel composition, using the mix structure of $\Bot$, as the
  test
  \begin{equation*}
    (h_1\mid\cdots\mid h_n)=\Mixn n\Comp(h_1\ITens\cdots\ITens h_n)\Comp\Contrm
    U{k,n}:\Supertens{(\Excl U)}k\to\Bot\,.
  \end{equation*}
\item Let $f:\Supertens{(\Excl U)}k\to U$ be a term. We have $\phi\Comp
  f:\Supertens{(\Excl U)}k\to\Orth{(\Excl{U^\nat})}$ and hence we can define a
  test $\tau(f)=\Orth{\Coweak{U^\nat}}\Comp\phi\Comp f:\Supertens{(\Excl
    U)}k\to\Bot$ since $\Coweak{U^\nat}:\One\to\Excl{U^\nat}$ and
  $\Orth\One=\Bot.$
\item Last, let $h:\Supertens{(\Excl U)}k\to\Bot$ be a test. We have
  $\Orth{\Weak{U^\nat}}\Comp h:\Supertens{(\Excl U)}k\to\Orth{(\Excl{U^\nat})}$
  since $\Weak{U^\nat}:\Excl{U^\nat}\to\One$, and hence
  $\bar\tau(h)=\Funinv\phi\Comp\Orth{\Weak{U^\nat}}\Comp h:\Supertens{(\Excl
    U)}k\to U$ is a term.
\end{iteMize}

\section{The \dzlam-Calculus with Tests}\label{sec:dzlam-calculus+tests}

The definition of the \dzlam-calculus with tests requires some preliminary notations that we give below.

\subsection{\bf Sets and modules.}\label{subs:sets} 
We denote by $\nat$ the set of natural numbers and by $\one$ an arbitrary singleton set. Given a set $S$, we write $\Pow{S}$ (resp.\ $\Powf{S}$) for the set of all
 (resp.\ all finite) subsets of $S$. Given $k\in\nat$, we denote by $\perm{k}$ the \emph{set of all permutations of $\{1,\ldots,k\}$}.

Let $\bool$ be the semiring $\{0,1\}$ with $1+1=1$ and multiplication defined in the obvious way. For any set $S$, we write $\bool\langle S \rangle$ for the free
 $\bool$-module generated by $S$, so that $\bool\langle S \rangle\cong\Powf{S}$ with addition corresponding to union, and scalar multiplication defined in the
 obvious way. However we prefer to keep the algebraic notations for elements of $\bool\langle S \rangle$, hence set unions will be denoted by $+$ and the empty set
 by $0$.

\subsection{\bf Multisets.}\label{subs:msets} 
Let $S$ be a set. A \emph{multiset} $a$ over $S$ can be defined as an unordered list $a = [\ga_1,\ga_2,\ldots]$ with repetitions such that $\ga_i\in S$ for all
 indices $i$. 
 A multiset $a$ is called \emph{finite} if it is a finite list; we denote by $\card a$ its cardinality. We write $\Mfin{S}$ for the set of all finite
 multisets over $S$. Given two multisets $a = [\ga_1,\ga_2,\ldots]$ and $b = [\gb_1,\gb_2,\ldots]$ the \emph{multiset union} of $a,b$ is defined by
 $a\mcup b = [\ga_1,\gb_1,\ga_2,\gb_2,\ldots]$; summing up, $\nat\langle S \rangle\cong\Mfin{S}$. Given two finite sequences of multisets $\seq a,\seq b$ of the same length $n$ we define
 $\seq a\mcup\seq b = (a_1\mcup b_1,\ldots,a_n\mcup b_n)$. Given a strict order $>$ on $S$, the \emph{multiset ordering} \cite[Def.~A.6.2]{terese03}
 is the smallest transitive relation $\mge$ on $\Mfin{S}$ such that $(\forall \gb \in b.\ \ga > \gb) \Rightarrow (a \mcup [\ga] \mge a \mcup b)$, for all
 $\ga \in S$ and all $a,b \in \Mfin{S}$.
Intuitively, $a \mge b$ holds if $b$ can be obtained from $a$ by replacing some of its elements by finitely many (possibly zero) smaller elements.

We now introduce the \emph{\dzlam-calculus with tests} which is the promotion-free fragment of the \dlam-calculus with tests we will present in
 Section~\ref{sec:FullRC}.

\subsection{Syntax.} The \dzlam-calculus with tests has four syntactic categories: 
\emph{terms} that are in functional position,
\emph{bags} that are in argument position and represent multisets of linear resources,
\emph{tests} that are ``corked'' multisets  of terms having only two possible outcomes
and \emph{finite formal sums} representing all possible results of a computation.
{\em Expressions} are either terms, bags or tests and will be used to state results holding for all categories.

\begin{figure}[t]
\centering
\textbf{The  \dzlam-calculus with tests}\\[2ex]
\subfigure[Grammar of terms, bags, tests, expressions, sums.]{\label{fig:grammar}
    \begin{tabular}{@{}p{.1\linewidth}@{}p{.125\linewidth}@{}p{.525\linewidth}@{}p{.15\linewidth}@{}}
    $\Set{\gto}$:       &$M,N,L,H$         &$\gramm x \mid \lambda x.M \mid MP \mid \gto(V)$        &\hfill terms\\[3pt]
    $\Set{b}$:       &$P,Q$         &$\gramm [L_1,\ldots,L_k]$        &\hfill bags\\[3pt]    
    $\Set{\gt}$:       &$V,W$         &$\gramm \gt[L_1,\ldots,L_k]$        &\hfill tests\\[3pt]        
    $\Set{e}$:       &$A,B$         &$\gramm M\mid V\mid P$        &\hfill expressions\\[3pt]     
    \\   
    \hline
    \\
    \multicolumn{3}{p{.80\linewidth}}{
    \hspace{-6pt}$ \sM, \sN,\sL,\sH\in \sums{\gto}$} &sums of terms\\[3pt]
    \multicolumn{3}{p{.80\linewidth}}{
    \hspace{-6pt}$ \sP,\sQ\in \sums{b}$} &\hfill sums of bags\\[3pt]    
    \multicolumn{3}{p{.80\linewidth}}{
    \hspace{-6pt}$ \sV,\sW\in\sums{\gt}$} &\hfill sums of tests\\[3pt]        
    \multicolumn{3}{p{.80\linewidth}}{
    \hspace{-6pt}$ \sA,\sB\in \sums{e} :=\sums{\gt}\cup\sums{\gto}\cup\sums{b}$} &\hfill sums of expr.\\[3pt]        
    \end{tabular}
}
\\[3pt]
\hrulefill\\[4pt]
\subfigure[Notation on parallel composition of tests.]{
\label{fig:notpar}
\begin{minipage}{\linewidth}
$$
\parallel_{i=1}^n V_i \ass V_1\paral\cdots\paral V_n,\qquad \textrm{(in particular }\parallel_{i=1}^0 V_i \ass \varepsilon)
$$
\vspace{-3pt}
\end{minipage}
}
\\[3pt]
\caption{\footnotesize Syntax and notations of  \dzlam-calculus with tests.}
\label{fig:statics}
\end{figure}
\begin{defi} The formal grammars defining \emph{terms, bags, tests} and \emph{sums} are given in Figure~\ref{fig:grammar}.
\end{defi}

\emph{Terms} are the real protagonists of the \dzlam-calculus with tests. The term $\lam x.M$ represents the \emph{\lam-abstraction} of the variable $x$ in the
 term $M$ and $MP$ the \emph{application} of a term $M$ to a bag $P$ of linear resources. Thus, in $(\lam x.M)P$, each resource in $P$ is available exactly once for
 $\lam x.M$ and if the number of occurrences of $x$ in $M$ ``disagrees'' with the cardinality of $P$ then the result is 0 (see later, when sums of expressions are
 introduced). The operator $\gto(\cdot)$  will be discussed later on, after the notion of test is explained.

As usual we assume that application associates to the left and lambda abstraction to the right.
Therefore we will write $\lam x_1\dots x_n.MP_1\cdots P_k$ for $\lam x_1.(\cdots(\lam x_n.(\cdots(MP_1)\cdots P_k))\cdots)$.
Moreover, the notation $MP^{\sim n}$ will stand for $MP\cdots P$ ($n$ times).

\begin{nota}
Concerning specific terms, we set 
$$
	\bold{I}\ass \lam x.x,\quad \bold{T}\ass \lam xy.x,\quad \bold{F} \ass \lam xy.y,\quad \bold{D}:=\lam x.x[x],
$$ 
$$
	\Xi_{n_1,\ldots,n_m} \ass \lam x_1\dots x_m.\bold{I}[x_1]^{\sim n_1}\cdots [x_m]^{\sim n_m},\textrm{ for all }n_1,\dots,n_m\in\nat,
$$
where the symbol `$\ass$' denotes definitional equality.
\end{nota}
 
{\em Tests} are multisets of terms, the ``$\gt$'' being a tag for distinguishing them from bags. 
Intuitively, they are expressions that can produce two results: either {\em success}, represented by $\varepsilon$, or {\em failure}, represented by $0$.

Throughout the paper, we will enforce the distinction between bags and tests by using systematically the following notational conventions.

\begin{nota}
For bags, we use the usual multiset notation: 
\begin{iteMize}{$\bullet$}
\item $[]$ is the empty bag,
\item $P\mcup P'$ is the union of bags.
\end{iteMize}
For tests we write: 
\begin{iteMize}{$\bullet$}
\item $\varepsilon$ for the empty multiset, 
\item $V \paral W$ for the multiset union of $V$ and $W$. 
\end{iteMize}
In other words, $\varepsilon=\gt[]$ and $\gt [L_1,\ldots,L_k]\  \paral\  \gt [L_{k+1},\ldots,L_n]=\gt [L_1,\ldots,L_n]$.
Other notations on parallel composition of tests are introduced in Figure~\ref{fig:notpar}.
\end{nota}

The test $V\paral W$ represents the (must-)parallel composition of $V$ and $W$ (i.e., $V\paral W$ succeeds if both $V$ and $W$ succeed and the order of evaluation is
 inessential). We prefer to use the parallel notation as syntactic sugar in order to avoid both the explicit treatment of associativity and commutativity axioms
 (plus neutrality of $\varepsilon$). This is perfectly coherent with the implementation of tests as multisets of terms.

\label{pageref:gto} The operator $\gto(\cdot)$ allows to build a term out of a test:
intuitively, the term $\gto(V)$ may be thought of as $V$ preceded by 
an infinite sequence of dummy $\lambda$-abstractions. 
Dually, the ``cork construction'' $\gt[L_1,\ldots,L_k]$ may be thought of as 
an operator applying to all its arguments an infinite sequence of empty bags.
This suggests in particular that it is sound to reduce $\gt[\gto(V)]$
to $V$.

Hence the term $\gto(V)$ raises an exception 
encapsulating $V$ and the test $\gt[L_1,\ldots,L_k]$
catches the exception possibly raised by, say, $L_i$ and replaces $L_i$
by the multiset of terms encapsulated in that exception.
The context of the exception is thrown away by the 
dummy abstractions of $\gto$ and the dummy applications of  $\gt$.
A test needs to catch an exception in order to succeed; for instance,
$\gt[M]$ fails as soon as $M$ is a $\gto$-free, closed term.

{\em Sums.} 
Remember from Subsection~\ref{subs:sets} that $\sums{\gto}$ (resp.\ $\sums{\gt}$, $\sums{b}$) 
denotes the set of finite formal sums of terms (resp.\ tests, bags)
with an idempotent sum.
We also set $\sums{e} \ass \sums{\gt}\cup\sums{\gto}\cup\sums{b}$.
This is an abuse of notation as $\sums{e}$ here does not denote the $\bool$-module generated over 
$\Set{\gt}\cup\Set{\gto}\cup\Set{b}$, but rather the union of the three  $\bool$-modules;
this means that sums should be taken only in the same sort.
The typical metavariables to denote sums are given in Figure~\ref{fig:grammar}.

The \emph{$\alpha$-equivalence} relation and the set $\FV(\sA)$ of \emph{free variables of $\sA$} are defined 
as usual, like in the ordinary \lam-calculus~\cite{Bare}.
Hereafter, (sums of) expressions are considered up to $\alpha$-equivalence.

Because of the absence of promotion the number of linear resources that a term $\lam x.M$ is expecting is just the number of occurrences of $x$ in $M$
 (the degree of $x$ in $M$).

\begin{defi}\label{def:degree} Let $A\in\Set{e}$. The \emph{degree of $x$ in $A$}, written $\dg{x}{A}$, is the number of free occurrences of $x$ in $A$ 
and is defined by induction as follows:
\begin{iteMize}{$\bullet$}
\item $\dg{x}{x} = 1$,
\item $\dg{x}{y} = 0$, for $x\neq y$,
\item $\dg{x}{\lam y.M} = \dg{x}{M}$, where we assume wlog $x\neq y$,
\item $\dg{x}{MP} = \dg{x}{M} + \dg{x}{P}$,
\item $\dg{x}{\gto(V)} = \dg{x}{V}$,
\item $\dg{x}{\bag{L_1, \ldots, L_k}} = \sum_{i=1}^{k} \dg{x}{L_i}$,
\item $\dg{x}{\gt[L_1, \ldots, L_k]} = \sum_{i=1}^{k} \dg{x}{L_i}$.
\end{iteMize}
\end{defi}

\subsection{Two Kinds of Substitutions}\label{subs:Subst}

In this subsection we introduce two kinds of substitutions: the usual \lam-calculus substitution and a linear one, 
which is proper to  differential and resource calculi (see\ 
\cite{Boudol93,EhrhardR03,Tranquilli10}).

In order to proceed, we first need to introduce some notational conventions concerning the sums.
Indeed the grammar for terms and tests does not include any sums, 
so they may arise only on the ``surface''. 
For instance, $\bold{I} + \bold{I}$ is a legal sum of expressions, 
while $\lam x.(x+x)$ cannot be generated using the grammar of Figure~\ref{fig:grammar}.
\begin{figure}[t]
\centering
\textbf{Extended Syntax}\\[2ex]
\subfigure[Notation on $\sums{e}$.]{
\label{fig:notsums}
\begin{minipage}{\linewidth}
\begin{gather*}
	\textstyle \lam x.(\sum_i M_i) = \sum_i \lam x . M_i\qquad
	\textstyle \sM(\sum_i P_i) = \sum_i \sM P_i\qquad
	\textstyle (\sum_i M_i)\sP = \sum_i M_i \sP\\[2pt]
	\textstyle\gto(\sum_i V_i) = \sum_i \gto(V_i)\qquad
	\textstyle(\sum_i P_i)\mcup \sP = \sum_i P_i\mcup \sP\qquad
	\textstyle[\sum_i L_i]= \sum_i [L_i]\\[2pt]
	\textstyle\gt[\sum_i M_i] = \sum_i \gt[M_i]\qquad
	\textstyle(\sum_i V_i)\ \paral\ \sV = \sum_i\  V_i\ \paral\ \sV\\[-6pt]
\end{gather*}
\end{minipage}
}
\\[3pt]
\hrulefill\\[4pt]
\centering
\textbf{Linear Substitution}\\[2ex]
\subfigure[Definition of linear substitution. In the abstraction case we assume wlog $x\neq y$.]{\label{fig:linsubst}
\begin{minipage}{\linewidth}
$$
\begin{array}{rcl}
y\lsubst{x}{N} &=&
\begin{cases}
N  & \textrm{if $y = x$,} \\
0  & \textrm{otherwise,} \\
\end{cases}\\[2ex]
\bag{L_1,\ldots,L_k}\lsubst{x}{N} &=&\Sigma_{i = 1}^k [L_1,\ldots,L_i\lsubst{x}{N}\ldots,L_k],\\[1ex]
\tau[L_1,\ldots,L_k]\lsubst{x}{N} &=&\Sigma_{i = 1}^k \tau[L_1,\ldots,L_i\lsubst{x}{N},\ldots,L_k],\\[1ex]
(MP)\lsubst{x}{N} &=& M\lsubst{x}{N}P + M(P\lsubst{x}{N}),\\[1ex]
\gto(V)\lsubst{x}{N} &=& \gto(V\lsubst{x}{N}),\\[1ex]
(\lam y.M)\lsubst{x}{N} &=& \lam y.M\lsubst{x}{N}.\\[1ex]
\end{array}
$$
\end{minipage}
}
\caption{\footnotesize Notations on sums and definition of linear substitution.}
\label{fig:statics2}
\end{figure}

\begin{conv} As a syntactic sugar -- and \emph{not} as actual syntax --
we extend all the constructors to sums by multilinearity, setting for instance
$$
	(\textstyle\sum_i M_i)(\textstyle\sum_j P_j) := \textstyle\sum_{i,j} M_iP_j,
$$ 
in such a way that the equations in Figure~\ref{fig:notsums} hold.
\end{conv}

This kind of meta-syntactic notation is discussed thoroughly in \cite{EhrhardR08}. 

\begin{rem}
In the  particular case of empty sums, we get
 $\lam x.0 \ass 0$, $M0 \ass 0$, $0P \ass 0$, $\gt[0] \ass 0$, $\gto(0) \ass 0$, $V\paral 0 \ass 0$, $[0] \ass 0$ and $0\mcup P \ass 0$.
 Therefore $0$ annihilates any term, bag or test (but not the sums).
\end{rem}

We now give some examples of this extended (meta-)syntax.

\begin{exa} We have:
\begin{enumerate}[1.] 
\item $\lam xy.(x + y)\ass \lam xy.x + \lam xy.y\ass \bold{T} + \bold{F}$,
\item $\lam x.(x + x)\ass \lam x.x + \lam x.x\ass \bold{I}+ \bold{I} = \bold{I}$ by sum idempotency,
\item $(x_1+x_2)[y] \ass x_1[y] + x_2[y]$,
\item $x[y_1 + y_2] \ass x([y_1] + [y_2])\ass x[y_1] + x[y_2]$, therefore:
\item $(x_1+x_2)[y_1 + y_2]\ass x_1[y_1] + x_1 [y_2] + x_2[y_1] + x_2[y_2]$.
\end{enumerate}
\end{exa}

In the following two definitions we make an essential use of the extended syntax.
We recall that an operator $F(-)$ is \emph{extended by linearity} by setting $F(\Sigma_{i} x_i) = \Sigma_{i} F(x_i)$.

\begin{defi}[Substitution] Let $A\in\Set{e}$ and $N\in\Set{\gto}$.
The \emph{(capture-free) substitution of $N$ for $x$ in $A$}, denoted by $A\subst{x}{N}$, is defined as usual.
Accordingly,  $A\subst{x}{\sN}$ denotes an expression of the extended  syntax.
Finally, we extend this operation to sums as in $\sA\subst{x}{\sN}$ by linearity in $\sA$.
\end{defi}

\begin{defi}[Linear Substitution] The \emph{linear (capture-free) substitution of $N$ for $x$ in $A$}, denoted by $A\lsubst{x}{N}$, is defined 
in Figure~\ref{fig:linsubst}.
The expression $A\lsubst{x}{\sN}$ belongs to the extended syntax. We extend this operation 
to sums as in $\sA\lsubst{x}{\sN}$ by linearity in $\sA$, as we did for usual substitution.
\end{defi}

Roughly speaking, the linear substitution $A\lsubst{x}{N}$ replaces 
\emph{exactly one} free occurrence of $x$ in $A$ with the term $N$.
If there is no occurrence of $x$ in $A$ then the result is 0.
In presence of multiple occurrences, all possible choices are made and the result is the sum of terms corresponding to them.

\begin{rem}
Observe that  $\sA\lsubst{x}{\sN}$ is linear in $\sA$ and in $\sN$, 
whereas  $\sA\subst{x}{\sN}$ is linear in $\sA$ but not in  $\sN$.
\end{rem}

We now give some examples of linear and classic substitution.

\begin{exa} Let $A\in\Set{e}$ and $M,N\in\Set{\gto}$.
\begin{enumerate}[1.]
\item If $A$ is closed, then $A\lsubst{x}{M} =0$,
\item $(\lam y.y[y][x])\lsubst{x}{\bold{I}} = \lam y.y[y][\bold{I}]$,
\item $(\lam y.y[x][x])\lsubst{x}{\bold{I}} = \lam y.y[\bold{I}][x] + \lam y.y[x][\bold{I}]$,
\item $(\lam y.y[x][x])\subst{x}{\bold{I}} = \lam y.y[\bold{I}][\bold{I}]$,
\item $(x[x])\subst{x}{(y + z)} =  y[y] + y[z] + z[y] + z[z]$.
\end{enumerate}
\end{exa}

Linear substitutions commute in the sense expressed by the next theorem, 
whose proof is rather classic and thus omitted.

\begin{thm}[Schwarz's Theorem, cf.\ \cite{EhrhardR03}]\label{thm:Schwarz}
For $\sA\in\sums{e}$, $\sM,\sN\in\sums{\gto}$ and $y\notin\FV(\sM)\cup\FV(\sN)$ we have:
$$
	\sA\lsubst{y}{\sM}\lsubst{x}{\sN} = \sA\lsubst{x}{\sN}\lsubst{y}{\sM} + \sA\lsubst{y}{\sM\lsubst{x}{\sN}}.
$$
In particular, if $x\notin\FV(\sM)$ the two substitutions commute.\qed
\end{thm}

\begin{nota}\
\begin{iteMize}{$\bullet$}
\item Given a bag $P=[L_1,\ldots,L_k]$ and $x\notin\FV(P)$ we set $\sA\lsubst{x}{P}\ass \sA\lsubst{x}{L_1}\cdots\lsubst{x}{L_k}$.

\item Given bags $P_1,\ldots,P_n$ and $\seq x\notin\FV(\seq P)$ we set $\sA\lsubst{\seq x}{\seq P}\ass \sA\lsubst{x_1}{P_1}\cdots\lsubst{x_n}{P_n}$.
\end{iteMize}
In particular, $\sA\lsubst{x}{[]} = \sA$.
\end{nota}

The above notation $\sA\lsubst{x}{P}$ makes sense because, by Theorem~\ref{thm:Schwarz}, the expression $\sA\lsubst{x}{L_1}\cdots\lsubst{x}{L_k}$ 
 is actually independent from the enumeration of $L_1,\ldots,L_k$ in $P$. Moreover recall that we use $\ga$-equivalence, so that bound variables can 
 be renamed in order to avoid capture of free variables during substitution.

\subsection{The Operational Semantics}\label{subs:OpSem}

In this section we are going to introduce the reduction rules defining the operational semantics of the \dzlam-calculus with tests.

\begin{defi}\label{def:reductionrules} The \emph{reduction semantics} of the \dzlam-calculus with tests is generated by the 
rules in Figure~\ref{fig:dzlamopsem}. \begin{figure}[t]
\centering
\textbf{Reduction Semantics}\\[2ex]
\subfigure[Reduction rules.
In the $(\beta)$ rule we assume wlog $x\not\in\FV(P)$.]{
\label{fig:dzlamopsem}
\begin{minipage}{\linewidth}
\begin{center}
\begin{tabular}{lcl}
$(\lam x.M)P\to_{\beta} M\lsubst{x}{P}\subst{x}{0}$,&\qquad\qquad&
$\gto(V)P\to_{\gto} 
\begin{cases}
\gto(V)  & \textrm{if $P= []$,} \\
0  & \textrm{otherwise,} \\
\end{cases}$\\
\\
$\gt[\lam x.M]\paral V \to_{\gt} \gt[M\subst{x}{0}]\paral V$,&&
$\gt[\gto(V)]\paral W  \to_{\gamma} V\paral W $.\\~\\
\end{tabular}
\end{center}
\end{minipage}
}
\\[3pt]
\hrulefill\\[4pt]
\centering
\textbf{Context Closure}\\[2ex]
\subfigure[Context closure of a relation $\texttt{R}\subseteq \Set{e}\times\sums{e}$.]{\label{fig:context_closure}
\begin{minipage}{\linewidth}
\begin{gather*}
\infer[\mathrm{lam}]{\lam x. M \rel{R} \lam x.\sM}{M\rel{R} \sM}
\qquad\qquad\qquad
\infer[\mathrm{appl}]{M\sP\rel{R}  \sM \sP}{M \rel{R}  \sM}
\qquad\qquad\qquad
\infer[\mathrm{appr}]{\sM P\rel{R} \sM \sP}{P\rel{R} \sP}\\
\infer[\mathrm{linres}]{\bag{M} \mcup \sP \rel{R} \bag{\sM}\mcup \sP}{M\rel{R} \sM}
\qquad\qquad
\infer[\mathrm{taubar}]{\gto(V) \rel{R} \gto(\sW)}{V\rel{R} \sW}
\qquad\qquad
\infer[\mathrm{tau}]{\gt[M]\paral V \rel{R} \gt[\sM]\paral V}{M\rel{R} \sM}\\
\infer[\mathrm{sum}]{A + \sB \rel{R} \sA + \sB}{A \rel{R} \sA}\\[-1ex]
\end{gather*}
\end{minipage}
}
\caption{\footnotesize Operational Semantics.}
\end{figure}
\end{defi}

The reduction preserves the sort of an expression in the sense that terms rewrite to (sums of) terms and tests to (sums of) tests.

The left side of a reduction rule in Figure~\ref{fig:dzlamopsem} is called a \emph{redex} while the right side is its \emph{contractum}.
Redexes are classified, depending on their kind, as follows.

\begin{defi}\label{def:redex}\hfill
\begin{iteMize}{$\bullet$}
\item
	A \emph{term-redex} is any term of the form $(\lam x.M)P$ or $\gto(V)P$.
\item
	A \emph{test-redex} is any test of the form $\gt[\lam x.M] \paral V$ or $\gt[\gto(V)]\paral W$.
\end{iteMize}
\end{defi}

The following remark gives a more explicit characterization of a $\beta$-contractum. 
Remember that the degree of $x$ in $M$ has been defined in Definition~\ref{def:degree}.

\begin{rem}\label{rem:beta-contr}
If $M$ has $k$ free occurrences of $x$ (represented by $x^1,\dots,x^k$) then we have 
$$ M\lsubst{x}{L_1}\cdots\lsubst{x}{L_k}\subst{x}{0} = \Sigma_{\sigma \in \perm{k}} M\{L_{\sigma(1)}/x^1,\ldots,L_{\sigma(k)}/x^k\}; $$
if $\dg{x}{M} \neq k$, then $M\lsubst{x}{L_1}\cdots\lsubst{x}{L_k}\subst{x}{0} = 0$.
\end{rem}

From Remark~\ref{rem:beta-contr} it is clear that, because of the presence of linear substitution, 
the $\beta$-reduction is a relation from terms to sums of terms, namely $\to_\beta\ \subseteq\ \Set{\gto}\times\sums{\gto}$.

\begin{defi} \
\begin{enumerate}[1.]
\item The \emph{contextual closure} of a relation $\texttt{R}\subseteq\Set{e}\times\sums{e}$ is 
the smallest relation in $\sums{e}\times\sums{e}$ containing $\texttt{R}$ and respecting the rules of Figure~\ref{fig:context_closure}.
\item The \emph{reduction} $\to\ \subseteq\ \sums{e}\times\sums{e}$ is the contextual closure of $\to_{\beta}\cup\to_{\gto}\cup\to_{\gt}\cup\to_{\gamma}$.
\item The \emph{multistep reduction} $\msto\ \subseteq\ \sums{e}\times\sums{e}$ is the transitive and reflexive closure of $\to$.
\end{enumerate}
\end{defi}

We now provide some examples of reduction.
Note that parallel composition is treated asynchronously, indeed $V\to \sV$ entails $V\paral W\to \sV\paral W$. 

\begin{exa}\label{ex:beta-reduction}\
\begin{enumerate}[1.] 
\item $\bold{D}[\bold{I},\bold{F}]\to_\beta \bold{I}[\bold{F}] + \bold{F}[\bold{I}]\to_\beta \bold{F} + \bold{F}[\bold{I}]\to_\beta \bold{F}$,
\item $\gt[\bold{D}[\bold{D},\bold{D}]]\to_\beta \gt[\bold{D}[\bold{D}]] \to_\beta 0$,
\item $\gt[\bold{I}[\gto(\varepsilon)],\bold{T}[\gto(\varepsilon)]]\to_\beta 
	  \gt[\gto(\varepsilon),\bold{T}[\gto(\varepsilon)]]\to_\beta
	  \gt[\gto(\varepsilon),\lam x.\gto(\varepsilon)]\to_\tau
	  \gt[\gto(\varepsilon),\gto(\varepsilon)]\to_\gamma
	  \gt[\gto(\varepsilon)]\to_\gamma
	  \varepsilon$,  
\item\label{ex:beta-reduction4} $\Xi_{n_1,\ldots,n_m}[\underbrace{\bold{I},\dots,\bold{I}}_{n_1}]\cdots[\underbrace{\bold{I},\dots,\bold{I}}_{n_m}]\msto_\beta \bold{I}$, for all $n_1,\dots,n_m\in\nat$.
\end{enumerate}
\end{exa}

\begin{defi}\label{def:nf}
An expression $A$ is \emph{in normal form} if there is no $\sB$ such that $A\to \sB$.
A sum of expressions $\sA$ \emph{is in normal form} if all its summands are in normal form.
\end{defi}

From Definition~\ref{def:nf} we have that $0$ is in normal form.

The following lemma gives an explicit characterization of terms in normal form.

\begin{lem} If a term $M\in\Set{\gto}$ is in normal form then
\begin{enumerate}[1.] 
\item either $M = \lam \seq x.yP_1\cdots P_n$ for some $n\ge 0$ and each $P_i$ is a bag of terms in normal form, 
\item or $M = \lam \seq x.\gto(\parallel_{i=1}^n\gt[y_i P_{i,1}\cdots P_{i,k_i}])$
where $n\ge 0$, $k_i\ge 0$ and each $P_{i,j}$ is a bag of terms in normal form.\qed
\end{enumerate}
\end{lem}

\subsection{Operational properties.} In this subsection we show that the  \dzlam-calculus enjoys Church-Rosser 
and strong normalization, even in the untyped version of the calculus.

The proof of strong normalization is purely combinatorial, based on a measure given in the following definition.

\begin{defi} The \emph{size} of an expression $A$, written $\size{A}$, is defined by induction as follows:
\begin{iteMize}{$\bullet$}
\item $\size{x} = 1$,
\item $\size{\lam y.M} = \size{M} +1$, 
\item $\size{MP} = \size{M} + \size{P} + 1$,
\item $\size{\gto(V)} = \size{V} + 1$,
\item $\size{\bag{L_1, \ldots, L_k}} = \sum_{i=1}^{k} \size{L_i} + 1$,
\item $\size{\gt[L_1, \ldots, L_k]} = \sum_{i=1}^{k} \size{L_i} + 1$.
\end{iteMize}
The \emph{size} of a sum of expressions $\sA=A_1+\cdots+A_k$, written $\msize{\sA}$, is the multiset of the sizes of the summands of $\sA$,
 namely $\msize{\sA} = [\size{A_1},\ldots,\size{A_k}]$.
\end{defi}

The intuition behind strong normalization is that $\msize \sA$ becomes smaller by replacing one (or more) of its elements by an arbitrary number of smaller elements, i.e.,
with respect to the multiset ordering $\mge$ induced on $\Mfin{\nat}$ by the usual order $>$ of $\nat$.
It is well known that $\mge$ is well-founded.

\begin{thm}\label{thm:CR+SN} 
The \dzlam-calculus with tests is strongly normalizing and Church-Rosser.
\end{thm}

\begin{proof} 
The fact that there are no infinite reduction chains is trivial, since every reduction step decreases the size of an expression. In other words $\sA \to \sB$ entails
 $\msize{\sA}\mge \msize{\sB}$.

For the Church-Rosser property just check local confluence and conclude by Newman's lemma.
\end{proof}

The following lemma formalizes our intuition behind the behaviour of the cork $\gt(\cdot)$.
As a corollary we get that a closed test can only reduce either to $\varepsilon$ or to 0.

\begin{lem}\label{lem:closed-go-to-eps-zero}
For any closed term $M$,  either $\gt[M]\msto\varepsilon$ or $\gt[M]\msto 0$.
\end{lem}

\begin{proof} As \dzlam-calculus with tests is strongly normalizing, 
we have that $M\msto \Sigma_{i = 1}^k M_i$,
where each $M_i$ is a closed  normal form. 
If $k =0$ then $\gt[M]\msto 0$ since $\gt[0] = 0$.
Otherwise for each $M_i$ there are two possibilities:
\begin{iteMize}{$\bullet$}
\item $M_i = \lam\seq x.x_j  P_1\cdots P_n$ with $x_j\in\seq x$ and $n\ge 0$.
	Then $\gt[M_i] \msto \gt[(x_j P_1\cdots P_n)\subst{\seq x}{0}] = \gt[0] = 0$.
\item $M_i = \lam \seq x.\gto(\parallel_{j=1}^n \gt[x_j P_{j,1}\cdots P_{j,k_j}])$ with $n\ge 0$ and $x_j\in\seq x$.
	If $n = 0$ then we have $\parallel_{j=1}^n \gt[x_j P_{j,1}\cdots P_{j,k_j}] = \varepsilon$ and 
	$\gt[\lam \seq x.\gto(\varepsilon)]\msto \gt[\gto(\varepsilon)]\to
	\varepsilon$.
	If $n > 0$ then $\gt[M_i]\msto \gt[\gto(\parallel_{j=1}^n \gt[0 P_{j,1}\subst{\seq x}{0}\cdots P_{j,k_j}\subst{\seq x}{0}])] = 0$.
\end{iteMize}
We conclude since $\gt[M]\msto \Sigma_{i = 1}^k \gt[M_i]$, and this latter 
expression reduces to a finite (possibly empty)  sum of $\varepsilon$'s, which is 
thus equal either to $0$ or to $\varepsilon$.
\end{proof}

\begin{cor}\label{cor:closed-tests-go-to-eps-zero}
If $V$ is a closed test then either $V\msto\varepsilon$ or $V\msto 0$. 
\end{cor}

Therefore, it makes sense to define the convergence of a test as follows.

\begin{defi}\label{def:convergence-no-promotion}
A test $V$ \emph{converges} if and only if $V \msto \varepsilon$.
\end{defi}

It is easy to check that a test $V$ can converge only if it is closed; 
indeed, a free variable $x$ occurring in $V$ cannot be erased during the reduction.

\subsection{\bf Operational Pre-order.} A \emph{term-context} $D\hole{\cdot}$ is a term having one occurrence 
of a \emph{hole}, denoted by $\hole{\cdot}$, appearing in term-position;
a \emph{test-context} $C\hole{\cdot}$ is a test having one occurrence 
of a \emph{hole}, still appearing in term-position.

\begin{defi} Term-contexts $D\hole{\cdot}$ and test-contexts $C\hole{\cdot}$ are defined by the following grammar:
\begin{gather*}
	D\hole{\cdot}\gramm \hole{\cdot} \mid \lam x.D \mid DP \mid M[D,\seq L] \mid \gto(C)\\
	C\hole{\cdot}\gramm \gt[D,\seq L]
\end{gather*}
The set of term-contexts is denoted by $\Set{\gto}_{\hole{\cdot}}$ 
and the set of test-contexts by $\ContSet$.
\end{defi}

Given $M\in\Set{\gto}$ we indicate by $C\hole{M}$ the test resulting 
by blindly replacing $M$ for the hole (allowing capture of free variables) in $C\hole{\cdot}$.
Similarly, given a term-context $D\hole{\cdot}$, $D\hole{M}$ denotes the term obtained 
by blindly substituting $M$ for the hole in $D\hole{\cdot}$.

\begin{exa}\
\begin{enumerate}[1.] 
\item Let $D\hole{\cdot} \ass \lam xy.\hole{\cdot}$, then $D\hole{x[x][y]} = \lam xy.x[x][y]$,
\item let $D\hole{\cdot} \ass \hole{\cdot}[\gto(\varepsilon),\bold{I}]$, then $D\hole{\bold{D}} = \bold{D}[\gto(\varepsilon),\bold{I}],
		\to \gto(\varepsilon)[\bold{I}] + \bold{I}[\gto(\varepsilon)] \to_\beta  \bold{I}[\gto(\varepsilon)]\to_\beta \gto(\varepsilon)$,
\item the simplest test context is $C\hole{\cdot} \ass \gt[\hole{\cdot}]$, 
we have $C\hole{\bold{I}} = \gt[\bold{I}]\to_{\tau} \gt[x\subst{x}{0}] = 0$,
\item let $C\hole{\cdot} \ass \gt[(\lam x.\hole{\cdot})[\gto(\varepsilon)]]$ then $C\hole{x} = \gt[\bold{I}[\gto(\varepsilon)]]\to \gt[\gto(\varepsilon)]\to \varepsilon$.
\end{enumerate}

\end{exa}

We say that a test-context $C\hole{\cdot}$ (resp.\ a term-context $D\hole{\cdot}$) is \emph{closed} if it contains no free variable;
it is \emph{closing $M$} if  $C\hole{M}$ (resp.\ $D\hole{M}$) is closed.

\begin{defi} \label{def:obsle}
The \emph{operational pre-order} $\Tobsle$ on the \dzlam-calculus with tests is defined as follows (for all $M,N \in\Set{\gto}$):
$$
	M\Tobsle N \iff \forall C\hole{\cdot}\in\ContSet\textrm{ closing $M,N$ }(C\hole{M}\msto \varepsilon\ \imp C\hole{N}\msto \varepsilon).
$$
We set $M\Tobseq N$ iff $M\Tobsle N$ and $N\Tobsle M$.
\end{defi}

This coincides with a standard idea of operational preorder. 
The restriction of observations to test-contexts deserves however a discussion. 
First, note that tests provide a canonical notion of observation since -- by design -- they either converge (to $\varepsilon$) or reduce to 0. 
Hence, the choice of test-convergence as the basic observation in our calculus is very natural.

A second motivation comes {\em a posteriori}. Indeed, as we will prove in Section~\ref{sec:FA-notests-nobang} (Theorem~\ref{thm:main1-1}), 
for test-free terms $M,N$ we have $M\Tobsle N$ exactly when, for all test-free term-contexts $D\hole{\cdot}$, $D\hole{M}$ is solvable entails $D\hole{N}$ is solvable 
(the notion of solvability for test-free terms is given in Definition~\ref{def:dzlam-solvable}).

\section{A Relational Semantics}\label{sec:rela-sem}	                         %

This section is devoted to build a relational model $\cD$ of \dzlam-calculus with tests,
that has been first introduced in \cite{BucciarelliEM07} as a model of the ordinary \lam-calculus. 

We first give a sketchy presentation of the Cartesian closed category where $\cD$ lives. 
We recall that the definitions and notations concerning multisets have been introduced in Subsection~\ref{subs:msets}.

\subsection{The Category $\MRel$}
The category $\MRel$ is the co-Kleisli category for the finite-multiset comonad on the category $\bold{Rel}$ of sets and relations.

This category can be described directly as follows:
\begin{iteMize}{$\bullet$}
\item 
    The objects of $\MRel$ are all the sets.
\item 
    A  morphism from $S$ to $T$ is a relation from $\Mfin S$ to $T$; in other words, $\MRel(S,T)=\Pow{\Mfin S\times T}$.
\item 
    The identity of $S$ is the relation $\Id{S}=\{([\alpha],\alpha)\st \alpha\in S\} : S\to S$.
\item 
    The composition of $s : S\to T$ and $t : T\to U$ is defined by:
    $$
    \begin{array}{ll}
    t\comp s=\{(a,\beta)\quad \st&\exists k\in\nat,\ \exists (a_1,\alpha_1),\dots,(a_k,\alpha_k)\in s\textrm{ such that } \\
					       &a = a_1\mcup\dots\mcup a_k\ \text{and}\ ([\alpha_1,\dots,\alpha_k],\beta)\in t\qquad \}.\\
    \end{array}
    $$
\end{iteMize}

\noindent Given two sets $S,T$, we denote by $\With{S}{T}$ their disjoint union $(\{1\}\times S) \cup (\{2\}\times T)$. 
Hereafter we adopt the following convention.

\begin{conv} 
We consider the canonical bijection between $\Mfin{S}\times\Mfin{T}$ and $\Mfin{\With{S}{T}}$ 
as an equality.
Therefore, we will still denote by $(a_1,a_2)$ the corresponding element of $\Mfin{\With{S}{T}}$.
\end{conv}

\begin{thm}\label{thm:MRel-ccc} The category $\MRel$ is a Cartesian closed category.
\end{thm}

\begin{proof} The terminal object $\Termobj$ is the empty set $\emptyset$, and the unique element of
$\MRel(S,\emptyset)$ is the empty relation.

Given two sets $S$ and $T$, their categorical product in $\MRel$ is
their disjoint union $\With{S}{T}$ and the corresponding projections are given by:
\begin{gather*} 
	 \Proj{1}=\{([(1,\ga)],\ga)\st \ga\in S \} : \With{S}{T}\to S,\\
	 \Proj{2}=\{([(2,\ga)],\ga)\st \ga\in T \} : \With{S}{T} \to T.
\end{gather*}

It is easy to check that this is actually the categorical product of $S$ and $T$ in $\MRel$; 
given $s : U\to S$ and $t : U\to T$, the corresponding morphism $\Pair st : U\to\With{S}{T}$ is given by:
$$
  \Pair st=\{(a,(1,\ga))\st(a,\ga)\in s\}\cup\{(b,(2,\gb))\st(b,\gb)\in t\}\,.
$$
Given two objects $S$ and $T$, the exponential object $\Funint ST$ is $\Mfin S\times T$ and the evaluation morphism is given by:
$$
\eval_{ST} =\{(([(a,\gb)],a),\gb)\st a\in\Mfin S\ \text{and}\ \gb\in T\} : \With{\Funint ST}S\to T\,.
$$
Again, it is easy to check that in this way we defined an exponentiation. 
Indeed, given any set $U$ and any morphism $s : \With US\to T$, there is exactly one morphism $\curry(s) :  U\to \Funint ST$ such that:
$$
  \eval_{ST}\comp (\curry(s)\times \Id{S}) = s.
$$
which is $\curry(s)=\{(a,(b,\gb))\st((a,b),\gb)\in s\}$. 
\end{proof}

As shown in \cite{Manzonetto10}, $\MRel$ is actually a Cartesian closed differential category \cite{BluteCS09}.
It is not difficult to check that it is moreover an instance of the categorical framework presented in Section~\ref{sec:thomas}.

\subsection{An Extensional Reflexive Object}

We build a reflexive object $\cD$, which is extensional in the sense that $\cD \cong \Funint{\cD}{\cD}$.
The elements of $\cD$ are infinite sequences of multisets, that are quasi-finite in the following sense.

\begin{defi}An infinite sequence $\ga = (a_1,a_2,\dots)$ of multisets is \emph{quasi-finite} 
if $a_i = []$ holds for all but a finite number of indices $i$.
\end{defi}

Given a set $S$, we denote by $\Omegatuple S$ the set of all quasi-finite 
$\nat$-indexed sequences of finite multisets over $S$.

\begin{defi} We build a family of sets $(D_n)_{n\in\nat}$ as follows:  
\begin{iteMize}{$\bullet$}
\item 
	$D_0=\emptyset$, 
\item 
	$D_{n+1}=\Omegatuple{D_n}$.
\end{iteMize}
Since the operation mapping a set $S$ into $\Omegatuple S$ is monotonic with respect to inclusion\footnote{
This means that $S\subseteq S'$ entails $\Omegatuple S \subseteq \Omegatuple{S'}$.
} and $D_0\subseteq D_1$, we have $D_n\subseteq D_{n+1}$ for all $n\in\nat$. 
Finally, we set $\cD=\bigcup_{n\in\nat}D_n$.
\end{defi}

To define an isomorphism between $\cD$ and $\Mfin{\cD}\times{\cD}$
just note that every element $\ga = (a_1, a_2, a_3,\ldots)\in \cD$ stands for the pair
$(a_1, (a_2, a_3,\ldots))$ and \emph{vice versa}. 
From this simple remark, it follows that $\cD\cong\Funint{\cD}{\cD}$ (we have a canonical bijection between these
two sets, and therefore an isomorphism in $\MRel$).

\begin{nota}
Given $\ga = (a_1, a_2, a_3,\ldots)\in\cD$ and $a\in\Mfin{\cD}$, we write $a\at \ga$ for the element $(a,a_1,a_2, a_3,\ldots)\in\cD$.
We denote by $*$ the element  $([],[],\ldots,[],\ldots)\in\cD$. 
\end{nota}

Remark that $[] :: * = *$.

\subsection{Interpreting the \dzlam-calculus with tests}\label{subsec:dzlam-interpretation}
We now define the interpretation of an expression $A$ of the  \dzlam-calculus with tests
in the model $\cD$. As usual, an expression $A$ will be interpreted by a morphism of the category $\MRel$.

For all terms $M$, bags $P$, tests $Q$ and repetition-free se\-quen\-ces $\seq x, \seq y, \seq z$
respectively containing the free variables of $M,P,Q$, we define by mutual induction the interpretations
 $\Int{M}_{\seq x}:\cD^n \to \cD$,
 $\Int{P}_{\seq y}:\cD^m \to \Mfin{\cD}$ and
 $\Int{Q}_{\seq z}:\cD^k \to \one$ ($\one$ is a singleton set and $n,m,k$ are the lengths of $\seq x,\seq y,\seq z$) as follows\footnote{%
Since $\Mfin{\With ST}\cong\Mfin{S}\times\Mfin{T}$ we have, up to isomorphism, $\Int{M}_{\seq x}\subseteq\Mfin{\cD}^n\times \cD$,
$\Int{P}_{\seq y}\subseteq\Mfin{\cD}^{m+1}$ and $\Int{Q}_{\seq z}\subseteq\Mfin{\cD}^k\times\one \cong \Mfin{\cD}^k$.}:

\begin{iteMize}{$\bullet$}
\item $\Int{x_i}_{\seq x} = \{(([],\ldots,[],[\alpha],[],\ldots,[]),\alpha) \st \alpha \in \cD\}$, where $[\alpha]$ stands in $i$-th position, and
 $\seq x = x_1,\ldots,x_i,\ldots,x_n$,
\item $\Int{\lambda y.M}_{\seq x} = \{(\seq{a},b\at \alpha) \st ((\seq a,b), \alpha)\in\Int{M}_{\seq x,y}\}$, where we suppose wlog that $y\notin\seq x$,
\item $\Int{MP}_{\seq x} = \{ (\seq a_1 \mcup \seq a_2,\alpha) \st \exists b \in \Mfin\cD\  (\seq a_1,b \at \alpha)\in \Int{M}_{\seq x},\
(\seq a_2,b)\in\Int{P}_{\seq x}\}$,
\item $\Int{\gto(V)}_{\seq x} = \{ (\seq{a},*)\st \seq a\in\Int{V}_{\seq x}\}$,
\item $\Int{[L_1,\ldots,L_k]}_{\seq y} = \{ (\mcup_{i=1}^k \seq a_i,[\beta_1,\ldots,\beta_k]) \st (\seq a_i,\beta_i)\in \Int{L_i}_{\seq y},\ 1\le i\le k\}$,
\item $\Int{\gt[M]}_{\seq z} = \{\seq{a} \st (\seq{a}, *)\in\Int{M}_{\seq z}\}$,
\item $\Int{V_1\paral V_2}_{\seq z} = \{\seq a_1\mcup \seq a_2 \st \seq a_1\in\Int{V_1}_{\seq z}, \seq a_2\in\Int{V_2}_{\seq z}\}$,
\item $\Int{\varepsilon}_{\seq z} = \{ ([],\ldots,[]) \}$.
\end{iteMize}
In particular $\Int{[]}_{\seq x} = \{([],\ldots,[])\}\in\Mfin{\cD}^{n+1}$. \\

The interpretation is then extended to the elements of $\sums{e}$ by setting 
$\Int{\Sigma_{i=1}^k A_i}_{\seq x} = \cup_{i=1}^k \Int{A_i}_{\seq x}$.

\begin{rem} 
Since every test $V$ is of the form $\gt[L_1,\ldots,L_k]$ we might define its interpretation directly by setting
$\Int{V}_{\seq x} = \{ \mcup_{i=1}^k \seq a_i \st (\seq a_i,*)\in \Int{L_i}_{\seq x}, 1\le i \le k\}$.
Closed terms (resp.\ tests, bags)  are interpreted by relations between the singleton $\Mfin{\emptyset}$ and $\cD$ (resp. $\one$, $\Mfin{\cD}$),
we denote them plainly as subsets of $\cD$ (resp. $\one$, $\Mfin{\cD}$).
\end{rem}

The following are examples of interpretations.
\begin{exa}\label{ex:interp}\
\begin{enumerate}[1.]
\item\label{ex:interp1} $\Int{\gto(\varepsilon)} = \{*\}$,
\item\label{ex:interp2} $\Int{\bold{I}} = \{[\alpha]\at\alpha \st \alpha \in \cD \}$,
\item\label{ex:interp3} $\Int{\bold{D}}=\{ [[\alpha]::\beta,\alpha]::\beta\st\alpha,\beta\in {\cD} \}$,
\item\label{ex:interp4} $\Int{\bold{D}[\bold{I}]}=\emptyset$. Indeed the elements of $\Int{[\bold{I}]}$ can never match $[[\alpha]::\beta,\alpha]$ because they are singleton multisets.
\end{enumerate}
\end{exa}

\begin{conv} Hereafter, whenever we write $\Int{\sA}_{\seq x}$ we suppose that $\seq x$ 
is a repetition-free list of variables of length $n$ containing $\FV(\sA)$.
Moreover, we will sometimes silently use the fact $\Int{M}_{\seq x,y}=  \{((\seq a,[]),\alpha) \st (\seq a,\alpha)\in \Int{M}_{\seq x}\}$
whenever $y\notin\FV(M)$.
\end{conv}

Clearly the interpretation is monotonic, in the sense expressed by the following lemma.

\begin{lem}\label{lemma:monotonicity}
For any test-context $C\hole{\cdot}$ (resp.\ term-context $D\hole{\cdot}$) with free variables $\seq y$,
if $\Int{M}_{\seq x}\subseteq\Int{N}_{\seq x}$ then
$\Int{C\hole{M}}_{\seq x,\seq y}\subseteq\Int{C\hole{N}}_{\seq x,\seq y}$ (resp.\ $\Int{D\hole{M}}_{\seq x,\seq y}\subseteq\Int{D\hole{N}}_{\seq x,\seq y}$).
\end{lem}

\begin{proof} By a straightforward mutual induction on $C\hole{\cdot},D\hole{\cdot}$.
\end{proof}

The following substitution lemmas are needed for proving the invariance
of the interpretation under reduction.
The proofs are lengthy but not difficult, and are provided in Appendix~\ref{app:tech_app}.

\begin{lem}[Linear Substitution Lemma]\label{lemma:lsubst}
Let $M\in\Set{\gto}, V\in\Set{\gt}$ and $P = [L_1,\ldots,L_k]\in\Set{b}$
such that $\dg{y}{M} = \dg{y}{V} = k$. 
We have:
\begin{enumerate}[(i)]
\item 
$(\seq a, \alpha)\in \Int{M\lsubst{y}{P}}_{\seq x}$ iff there exist
$(\seq a_i,\beta_i)\in\Int{L_i}_{\seq x}$ (for $1\le i \le k$) and
$\seq a_0\in\Mfin{\cD}^n$ such that
$((\seq a_0,[\beta_1,\ldots,\beta_k]), \alpha)\in\Int{M}_{\seq x,y}$
and 
$\mcup_{i=0}^{k} \seq a_i=\seq a$.
\item 
$\seq a\in \Int{V\lsubst{y}{P}}_{\seq x}$  iff there exist
$(\seq a_i,\beta_i)\in\Int{L_i}_{\seq x}$ (for $1\le i \le k$) and
$\seq a_0\in\Mfin{\cD}^n$ such that
$(\seq a_0,[\beta_1,\ldots,\beta_k])\in\Int{V}_{\seq x,y}$ and 
$\mcup_{i=0}^{k} \seq a_i=\seq a$.
\end{enumerate}
\end{lem}

\begin{lem}[Regular Substitution Lemma]\label{lemma:subst}
Let $M\in\Set{\gto}$, $V\in\Set{\gt}$ and $\sN\in\sums{\gto}$. 
We have:
\begin{enumerate}[\em(i)]
\item $(\seq a, \alpha)\in \Int{M\subst{y}{\sN}}_{\seq x}$ iff 
	$\exists k\in\nat,\ \exists\beta_1,\ldots,\beta_k\in \cD$, 
	$\exists\seq a_0,\ldots,\seq a_{k}\in\Mfin{\cD}^n$ such that
	$(\seq a_i,\beta_i)\in\Int{\sN}_{\seq x}$ (for $1\le i \le k$),
	$((\seq a_0,[\beta_1,\ldots,\beta_k]), \alpha)\in\Int{M}_{\seq x,y}$ and 
	$\seq a = \mcup_{j=0}^{k} \seq a_j$,
\item $\seq a\in \Int{V\subst{y}{\sN}}_{\seq x}$ iff 
	$\exists k\in\nat,\ \exists\beta_1,\ldots,\beta_k\in \cD$, 
	$\exists\seq a_0,\ldots,\seq a_{k}\in\Mfin{\cD}^n$ such that
	$(\seq a_i,\beta_i)\in\Int{\sN}_{\seq x}$ (for $1\le i \le k$) and
	$(\seq a_0,[\beta_1,\ldots,\beta_k])\in\Int{V}_{\seq x,y}$ and 
	$\seq a = \mcup_{j=0}^{k} \seq a_j$.
\end{enumerate}
\end{lem}\medskip

\noindent The substitution lemmas above generalize straightforwardly to sums.
Although Lem\-ma~\ref{lemma:subst} is stated in full generality, 
for the \dzlam-calculus with tests is only useful for $\sN = 0$.
We keep this formulation since it is closer to the one needed in Section~\ref{sec:FullRC}
for the full \dlam-calculus with tests.

\begin{thm}\label{thm:Dmodel}
$\cD$ is a model of the \dzlam-calculus with tests, i.e., if ${\sA}\msto{\sB}$ then $\Int\sA_{\seq x}=\Int\sB_{\seq x}$.
\end{thm}

\begin{proof} It is easy to check that the interpretation is contextual. 
The fact that the semantics is invariant under reduction follows from 
Lemmas~\ref{lemma:lsubst} and \ref{lemma:subst}.
\end{proof}

\section{Full Abstraction for $\dzlam$-Calculus with Tests}\label{sec:FA-dzlam-with-Tests}

A model is \emph{equationally fully abstract} if the equivalence 
induced on terms by their interpretations is exactly $\Tobseq$; 
it is \emph{inequationally fully abstract} if the induced preorder is $\Tobsle$.
Obviously, every inequationally fully abstract model is also equationally fully abstract.

In this section we prove that $\cD$ is inequationally fully abstract for the \dzlam-calculus with tests (Theorem~\ref{thm:main1}), i.e., that
$\Int{M}_{\seq x}\subseteq\Int{N}_{\seq x}$ if and only if $M\Tobsle N$.

\subsection{Building Separating Test-Contexts}

We are going to associate a 
test-context $\cont{\alpha}{\cdot}$ with each element $\alpha\in\cD$, 
the idea being that -- for every closed term $M$ --
we have $\alpha\in\Int{M}$ if and only if $\cont{\alpha}{M}$ converges. 

\begin{defi}
Let $\alpha\in \cD$.
The {\em rank of $\alpha$}, written $\rank(\alpha)$, is the least $n\in\nat$ such that $\alpha\in D_{n+1}$;
the \emph{length of $\alpha$}, written $\len(\alpha)$, is $0$ if
$\alpha=*$, and it is the unique  $r$ such that
$\alpha = a_1 \at \cdots \at a_r\at *$
with $a_r \neq []$, otherwise.
\end{defi}
Note that if $\alpha = a_1 \at \cdots \at a_r \at *$ then for all $1\le i\le r$ and $\alpha_i\in a_i$ 
we have $\rank(\alpha) > \rank(\alpha_i)$. 
Hence $\rank(\alpha) = 0$ entails $\alpha = *$ and the following definition is well-founded.

\begin{defi}\label{def:alpha+-}
For $\alpha\in\cD$ of the form
$\alpha = [\alpha_{1,1},\ldots,\alpha_{1,k_1}] \at \cdots \at [\alpha_{r,1},\ldots,\alpha_{r,k_r}]\at *$ with $\len(\alpha) = r$, define by mutual induction  
a closed term $\trm{\alpha}$ and a test-context $\cont{\alpha}{\cdot}$ as follows:
\begin{iteMize}{$\bullet$}
\item $\trm{\alpha} = \lambda x_1\ldots x_r.\gto(\parallel_{i=1}^r (\cont{(\alpha_{i,1})}{x_i}\paral\cdots\paral\cont{(\alpha_{i,k_i})}{x_i}))$,
\item $\cont{\alpha}{\cdot} = \gt[\hole{\cdot}[\trm{(\alpha_{1,1})},\ldots,\trm{(\alpha_{1,k_1})}] \cdots [\trm{(\alpha_{r,1})},\ldots,\trm{(\alpha_{r,k_r})}]]$.
\end{iteMize}
Given $a = [\ga_1,\ldots,\ga_k]$ we set $\trm{a} = [\trm{\ga_1},\ldots,\trm{\ga_k}]$.
\end{defi}

\begin{exa} We have:
\begin{enumerate}[1.] 
\item $\trm{*} = \gto(\varepsilon)$ (as the empty parallel composition is equal to $\varepsilon$),
\item $\cont{*}{\cdot} = \gt[\hole{\cdot}]$, hence:
\item $\cont{([*]\at *)}{\cdot} = \gt[\hole{\cdot}[\gto(\varepsilon)]]$.
\end{enumerate}
\end{exa}

The next lemma, along with its corollaries, shows the interplay between the 
elements of $\cD$ and the terms/tests of Definition \ref{def:alpha+-}.
It provides  the main motivation for our  extension
of the \dzlam-calculus.

\begin{lem}\label{lem:int_alphas} Let $\alpha\in \cD$. Then:
\begin{enumerate}[\em(i)]
\item\label{lem:int_alphas1} 
	$\Int{\trm{\alpha}}= \{ \alpha \}$,
\item\label{lem:int_alphas2} 
	$\Int{\cont{\alpha}{x}}_{x} = \{[\alpha]\}$.
\end{enumerate}
\end{lem}

\begin{proof} 
The points (\ref{lem:int_alphas1}) and (\ref{lem:int_alphas2}) are 
proved simultaneously by induction on $\rank(\alpha)$.
We write IH(\ref{lem:int_alphas1}) and IH(\ref{lem:int_alphas2}) 
for the induction hypotheses concerning (\ref{lem:int_alphas1}) and (\ref{lem:int_alphas2}), respectively. 

If $\rank(\alpha) = 0$ then $\alpha=*$, hence $\Int{\trm{*}} = \Int{\gto(\varepsilon)} = \{*\}$ and $\Int{\cont{*}{x}}_{x} =\Int{\gt[x]}_{x} = \{[*]\}$.

If $\rank(\alpha)> 0$ and $\len(\alpha) = r$, we have $\alpha = a_1 \at \cdots\at a_r\at *$
with $a_i = [\alpha_{i,1},\ldots,\alpha_{i,k_i}]$ for $1\le i\le r$.

We prove  (\ref{lem:int_alphas1}). Remember that by definition  
$\Int{\trm{\alpha}} = \Int{\lam y_1\ldots y_r.\gto(\parallel_{i = 1}^r 
\parallel_{j= 1}^{k_i} \cont{(\alpha_{i,j})}{y_i})}$.
So we have  $\beta\in \Int{\trm{\alpha}}$ if and only if 
$\beta = b_1\at\cdots\at b_r\at *$ and for all $1\le i \le r,\ 1\le j\le k_i$ there is 
$\seq{d}_{i,j}\in\Int{\cont{(\alpha_{i,j})}{y_i}}_{\seq y}$ such that 
$\seq b = \mcup_{i =1}^r \mcup_{j =1}^{k_i}
\seq d_{i,j}$.
By IH(\ref{lem:int_alphas2}) we have $\seq d_{i,j}\in\Int{\cont{(\alpha_{i,j})}{y_i}}_{\seq y}$ 
iff $\seq d_{i,j} = (\seq{[]},[\alpha_{i,j}],\seq{[]})$ 
where $[\alpha_{i,j}]$ appears in $i$-th position.
Therefore $\mcup_{j =1}^{k_i} \seq d_{i,j}=(\seq{[]},a_i,\seq{[]})$ and
$b_i = a_i$ for every index $i$. Thus $\beta = \alpha$.

We prove (\ref{lem:int_alphas2}). By definition we have $\Int{\cont{\alpha}{x}}_{x} = \Int{\gt[x\trm{a_1}\cdots\trm{a_r}]}_{x}$.
Therefore $ c\in\Int{\cont{\alpha}{x}}_{x}$ if and only if there are 
$b_i = [\beta_{i,1},\ldots,\beta_{i,k_i}]$,
$ c_0, c_{i,1},\ldots, c_{i,k_i}\in\Mfin{\cD}$ (for $1\le i \le r$)
such that $( c_0,b_1\at\cdots\at b_{r}\at *)\in\Int{x}_{x}$, 
$( c_{i,j}, \beta_{i,j}) \in\Int{\trm{(\alpha_{i,j})}}_{x}$ (for all $1\le i\le r$ and $1\le j\le k_i$) and 
$ c =  c_0\mcup (\mcup_{i =1}^{r} \mcup_{j =1}^{k_i} c_{i,j})$.
As, by IH(\ref{lem:int_alphas1}), $\Int{\trm{(\alpha_{i,j})}}_{ x} = \{ ({[]},\alpha_{i,j})\}$ we get 
$c_{i,j} =  []$ and $\beta_{i,j} = \alpha_{i,j}$. 
Thus $ c =  c_0$, $\alpha = b_1\at\cdots\at b_r\at *$ and from this it follows that $( c,\alpha)\in\Int{x}_{ x}$.
We conclude that $c=[\alpha]$.
\end{proof}

\begin{cor}\label{cor:general-test}
$\Int{\cont{\alpha}{M}}_{\seq x} = \{\seq c \st (\seq c, \alpha)\in\Int{M}_{\seq x}\}$.
\end{cor}

\begin{proof}
By Lemma \ref {lem:int_alphas}(ii) we have that $\Int{\cont{\alpha}{y}}_{\seq x,y}=\{([],\ldots,[],[\alpha])\}$. 
As $\cont{\alpha}{\cdot}$ does not have outer $\lam$-abstractions we have  $\cont{\alpha}{M}=\cont{\alpha}{y}\lsubst{y}{[M]}$.
We then apply Lemma \ref{lemma:lsubst} to conclude.
\end{proof}

\begin{cor}
All finite subsets of $\cD$ are definable.
\end{cor}

\begin{proof} By Lemma~\ref{lem:int_alphas}(\ref{lem:int_alphas1}), 
for every finite set $u = \{\alpha_1,\ldots,\alpha_k\}$
we have $\Int{\trm{\alpha_1}+\cdots + \trm{\alpha_k}}=u$.
\end{proof}

Lemma~\ref{lem:int_alphas} reveals
the behaviour of a test-context $\cont{\alpha}{\cdot}$ when applied to a term $\trm{\beta}$.

\begin{cor}\label{lem:alpha-beta-to-eps-zero}
Let $\alpha,\beta\in \cD$. If $\alpha = \beta$ then $\cont{\alpha}{\trm{\beta}} \msto \varepsilon$,
otherwise $\cont{\alpha}{\trm{\beta}} \msto 0$.
\end{cor}

\begin{proof} 
By Lemma  \ref{lem:int_alphas},
$\Int{\cont{\alpha}{\trm{\beta}}}=\{()\}\subseteq\Mfin{\cD}^0$ 
if  $\alpha = \beta$, $\emptyset$ otherwise.
By Corollary~\ref{cor:closed-tests-go-to-eps-zero}, we know  that
$\cont{\alpha}{\trm{\beta}}$  reduces either to $\varepsilon$ or to $0$.
The result follows by soundness (Theorem~\ref{thm:Dmodel}).
\end{proof}

\subsection{(In)equational Full Abstraction}\label{subsec:FAres1}

In this subsection, we show that the operational preorder $\Tobsle$ 
(see Definition~\ref{def:obsle}) coincides with the inclusion of interpretations in $\cD$. 
We first need a couple of preliminary results.

Recall from Definition~\ref{def:alpha+-} that, given $a = [\ga_1,\ldots,\ga_k]$, 
$\trm{a}$ denotes the bag $[\trm{\ga_1},\ldots,\trm{\ga_k}]$.

\begin{lem}\label{lemma:newlsubst}
Let $V\in \Set{\gt}$, $\FV(V)\subseteq \seq x$ and $\seq a\in\Mfin{\cD}^n$ where $n$ is the length of $\seq x$. 
Then $\seq{a} \in \Int{V}_{\seq{x}}$ if and only if $ \Int{V\lsubst{\seq{x}}{\trm{\seq{a}}}}\neq \emptyset$ and $\dg{x_i}{V} = \card a_i$.
\end{lem}

\begin{proof}
The result follows by applying $n$ times (one for each variable in $\seq{x}$) Lemma~\ref{lemma:lsubst} and Corollary~\ref{cor:general-test}.
\end{proof}

\begin{rem}\label{rem:subst}
For $M\in\Set{\gto}$ we have $(\cont{\ga}{M})\lsubst{\seq x}{\trm{\seq{a}}} = \cont{\ga}{M\lsubst{\seq x}{\trm{\seq{a}}}}$.
\end{rem}

The ensuing proposition is the key argument for proving that the model $\cD$ is inequationally fully abstract.

\begin{prop}\label{prop:context}
Let $M \in \Set{\gto}$, $\seq{x}\supseteq\FV(M)$, $\alpha\in \cD$ and $\seq a\in\Mfin{\cD}$. 
The following are equivalent:
\begin{enumerate}[\em(i)]
\item $(\seq{a},\alpha) \in \Int{M}_{\seq{x}}$,
\item $\cont{\alpha}{M\lsubst{\seq{x}}{\trm{\seq{a}}}}\msto \varepsilon$.
\end{enumerate}
\end{prop}

\proof We have the following chain of equivalences: 
\begin{iteMize}{$\bullet$}
\item[\ ]$(\seq{a},\alpha) \in \Int{M}_{\seq{x}}$ 
$\iff$ $\seq a \in \Int{\cont{\ga}{M}}_{\seq{x}}$, by Corollary~\ref{cor:general-test},
\item[$\iff$] $\Int{\cont{\ga}{M\lsubst{\seq x}{\trm{\seq{a}}}}}\neq\emptyset$ and $\dg{x_i}{M} = \card a_i$, by Lem\-ma~\ref{lemma:newlsubst}, using Remark~\ref{rem:subst},
\item[$\iff$] $\cont{\ga}{M\lsubst{\seq x}{\trm{\seq a}}}\msto\varepsilon$, by Corollary~\ref{cor:closed-tests-go-to-eps-zero}, i.e. the fact that closed tests can only reduce to either $\varepsilon$ or $0$, and Theorem~\ref{thm:Dmodel}, i.e. the soundness of the model.\qed
\end{iteMize}

We are now able to prove the main result of the section. 

\begin{thm}\label{thm:main1} 
$\cD$ is inequationally fully abstract for the \dzlam-calculus with tests (for all $M,N\in\Set{\gto}$):
$$
\Int{M}_{\seq x}\subseteq \Int{N}_{\seq x} \iff M\Tobsle N
$$
\end{thm}

\begin{proof} 
$(\imp)$ 
Assume that $\Int{M}_{\seq x}\subseteq\Int{N}_{\seq x}$, 
and let  $C\hole{\cdot}$ be a test-context closing both $M$ and 
$N$ and such that $C\hole{M}\msto\varepsilon$.
By Theorem~\ref{thm:Dmodel}, $\Int{C\hole{M}} = \Int{\varepsilon} = \{()\}$.
By monotonicity of the interpretation we get $\Int{C\hole{M}}\subseteq\Int{C\hole{N}}$,
thus $\Int{C\hole{N}}\neq\emptyset$. 
By Corollary~\ref{cor:closed-tests-go-to-eps-zero} this entails that $C\hole{N}\msto \varepsilon$. 

$(\Leftarrow)$ Suppose, by the way of contradiction, that $M\Tobsle N$ holds but 
there is an element $(\seq a,\alpha)\in\Int{M}_{\seq x}-\Int{N}_{\seq x}$.
Then the test-context $C\hole{\cdot} = \cont{\alpha}{(\lam\seq x.\hole{\cdot})\trm{\seq{a}}}$
is such that $C\hole{M}\msto \cont{\alpha}{M\lsubst{\seq{x}}{\trm{\seq{a}}}}\msto \varepsilon$
and $C\hole{N}\not\msto \varepsilon$ by Proposition~\ref{prop:context},
which is a contradiction.
\end{proof}

\begin{cor}\label{cor:main1}  $\cD$ is equationally fully abstract for the \dzlam-calculus with tests.
\end{cor}
The reader who is only interested in the extension of Theorem~\ref{thm:main1} (and of its corollary) to the full \dlam-calculus with tests 
can skip safely the next section.

\section{Full Abstraction for $\dzlam$-Calculus without Tests}\label{sec:FA-notests-nobang}

In this section we are going to prove that tests do not add any discriminatory power 
to the contexts already present in the \dzlam-calculus.
This means that whenever there is a test-context $C\hole{\cdot}$ separating two test-free terms $M,N$ (sending, say, $M$ to $\varepsilon$ and $N$ to 0)
there exists also a term-context $D\hole{\cdot}$ that is still able to separate $M$ from $N$, without using the operators $\gt$ and $\gto$.
(As we will discuss in Section~\ref{sec:conclusions}, this is not the case for the full $\dlam$-calculus with tests.)

From this syntactic result and the full abstraction for the \dzlam-calculus with tests (Theorem~\ref{thm:main1})
we conclude that the model $\cD$ is also inequationally fully abstract for its test-free fragment (Theorem~\ref{thm:main1-1}, below).

\subsection{The $\dzlam$-Calculus (Without Tests)} 
The \emph{$\dzlam$-calculus} is a restriction of the $\dzlam$-calculus with tests presented in Section~\ref{sec:dzlam-calculus+tests}. 
The restriction is obtained by erasing from the syntax the constructors $\gt$ and $\gto$ and the corresponding reduction rules, i.e.\ $(\gt),(\gto)$ and $(\gamma)$.
In other words the tests are no longer part of the language and $\to_\beta$ is the only reduction rule of the system.

This description is enough to completely characterize the system --- for a more detailed description, see \cite{EhrhardR06bis,EhrhardR08}.

\begin{nota} We write $\Set{r}$ (resp.\ $\sums{r}$) for the set of (resp.\ finite sums of) terms of the $\dzlam$-calculus.
The set of all (term-)contexts of the $\dzlam$-calculus will be denoted by $\Set{r}_{\hole{\cdot}}$.
\end{nota}

We still write $M,N,L,H$ for terms in $\Set{r}$, $\sM,\sN,\sL,\sH$ for sums of terms in $\sums{r}$, $P,Q$ for bags
and $D\hole{\cdot}$ for contexts. 
This will not create confusion because we will always specify the set they belong to.

In order to properly define the operational pre-order in this setting, we first need to introduce the notion of solvable term.

\subsection{Solvability in the $\dzlam$-Calculus}
In \lam-calculus \cite{Bare} a term $M$ is \emph{solvable} whenever there exist suitable arguments that, once supplied to $M$, make it reduce to the identity ---
this means that $M$ it is able to interact operationally with the environment.

In resource calculi solvability has been thoroughly studied by Pagani and Ronchi Della Rocca in \cite{PaganiR11,PaganiR10}. 
 Their work needs to be adapted because of the absence of promotion in our system. For the \dzlam-calculus the good notion of solvable term is the following.

\begin{defi}\label{def:dzlam-solvable}
A term $M\in\Set{r}$ is \emph{solvable} if there is a term-context $D\hole{\cdot}$ such that $D\hole{M}\msto_\gb \bold{I}$. We say that $M$ is \emph{unsolvable}
 otherwise.
\end{defi}

Reading \cite{PaganiR11,PaganiR10} one may wonder why in the previous definition we do not ask more generally that $D\hole{M}\msto_\gb \bold{I} + \sN$ for some $\sN\in\sums{r}$.
This is due to the fact that in our \dzlam-calculus the two definitions are equivalent, as shown in the next lemma. 
(So we choose the easier formulation.)

\begin{lem}\label{lemma:easylemma}
Let $M\in\Set{r}$ be a closed term. If $M\msto_\gb \bold{I} + \sM$ for some $\sM\in\sums{r}$, then there exists a sequence $\seq P$ of closed bags such that
 $M\seq P\msto_\gb \bold{I}$.
\end{lem}

\begin{proof}
Suppose $M$ closed such that $M\msto_\gb \bold{I} + \sM$.
Then $\sM$ is also closed and normalizes to a sum $\sM'=\Sigma_{j=1} \lambda \seq y_j.M_j'$
such that each $M'_j$ is not an abstraction itself. 
Now, if $\sM' = \bold{I}$ then we are done as $M\msto_\gb \bold{I} + \sM\msto_\gb \bold{I} + \bold{I} = \bold{I}$.
Otherwise, let $h$ be the maximum among the lengths of the sequences $\seq y_j$. 
Then $\sM'[\bold{I}]^{\sim h}$ is again a sum of closed terms
 and normalizes to a sum $\sM''$ of closed abstraction terms whose size is strictly smaller than $\sM'$. The reason is that for each summand
 $(\lambda \seq y_j.M_j')[\bold{I}]^{\sim h}$ which does not reduce to $0$, $M_j'$ must contain exactly one occurrence of each variable in $\seq y_j$.
 Hence $M_j'\lsubst{{\seq y_j}}{[\bold{I}]}\subst{\seq y_j}{0}$ has the same size as $(\lambda \seq y_j.M_j')$ but it reduces (via contraction of the $\bold{I}$ that has replaced the head
 variable of $M_j'$) to a term having a strictly smaller size, unless $\lambda \seq y_j.M_j' \equiv \bold{I}$. Iterating this reasoning for at most a number of times 
 equal to $k = \size{\sM}+1$ one obtains that either $\sM[\bold{I}]^{\sim k}\msto_\gb 0$ or $\sM[\bold{I}]^{\sim k}\msto_\gb \bold{I}$.
 
Therefore $M[\bold{I}]^{\sim k}\msto_\gb \bold{I}[\bold{I}]^{\sim k} + \sM[\bold{I}]^{\sim k}\msto_\gb \bold{I}$, as desired.
\end{proof}

As in \cite{PaganiR11,PaganiR10} we are going to characterize solvability from both a syntactic and a semantic point of view (Theorem~\ref{thm:equivalence-solvable-normalizable}).

\begin{prop}\label{prop:solvability}
Let $M\in\Set{r}$ and let $ \FV(M) = \seq x$.
If $M$ reduces to a normal form different from 0, then there are two sequences $\seq P, \seq P'$ of closed bags such that:
$$ (M\seq P)\lsubst{\seq x}{\seq P'}\subst{\seq x}{0}\msto_\beta \bold{I} + \sM,\textrm{ for some }\sM\in\sums{r}. $$
\end{prop}

\begin{proof}
By induction on the size of $M$.
Let $\seq x = x_1,\ldots,x_n$ and suppose that $M\msto_\gb \lam y_1\dots y_m.y Q_1\cdots Q_q + \sM$ where 
$m,q\in\nat$, 
$Q_i = [M_{i,1},\dots, M_{i,k_i}]$ for all $1\le i \le q$, 
each $M_{i,j}$ is in normal form for every $1\le j\le k_i$ and $\sM\in\sums{r}$.
For the sake of simplicity, assume $y = y_h$ for some $1\le h\le m$ (the proof is analogous when $y\in\seq x$).

By induction hypothesis, for all  $1\le i \le q$  and $1\le j \le k_i$
 there are sequences $\seq P_{i,j},\seq P'_{i,j},\seq P''_{i,j}$ of closed bags such that $M_{i,j}\seq P_{i,j}\lsubst{\seq y}{\seq P'_{i,j}}\lsubst{\seq x}{\seq P''_{i,j}}\subst{\seq y,\seq x}{0}\msto_\gb \bold{I} + \sM_{i,j}$ for some $\sM_{i,j}\in\sums{r}$.
In the following, we will denote by $\sigma_{i,j}$ the substitution $\lsubst{\seq y}{\seq P'_{i,j}}\lsubst{\seq x}{\seq P''_{i,j}}\subst{\seq y,\seq x}{0}$.

We start by defining the closed term $H$ that will be plugged in head position:
$$
	H = \lam z_1\dots z_q.\Xi_{k_1,\ldots,k_q}[z_1\seq P_{1,1},\dots,z_1\seq P_{1,k_1}]\cdots[z_q\seq P_{q,1},\dots,z_q\seq P_{q,k_q}].
$$
In the rest of the proof we shorten $\Xi_{k_1,\ldots,k_q}$ to $\Xi$.
In addition, we set:
$$
\begin{array}{ll}
Q'_k = \mcup^q_{i = 1} \mcup^{k_i}_{j = 1} P'_{i,j,k}&\textrm{ for all }1\le k \le m\\
Q''_k = \mcup^q_{i = 1} \mcup^{k_i}_{j = 1} P''_{i,j,k}&\textrm{ for all }1\le k \le n\\
\end{array}
$$
We now prove that $(MQ'_1\cdots Q'_{h-1}(Q'_{h}\mcup [H])Q'_{h+1} \cdots Q'_m) \lsubst{\seq x}{\seq Q''}\subst{\seq x}{0}\msto_\gb \bold{I} + \sM$. In the reduction path we will only focus on the term reducing to $\bold{I}$ by collecting all the others into generic $\sM$'s.
Indeed, we have:
$$
\begin{array}{l}
(MQ'_1\cdots Q'_{h-1}(Q'_{h}\mcup [H])Q'_{h+1} \cdots Q'_m) \lsubst{\seq x}{\seq Q''}\subst{\seq x}{0}\msto_\gb\\
((\lam y_1\dots y_m.y Q_1\cdots Q_q)Q'_1\cdots Q'_{h-1}(Q'_{h}\mcup [H])Q'_{h+1} \cdots Q'_m) \lsubst{\seq x}{\seq Q''}\subst{\seq x}{0} + \sM_1\msto_\gb\\
(HQ_1\cdots Q_q)\lsubst{\seq y}{\seq Q'}\lsubst{\seq x}{\seq Q''}\subst{\seq y,\seq x}{0} + \sM_2\msto_\gb\\
(\Xi [M_{1,1}\seq P_{1,1},\dots,M_{1,k_1}\seq P_{1,k_1}]\cdots[M_{q,1}\seq P_{q,1},\dots,M_{q,k_q}\seq P_{q,k_q}])\lsubst{\seq y}{\seq Q'}\lsubst{\seq x}{\seq Q''}\subst{\seq y,\seq x}{0} + \sM_3\msto_\gb\\
\Xi [M_{1,1}\seq P_{1,1}\sigma_{1,1},\dots,M_{1,k_1}\seq P_{1,k_1}\sigma_{1,k_1}]\cdots[M_{q,1}\seq P_{q,1}\sigma_{q,1},\dots,M_{q,k_q}\seq P_{q,k_q}\sigma_{q,k_q}] + \sM_4\\
\end{array}
$$
By induction hypothesis, this reduces to $\Xi [\underbrace{\bold{I},\dots,\bold{I}}_{k_1}]\cdots[\underbrace{\bold{I},\dots,\bold{I}}_{k_q}] + \sM_5\msto_\gb \bold{I} + \sM$ for $\sM \ass \sM_5$ (see Example~\ref{ex:beta-reduction}(\ref{ex:beta-reduction4})).
\end{proof}

Note that in the statement above $\sM$ must be closed because $M\seq P\lsubst{\seq x}{\seq P'}\subst{\seq x}{0}$ is.

\begin{thm}\label{thm:equivalence-solvable-normalizable}
Let $M\in\Set{r}$, then the following three sentences are equivalent.
\begin{enumerate}[\em(i)]
\item\label{thm:equivalence-solvable-normalizable1} 
	$M$ is solvable,
\item\label{thm:equivalence-solvable-normalizable2} 
	$M \msto_\gb N + \sN$ for some $N$ in normal form and $\sN\in\sums{r}$,
\item\label{thm:equivalence-solvable-normalizable3} 
	$\Int{M}_{\seq x}\neq\emptyset$.
\end{enumerate}
\end{thm}

\begin{proof} (i $\Rightarrow$ ii) Suppose by contradiction that there is no normal $N$ such that $M \msto_\gb N + \sN$ for some $\sN\in\sums{r}$.
Since the \dzlam-calculus is strongly normalizing, the only possibility is that $M\msto_\gb 0$.
Therefore, for every term-context $D\hole{\cdot}$ we would have $D\hole{M}\msto_\gb D\hole{0} = 0$.
This is a contradiction since the calculus is Church-Rosser and by hypothesis there should be a term-context $D\hole{\cdot}$ such that $D\hole{M}\msto_\beta \bold{I}$.

(ii $\Rightarrow$ i) Let $M \msto_\gb N + \sN$ and $\seq x = \FV(M)$. 
By Proposition~\ref{prop:solvability} there are sequences $\seq P,\seq P'$ of closed bags such that 
$M\seq P\lsubst{\seq x}{\seq P'}\subst{\seq x}{0}\msto_\gb \bold{I} + \sM$ for some $\sM\in\sums{r}$.
By Lemma~\ref{lemma:easylemma} there is a sequence $\seq P''$ of closed bags such that 
$M\seq P\lsubst{\seq x}{\seq P'}\subst{\seq x}{0}\seq P''\msto_\gb \bold{I}$. 
It is then enough to set $D\hole{\cdot} = (\lam \seq x.\hole{\cdot}\seq P)\seq P'\seq P''$
to have $D\hole{M}\msto_\gb \bold{I}$.

(ii $\Rightarrow$ iii) By Theorem~\ref{thm:Dmodel} $\Int{M}_{\seq x} = \Int{N}_{\seq x} \cup \Int{\sN}_{\seq x}$ for $\seq x = x_1,\dots,x_n\subseteq\FV(M)$.
We now prove, by structural induction on $N$, that $\Int{N}_{\seq x}\neq\emptyset$ for all $N$ in normal form. 
Let $N = \lambda z_1\dots z_m.yP_1\cdots P_p$ where $p\ge 0$ and each $P_i = [L_{i,1},\dots,L_{i,k_i}]$ 
contains all normal terms.
For the sake of simplicity we assume $y = z_h\in\seq z$ (the proof is analogous for $y\in\seq x$).

Base case $p=0$. Trivial since for all $\alpha\in\cD$ we have $((\seq{[]},[\alpha],\seq{[]}),\alpha)\in\Int{N}_{\seq x}$ where $[\alpha]$ occurs in $h$-th position.

Induction case $p>0$. By induction hypothesis, there exist 
$(\seq c_{i,j},\gb_{i,j}) \in \Int{L_{i,j}}_{\seq x,\seq z}$ for each $1\le i\le p$ and $1\le j\le k_p$.
Let $b_i = [\gb_{i,1},\dots,\gb_{i,k_i}]$ for every $1\le i \le p$
and
$\seq a_0 = (\seq{[]},[b_1\at \cdots\at b_p\at \alpha],\seq{[]})\in\Mfin{\cD}^{n+m}$ where
the only non-empty multiset is in $n+h$ position.
Then $(\seq a_0,b_1\at \cdots\at b_p\at \alpha)\in \Int{z_h}_{\seq x,\seq z}$ and $(\seq a_i,b_i) \in\Int{P_i}_{\seq x,\seq z}$ for $\seq a_i = \mcup_{j = 1}^{k_i} \seq c_{i,j}$.
It follows that $(\seq a_0\mcup\seq a_1\mcup\cdots\mcup \seq a_p,\alpha)\in\Int{z_hP_1\cdots P_p}_{\seq x,\seq z}$.
We conclude since $\Int{z_hP_1\cdots P_p}_{\seq x,\seq z}\neq\emptyset$ if and only if $\Int{\lam z_1\dots z_m.z_hP_1\cdots P_p}_{\seq x}\neq\emptyset$.

(iii $\Rightarrow$ ii) Suppose that $M\msto_\gb 0$. Then by Theorem~\ref{thm:Dmodel} we have $\Int{M}_{\seq x} = \Int{0}_{\seq x} = \emptyset$,
which is a contradiction.
\end{proof}

\begin{defi}\label{def:obsle-notest-nobangs} 
The \emph{operational pre-order} $\obsle$ on the $\dzlam$-calculus is defined as follows (for all $M,N\in\Set{r}$): 
$$
	M \obsle N \iff \forall D\hole{\cdot}\in\Set{r}_{\hole{\cdot}} \textrm{ closing $M,N$} (D\hole{M} \mbox{ is solvable } \imp D\hole{N} \mbox{ is solvable}). 
$$
We set $M \obseq N$ if and only if $M \obsle N$ and $N \obsle M$. 
\end{defi}

Let us consider the restriction of the preorder $\obsle$ (see Definition \ref{def:obsle}) to the terms of the $\dzlam$-calculus (without tests).
 Theorem \ref{thm:main1} shows that for all terms $M,N$ of the $\dzlam$-calculus (without tests) we have
 $\Int{M}_{\seq x}\subseteq \Int{N}_{\seq x} \iff M \Tobsle N$. Later in this section (Theorem~\ref{thm:main1-1}) we will prove that
 $\Int{M}_{\seq x}\subseteq \Int{N}_{\seq x} \iff M\obsle N$. Hence the preorder $\obsle$ coincides, on the test-free language, with $\Tobsle$.
 This is an \emph{a fortiori} justification of Definition \ref{def:obsle-notest-nobangs}, which was anyway supported by the intuition that
 solvable $\dzlam$-terms are a kind of arenas over which the solvability game can be successfully played and simulated by the throw/catch game of the 
 test constructions.

\subsection{Full Abstraction via Test Expansion}
As mentioned in Section~\ref{sec:dzlam-calculus+tests}, the term $\gto(V)$ roughly corresponds to $V$ preceded by an infinite sequence of dummy $\lambda$-abstractions; 
dually, the test $\gt[L_1]\paral\cdots\paral\gt[L_k]$ corresponds to providing each $L_i$ with an infinite sequence of empty bags.
(This is also clear from the reduction rules $(\gt)$ and $(\gto)$.)
In this section we show that the infinite nature of these sequences is not essential in the \dzlam-calculus.
Roughly speaking, one can find an $n$ such that $\lam x_1\dots x_n.V$ has the same behaviour of $\gto(V)$
and $n_i$'s such that each $L_i[]^{\sim n_i}$ has the same behaviour of $\gt[L_i]$.
The parallel composition $V = V_1\paral\cdots\paral V_k$ can be simulated in the $\dzlam$-calculus by $M = \lam x.x[V_1,\dots,V_k]$
in the sense that $V$ converges iff each $V_i$ converges and, similarly, $M$ is solvable iff each $V_i$ is solvable.

We then define a test-expansion (Definition~\ref{def:labelling-expansion}), from terms of the \dzlam-calculus with tests to test-free terms,
formalizing this intuition.
In order to expand the correct number of times the occurrences of $\gto$ and of the elements of a test, we first 
need to ``name'' each occurrence in a different way.
For this reason we label such occurrences with pairwise distinct indices. 

\begin{defi}
A \emph{labelled} expression $A$ is an expression of the \dzlam-calculus with tests such that 
every occurrence of a $\gto$ and every element of a test have been decorated with distinct natural numbers (called \emph{indices}).
We denote by $(\Set{\gto})^\lab,(\Set{b})^\lab,(\Set{\gt})^\lab,(\Set{e})^\lab,(\Set{\gt}_{\hole{\cdot}})^\lab$ 
the set of labelled terms, labelled bags, labelled tests, labelled expressions, labelled term-contexts, respectively.
\end{defi}

Let $\sA\in\bool\langle(\Set{e})^\lab\rangle$ be a sum of labelled expressions. 
We write $\nak{\sA}$ for its underlying expression;
in other words $\nak{\sA}$ is obtained stripping off all indices from $\sA$.
We write $\dom(\sA)$ for the set of indices occurring in $\sA$.
Note that the domains of two summands $A,A'\in\sA$ may have a non-empty intersection. 

\begin{exa}\label{ex:labelled} \ 
\begin{enumerate}[1.]
\item\label{ex:labelled1} $M = \lam xy.\gto_{1}(\gt[(\bold{I})_2,(x)_3, (y[\gto_5(\gt[])])_4])$ is a labelled term. 
Its domain is $\dom(M) = \{ 1,2,3,4,5\}$ and its underlying term is $\nak M = \lam xy.\gto(\gt[\bold{I},x, y[\gto(\gt[])]])$.
\item\label{ex:labelled2} $\bold{I}$ is a labelled term because it does not contain any occurrence of $\gt$ nor $\gto$. 
In this case we have $\dom(\bold{I}) = \emptyset$ and $\nak{\bold{I}} = \bold{I}$.
\item\label{ex:labelled3} $\lam x.\gto_1(\gt[(x)_2,(x)_3,(\bold{I})_1])$ is \emph{not} a labelled term, because the labels of $\gto$ and of $\bold{I}$ are both~1 
(they are not distinct).
\item\label{ex:labelled4} $V = \gt[(\bold{I})_2,(\bold{D})_5,(\gto_{7}(\gt[(\bold{I})_{11}]))_{13}]$ is a labelled test. 
Its domain is $\dom(V) = \{2,5,7,11,13\}$ and its underlying test is $\nak V = \gt[\bold{I},\bold{D},\gto(\gt[\bold{I}]))]$.
\item\label{ex:labelled5} $P = [ \lam x.\gto_1(\gt[(x)_2]),\bold{D}]$ is a labelled bag.
Its domain is $\dom(P) = \{1,2\}$ and its underlying bag is $\nak P =  [\lam x.\gto(\gt[x]),\bold{D}]$.
\item\label{ex:labelled6} 
$\sM = \lam x.\gto_3(\gt[(x)_1]) + \lam x.\gto_3(\gt[(x)_1,(x)_2])$ is a sum of labelled terms. 
Its domain is $\dom(\sM) = \{1,2,3\}$ and its underlying sum of terms is $\nak {\sM} =  \lam x.\gto(\gt[x]) + \lam x.\gto(\gt[x,x])$.
\end{enumerate}
\end{exa}

\noindent From (\ref{ex:labelled2}) we note that $\nak A = A$ for all test-free labelled expressions.
From (\ref{ex:labelled5}) we note that in a labelled bag the labels actually occur \emph{within} its elements.

\begin{defi} 
The reduction semantics for labelled expressions is inherited straightforwardly from the $\dzlam$-calculus with tests. 
In the $\beta$-rule, the terms are substituted together with their indices.
\end{defi}

Since there is no duplication during the reduction, if $A$ is a labelled expression reducing to $\sA$ then $\sA$
is a sum of labelled expressions (that is, all the indices occurring within each $A\in\sA$ are pairwise distinct).

\begin{defi}\label{def:labelling-expansion}
Let $A\in(\Set{e})^\lab$ be a labelled expression and $\ell$ be a function from $\nat$ to $\nat$.
The \emph{$\ell$-expansion}  $A^\ell$ of $A$ is an expression of the \dzlam-calculus without tests, defined by induction on $A$ as follows:
\begin{gather*}
x^\ell = x,\qquad (\lam x.M)^\ell = \lam x.M^\ell,\qquad (MP)^\ell = M^\ell P^\ell,\\
[L_1,\dots,L_k]^\ell = [L_1^\ell,\dots,L_k^\ell],\qquad
\gto_i(V)^\ell = \lam x_1\ldots x_{\ell(i)}.V^\ell\textrm{ where $\seq x\notin \FV(V^\ell)$},\\
(\gt[(L_1)_{i_1},\dots,(L_k)_{i_k}])^\ell =  \lam x.x[L_1^\ell[]^{\sim\ell(i_1)},\dots,L_k^\ell[]^{\sim\ell(i_k)}]\textrm{ where $x\notin\FV(\seq L^\ell)$}.
\end{gather*}
In particular $\varepsilon^\ell = \lam x.x[]$ for all $\ell$.
This is extended to sums by setting $(\sum_i A_i)^\ell = \sum_i A_i^\ell$ and to contexts by setting $\hole{\cdot}^\ell = \hole{\cdot}$.
\end{defi} 

Obviously, for all test-free labelled expressions $A$ we have $A = A^\ell$ for all $\ell$.

\begin{rem}\label{rem:eq-dom-entails-eq} \
\begin{enumerate}[1.]
\item\label{rem:eq-dom-entails-eq1} 
	For all labelled test-contexts $C\hole{\cdot}\in(\ContSet)^\lab$ and labelled terms $M\in(\Set{\gto})^\lab$ we have $(C\hole{M})^\ell = C^\ell\hole{M^\ell}$ for any $\ell$,
\item\label{rem:eq-dom-entails-eq2} 
	For all labelled expressions $A\in(\Set{e})^\lab$, if $\ell\restr\!\dom{(A)} = \ell'\restr\! \dom{(A)}$ then $A^\ell = A^{\ell'}$ for any $\ell,\ell'$.
\item\label{rem:eq-dom-entails-eq3} 
	In general, $A\to B$ \emph{does not} entail $A^\ell\msto_\gb B^\ell$. 
	For instance, let $A = (\gt[(\gto_1[\varepsilon])_2,(\bold{I})_3])$ and $B=\gt[(\bold{I})_3]$; 
	obviously $A\to_\gamma B$ while, if $\ell_0$ is the constant function with value $0$, we have 
	$A^{\ell_0} = \lam x.x[\lam z.z[],\bold{I}]$, $ B^{\ell_0} = \lam x.x[\bold{I}]$ and $A^{\ell_0}\not\msto_\gb B^{\ell_0}$.
\end{enumerate}
\end{rem}

The proofs of the following lemmas are given in the technical Appendix~\ref{app:tech_app}.

\begin{lem}\label{lemma:exists-forall} 
Let $\sV\in\bool\langle(\Set{\gt})^\lab\rangle$ be a sum of labelled closed tests. 
If $\sV\msto\varepsilon$ then there exists a map $\ell:\nat\to\nat$ such that $\sV^{(\ell+k)}$ is solvable for all $k\in\nat$.
\end{lem}

\begin{lem}\label{lemma:forall-exists} 
Let $\sV\in\bool\langle(\Set{\gt})^\lab\rangle$ be a sum of labelled closed tests. 
If $\sV\msto 0$ then there exists a natural number $k$ such that $\sV^{(\ell+k)}\msto 0$ for all $\ell:\nat\to\nat$.
\end{lem}

We are now ready to state and prove the main theorem of this section, from which immediately follows the equational full abstraction result for the
 $\dzlam$-calculus.

\begin{thm}\label{thm:main1-1}
$\cD$ is inequationally fully abstract for the \dzlam-calculus (for all $M,N\in\Set{r}$):
$$
\Int{M}_{\seq x}\subseteq \Int{N}_{\seq x} \iff M\obsle N
$$
\end{thm}

\begin{proof} $(\Rightarrow)$
Assume that $\Int{M}_{\seq x}\subseteq\Int{N}_{\seq x}$, 
and let  $D\hole{\cdot}\in\Set{r}_{\hole{\cdot}}$ be a context closing both $M$ and 
$N$ and such that $D\hole{M}$ is solvable.
From the characterization of solvability given in Theorem~\ref{thm:equivalence-solvable-normalizable}
we have $D\hole{M}\msto M' + \sM$ for some $M'\in\Set{r}$ in normal form and $\sM\in\sums{r}$.
By Theorem~\ref{thm:Dmodel}, $\Int{D\hole{M}} = \Int{M'}\cup\Int{\sM}$ which is different from $\emptyset$ by Theorem~\ref{thm:equivalence-solvable-normalizable}.
By monotonicity of the interpretation we get $\Int{D\hole{M}}\subseteq\Int{D\hole{N}}$,
thus $\Int{D\hole{N}}\neq\emptyset$. 
By Theorem~\ref{thm:equivalence-solvable-normalizable} this entails that $D\hole{N}$ is solvable too.

$(\Leftarrow)$ Suppose, by the way of contradiction, that $M\obsle N$ holds but 
there is an element $(\seq a,\alpha)\in\Int{M}_{\seq x}-\Int{N}_{\seq x}$.
By Proposition~\ref{prop:context} the test-context $C\hole{\cdot} = \cont{\alpha}{(\lam\seq x.\hole{\cdot})\trm{\seq{a}}}$
is such that $C\hole{M}\msto \cont{\alpha}{M\lsubst{\seq{x}}{\trm{\seq{a}}}}\msto \varepsilon$
and $C\hole{N}\not\msto \varepsilon$ (therefore $C\hole{N}\msto 0$ by Lemma~\ref{lem:closed-go-to-eps-zero}). 
Let $C'\hole{\cdot}\in(\Set{r}_{\hole{\cdot}})^\lab$ such that $\nak C'\hole{\cdot} = C\hole{\cdot}$.
By Lemma~\ref{lemma:exists-forall} there exists $\ell$ such that $(C'\hole{M})^{(\ell + k')}$ is solvable for every $k'\in\nat$. 
By Lemma~\ref{lemma:forall-exists} there exists $k\in\nat$ such that $(C'\hole{N})^{(\ell+k)}$ is unsolvable.
From Remark~\ref{rem:eq-dom-entails-eq}(\ref{rem:eq-dom-entails-eq1}) we get $(C'\hole{M})^{(\ell + k)} = C'^{(\ell + k)}\hole{M^{(\ell + k)}}$
and $(C'\hole{N})^{(\ell + k)} = C'^{(\ell + k)}\hole{N^{(\ell + k)}}$.
Since $M,N$ are test-free we have $M^{(\ell + k)} = M$ and $N^{(\ell + k)} = N$.
We conclude because we found a term-context $C'^{(\ell + k)}$ such that
$C'^{(\ell + k)}\hole{M}$ is solvable and $C'^{(\ell + k)}\hole{N}$ is unsolvable,
which is a contradiction.
\end{proof}

\begin{cor}\label{cor:main1-1}
$\cD$ is equationally fully abstract for the \dzlam-calculus.
\end{cor}

\begin{rem} A direct proof of Corollary~\ref{cor:main1-1} might be obtained exploiting a corollary of the B\"ohm Theorem for 
the $\dlam$-calculus proved in \cite{ManzonettoP11}.
We preferred to provide this proof based on test-expansion because it clarifies the behaviours of our test operators and works also in the inequational case.
\end{rem}

The rest of the paper is devoted to extend the full abstraction results of Subsection~\ref{subsec:FAres1} to the
\dlam-calculus with tests.
The main ingredients will be the head reduction introduced in Subsection~\ref{subsec:HeadRed} and
the Taylor expansion we define in Subsection~\ref{subsec:TaylorExp}.

\section{The \dlam-Calculus with Tests}\label{sec:FullRC}

The \emph{\dlam-calculus with tests} is an extension of the 
\dzlam-calculus with tests with a promotion operator available on resources. 
In this calculus a resource can be linear (it must be used exactly once) or not 
(it can be used \emph{ad libitum}) and in the latter case it is decorated with a ``$\bang$'' superscript.

\begin{figure}[t]
\centering
\textbf{The  \dlam-calculus with tests}\\[2ex]
\subfigure[Grammar of terms, bags, tests, expressions, sums.]{\label{fig:Fgrammar}
    \begin{tabular}{@{}p{.1\linewidth}@{}p{.125\linewidth}@{}p{.525\linewidth}@{}p{.15\linewidth}@{}}
    $\FSet{\gto}$:       &$M,N,L,H$         &$\gramm x \mid \lambda x.M \mid MP \mid \gto(V)$        &\hfill terms\\[3pt]
    $\FSet{b}$:       &$P,Q$         &$\gramm [L_1,\ldots,L_k,\sM^\oc]$        &\hfill bags\\[3pt]    
    $\FSet{\gt}$:       &$V,W$         &$\gramm \gt[L_1,\ldots,L_k]$        &\hfill tests\\[3pt]        
    $\FSet{e}$:       &$A,B$         &$\gramm M\mid V\mid P$        &\hfill expressions\\[3pt]     
    \\   
    \hline
    \\
    \multicolumn{3}{p{.80\linewidth}}{
    \hspace{-6pt}$ \sM, \sN,\sL,\sH\in \Fsums{\gto}$} &sums of terms\\[3pt]
    \multicolumn{3}{p{.80\linewidth}}{
    \hspace{-6pt}$\sP,\sQ\in \Fsums{b}$} &\hfill sums of bags\\[3pt]    
    \multicolumn{3}{p{.80\linewidth}}{
    \hspace{-6pt}$\sV,\sW\in\Fsums{\gt}$} &\hfill sums of tests\\[3pt]        
    \multicolumn{3}{p{.80\linewidth}}{
    \hspace{-6pt}$ \sA,\sB\in \Fsums{e} :=\Fsums{\gt}\cup\Fsums{\gto}\cup\Fsums{b}$} &\hfill sums of expr.\\[3pt]        
    \end{tabular}
}
\\[3pt]
\hrulefill\\[4pt]
\subfigure[Linear Substitution.]{
\label{fig:LinSubstBang}
\begin{minipage}{\linewidth}
$$
	\textrm{\bf Linear Substitution (New Rule)}
$$
\begin{align*}
[L_1,\pts,L_k,\sN^\bang]\lsubst{x}{N}= &\ \Sigma_{i=1}^{k}[L_1,\pts,L_i\lsubst{x}{N},\pts,L_k,\sN^\bang] +\ [L_1,\pts,L_k,\sN\lsubst{x}{N},\sN^\bang]\,.
\end{align*}
\vspace{-3pt}
\end{minipage}
}
\\[3pt]
\hrulefill\\[2ex]
\subfigure[Reduction rules. In the $(\beta)$ rule we assume wlog $x\not\in\FV(\bag{L_1,\dots,L_k})$.]{
\label{fig:newredrules}
\begin{minipage}{\linewidth}
\centering
\textbf{Reduction Semantics (New Rules)}
\begin{gather*}
(\lam x.M)[L_1,\pts,L_k,\sN^{\bang}] \to_{\beta} M\lsubst{x}{[L_1,\pts,L_k]}\subst{x}{\sN},\\[1ex]
\gto(V)[L_1,\ldots,L_k,\sN^{\bang}]\to_{\gto} \left\{
\begin{array}{ll}
\gto(V)  & \textrm{if $k = 0$,} \\
0  & \textrm{otherwise.} \\[1ex]
\end{array}
\right.
\end{gather*}

\end{minipage}
}
\\[3pt]
\hrulefill\\[4pt]
\centering
\textbf{Context Closure (New Rule)}\\[-0.5ex]
\subfigure[Context closure of a relation $\texttt{R}\subseteq \FSet{e}\times\Fsums{e}$.]{\label{fig:new_context_closure}
\begin{minipage}{\linewidth}
\begin{gather*}
\infer[\mathrm{bangres}]{\bag{(M+\sN)^\oc} \mcup \sP \rel{R} \bag{(\sM+\sN)^\oc}\mcup \sP}{M\rel{R} \sM}
\end{gather*}
\end{minipage}
}
\caption{\footnotesize Syntax, notations and reduction semantics of the \dlam-calculus with tests.}
\label{fig:Fstatics}
\end{figure}

\subsection{Syntax.} 
The grammar generating the terms, the tests and the expressions of the \dlam-calculus with tests,
is given in Figure~\ref{fig:Fgrammar}.
Note that such grammar is equal to the one for the \dzlam-calculus with tests (in particular tests are still plain multisets of linear resources), 
except for the rule concerning bags which  becomes: 
$$
	P \gramm [L_1,\ldots,L_k,\sN^\bang] \hspace{250pt}\text{bags}
$$
where  $\sN$ is a finite sum of terms of this new syntax.
We write 
$\FSet{\gto}$ for the set of terms generated by this new grammar,
$\FSet{\gt}$ for the set of tests, 
$\FSet{b}$ for the set of bags,
$\FSet{e}$  for the set of  expressions.

It should be clear that from now on bags are no more plain multisets
of terms: they are compound objects, consisting of a multiset of terms
$[L_1,\ldots,L_k]$ and a sum of terms $\sN$, denoted as 
$[L_1,\ldots,L_k,\sN^\bang]$. 
We shall deal with them as if they were multisets, defining union by 
 $[L_1,\ldots,L_k,\sN^\bang]\mcup [L_{k+1},\ldots,L_n,\sM^\bang]\ass
[L_1,\ldots,L_n,(\sN+\sM)^\bang]$. This operation is commutative,
associative and has $[0^\bang]$ as neutral element.

\begin{rem}
The \dzlam-calculus with tests is the sub-calculus of the \dlam-calculus with tests in which all bags have the shape $[L_1,\ldots,L_k,0^\bang]$, and this
 identification is compatible with the reduction rules.
\end{rem}

As in the  \dzlam-calculus with tests, we extend this syntax by multilinearity to sums of 
 expressions with the only exception that the bag $[L_1,\ldots,L_k,(\sN+\sM)^\bang]$ {\em is not} required to be equal to 
 $[L_1,\ldots,L_k,\sN^\bang]+ [L_1,\ldots,L_k,\sM^\bang]$. The intuition is that in the first expression $\sN+\sM$ can be used
 several times and each time one can choose non-deterministically $\sN$ or $\sM$, whereas in the second expression one has to choose once and for all 
 one of the summands, and then use it as many times as needed.

\subsection{Substitutions.}
Linear substitution is denoted and defined as in the \dzlam-calculus with tests (Figure~\ref{fig:linsubst}), except of course for bags, 
where we use the rule of Figure~\ref{fig:LinSubstBang}.
Linear substitution is extended to sums, as in $\sA\lsubst{x}{\sN}$, by bilinearity in both $\sA$ and $\sN$.

\begin{rem}
In the $\bang$-free case, that is when $\sN=0$, the above definitions and notations agree with those introduced in Subsection~\ref{subs:Subst},
because in that case we have 
$$[L_1,\ldots,L_k,\sN\lsubst{x}{N},\sN^\bang]=0,$$ 
since $0\lsubst{x}{N}=0$.
\end{rem}

We also define the regular substitution $A\subst{x}{\sN}$ 
for the \dlam-calculus with tests, by simply replacing each occurrence of 
$x$ in the expression $A$ with $\sN$ --- in that way we get an expression of the extended
syntax, since $\sN$ is a sum in general. 
This operation is then extended to sums, as in $\sA\subst{x}{\sN}$, by linearity in $\sA$.

\begin{exa} \ 
\begin{enumerate}[1.]
\item $(y[x^\bang])\lsubst{x}{z} = y[z,x^\bang] $,
\item $(y[x,0^\bang])\lsubst{x}{z} = y[z,0^\bang] $,
\item $(y[(x+z)^\bang])\lsubst{x}{z} = z[z,(x+z)^!]$,
\item $(y[x,x^\bang])\lsubst{x}{z} = y[z,x^\bang] + y[x,z,x^!]$,
\item $(x[x^\bang])\lsubst{x}{y}\lsubst{x}{z} = y[z,x^\bang] + z[y,x^\bang] + x[y,z,x^\oc]$,
\item $(x[x^\bang])\lsubst{x}{(y+z)} = y[x^!] + z[x^!] + x[y,x^!]+ x[z,x^!]$,
\item $(y[x,x^\oc])\subst{x}{(y+z)} = y[y,(y+z)^\oc] + y[z,(y+z)^\oc]$,
\item $(x[x^\oc])\subst{x}{(y+z)} = y[(y+z)^\oc] + z[(y+z)^\oc]$,
\item $(x[x,x^\oc])\subst{x}{(y+z)} = y[y,(y+z)^\oc] + y[z,(y+z)^\oc] + z[y,(y+z)^\oc] + z[z,(y+z)^\oc]$.
\end{enumerate}
\end{exa}

A Schwarz Theorem, analogous to Theorem~\ref{thm:Schwarz}, holds for the \dlam-calculus with tests. 
Hence, given a sum of expressions $\sA$ and a bag $P = [L_1,\ldots,L_k]$ with $x\notin\FV(P)$, 
it still makes sense to set $\sA\lsubst{x}{P}\ass\sA\lsubst{x}{L_1}\cdots \lsubst{x}{L_k}$ because 
this expression does not depend on the enumeration of $L_1,\ldots,L_k$ in $P$. 
In particular $\sA\lsubst{x}{[]} = \sA$.

\subsection{Operational semantics.} 
The reduction rules of $ \dlam$-cal\-cu\-lus extend those of the \dzlam-calculus with tests in the sense that they 
are equivalent on $\bang$-free expressions.

\begin{defi}
The rules $(\tau)$ and  $(\gamma)$ are exactly the same as the corresponding rules of the \dzlam-calculus,
while the $\beta$-reduction and $\gto$-reduction are rephrased as in Figure~\ref{fig:newredrules}.
In this setting, their contextual closure needs to be closed also under the rule of Figure~\ref{fig:new_context_closure}.
\end{defi}

The \dlam-calculus with tests is still Church-Rosser (just adapt the proof in \cite{PaganiT09}), while it is no more strongly normalizing. 
For instance the term $\Omega \ass \Delta[\Delta^\bang]$, for $\Delta \ass \lam x.x[x^\bang]$, has an infinite reduction chain, just like the paradigmatic homonymous unsolvable \lam-term.
Indeed, the usual \lam-calculus can be embedded into the \dlam-calculus with tests by
translating every application $MN$ into $M[N^{\bang}]$.

\begin{rem}
Reductions in the  \dlam-calculus with tests may be  tricky, due to the combination of linear and non linear resources and substitutions.
For instance, we can obtain eight $\Omega$-like terms of the \dlam-calculus with tests, of the form $M[N^{(!)}]$ where $M,N\in \{\bold{D},\Delta\}$ and $(!)$
 denotes the optional presence of the promotion. Not surprisingly all these terms, except for $\Omega$, reduce to $0$. E.g., 
 $\bold{D}[\Delta^!]\to_{\beta} \Delta[\Delta] \to_{\beta} \Delta[\Delta,0^!] \to_{\beta} \Delta[0^!] + 0[\Delta,0^!] \to_{\beta} 0$.
\end{rem}

Here are some other examples of reductions, involving tests.

\begin{exa} \
\begin{enumerate}[1.]
\item $(\lambda x. \gto(\varepsilon)[x^!])[\bold{I},0^!]\to_{\beta}
	\gto(\varepsilon)[x^!]\lsubst{x}{[\bold{I}]}\subst{x}{0}=
	\gto(\varepsilon)[\bold{I},0^!]\to_{\gto } 0,$
	\item $(\lambda x. \gto(\varepsilon)[x^!])[\bold{I},0^!]\to_{\gto }
	(\lambda x. \gto(\varepsilon))[\bold{I},0^!]\to_{\beta}
	\gto(\varepsilon)\lsubst{x}{[\bold{I}]}\subst{x}{0}=0,$
\item $(\lambda x. \gto(\varepsilon)[x^!])[\bold{I}^!]\to_{\beta}
	\gto(\varepsilon)[x^!]\subst{x}{\bold{I}}=
	\gto(\varepsilon)[\bold{I}^!]\to_{\gto } \gto(\varepsilon),$
\item $(\lambda x. \gto(\varepsilon)[x^!])[\bold{I}^!]\to_{\gto }
	(\lambda x. \gto(\varepsilon))[\bold{I}^!]\to_{\beta}
	\gto(\varepsilon)\subst{x}{\bold{I}}= \gto(\varepsilon).$
\end{enumerate}
\end{exa}

In this framework a \emph{test-context} $C\hole{\cdot}$ (resp.\ \emph{term-context} $D\hole{\cdot}$) is a test (resp.\ term) of the \dlam-calculus with tests 
having a single occurrence of its \emph{hole}, appearing in term-position.

\begin{defi} Term-contexts $D\hole{\cdot}$ and test-contexts $C\hole{\cdot}$ are defined by the following grammar:
\begin{gather*}
	D\hole{\cdot}\gramm \hole{\cdot} \mid \lam x.D \mid DP  \mid M[D,\seq L,\sN^\bang]\mid M[\seq L,(D + \sN)^{\bang}] \mid \gto(C)\\
	C\hole{\cdot}\gramm \gt[D,\seq L]
\end{gather*}
The set of term-contexts is denoted by $\Set{\gto\bang}_{\hole{\cdot}}$ 
and the set of test-contexts is denoted by $\FContSet$.
\end{defi}

\begin{defi}\label{def:convergence-with-promotion}
A test $V$ \emph{converges}, notation $\cOnv{V}$, if there exists a (possibly empty) sum $\sV$ such that $V \msto \varepsilon + \sV$.
\end{defi}

Convergence should not be confused with normalization. Note that Definition \ref{def:convergence-with-promotion} is the natural extension of Definition \ref{def:convergence-no-promotion};
 in presence of promotion, $\varepsilon$ and $0$ are not the only possible ``outcomes'' of closed tests because there are looping terms that may never interact with an outer cork $\gt[\cdot]$.
 That case represents ``failure'', i.e., a scenario where there is no possible sequence of choices (among summands of terms resulting from reduction)
 leading to the positive test $\varepsilon$.

\begin{defi} The \emph{operational pre-order} $\Fobsle$ on the \dlam-calculus with tests is defined by:
$$
	M\Fobsle N \iff \forall C\hole{\cdot}\in\FContSet\,\textrm{ closing $M,N$ }(\cOnv{C\hole{M}}\ \imp \cOnv{C\hole{N}}).
$$
We then set $M\Fobseq N$ iff $M\Fobsle N$ and $N\Fobsle M$.
\end{defi}

\subsection{Relational semantics} 

The \dlam-calculus with tests can be interpreted into $\cD$ by extending the interpretation of 
the \dzlam-calculus with tests given in Subsection~\ref{subsec:dzlam-interpretation} as follows:
$$
\hspace{-1pt}
\begin{array}{l}
\Int{[L_1,\dots,L_k,\sN^{\bang}]}_{\seq x} = \{(\mcup_{r=1}^{k+m} \seq a_r,[\beta_1,\ldots,\beta_{k + m}]) \st (\seq a_j,\beta_j)\in\Int{L_j}_{\seq x},1\le j\le k \textrm{ and } \\
\qquad\qquad\qquad\qquad\qquad\qquad\qquad\qquad\qquad\qquad\quad\ \!  (\seq a_i,\beta_i)\in \Int{\sN}_{\seq x},\ k < i\le k + m,\ m\geq 0\}.
\end{array}
$$
The following are examples of interpretations.
As shown in (\ref{ex:interp-bang4}) below, interpreting the non linear resource $\sN^{\bang}$  in $P=[\seq L,\sN^{\bang}]$ boils down to choosing an arbitrary number of elements of  $\Int{\sN}$, whereas exactly one element of each $\Int{L_i}$ is required.

\begin{exa}\label{ex:interp-bang}\
\begin{enumerate}[1.]
\item\label{ex:interp-bang1} 
	$\Int{\Delta}=\{ ([a::\alpha]\mcup a)::\alpha\st a\in\Mfin\cD ,\ \alpha \in {\cD} \}$.
\item\label{ex:interp-bang2} 
	$\Int{\Omega} = \emptyset$. Indeed, $\alpha\in\Int{\Omega}$ iff there exist $k > 0$ and $ \alpha_1,\dots,\alpha_k\in\Int{\Delta}$ such that $[\alpha_1,\dots,\alpha_k] \at \alpha\in\Int{\Delta}$.
	We prove by induction on $k$ that this is never satisfied.
	If $k = 1$ then by (\ref{ex:interp-bang1}) $[\alpha_1] :: \alpha\in\Int{\Delta}$ entails $\alpha_1 = []\at \alpha$ and it is easy to check that $\alpha_1\notin\Int{\Delta}$. 
	Otherwise, if $k>1$, $[\alpha_1,\dots,\alpha_k]\at \alpha\in\Int{\Delta}$ and each $\alpha_i\in\Int{\Delta}$ then, for some $j$, 
	$\alpha_j = [\alpha_1,\dots,\alpha_{j-1},\alpha_{j+1},\dots,\alpha_{k}]\at \alpha\in\Int{\Delta}$ and we are done by induction hypothesis.
\item\label{ex:interp-bang3} $\Int{\bold{D}[\bold{I}^!]}=
\{ [\gamma]::\gamma\st\gamma\in {\cD} \} $ since $ [\alpha]::\beta,\alpha\in\Int{\bold{I}}$ entails $\beta=\alpha=[\gamma]::\gamma$, for some $\gamma\in\cD$;
conversely, for any  $\gamma\in\cD$, both  $[\gamma]::\gamma$ and $[[\gamma]::\gamma]::[\gamma]::\gamma$ belong to $\Int{\bold{I}}$.
\item\label{ex:interp-bang4} $\Int{[\bold{D},\bold{I}^!]}= \{[[[\alpha]::\beta,\alpha]::\beta,
[\alpha_1]::\alpha_1,\ldots, [\alpha_k]::\alpha_k]]\st \alpha,\beta,\alpha_1,\ldots,\alpha_k\in \cD\textrm{ for }k\in\nat\}$.
\end{enumerate}
\end{exa}
The comparison between $\Int{\bold{D}}$ (Example~\ref{ex:interp}(\ref{ex:interp3})) and $\Int{\Delta}$ (item (\ref{ex:interp-bang1})) gives a grasp on the semantic counterpart of non-linearity.

It is easy to check that both the linear and the classic substitution lemmas generalize to this context.
While we can keep the same statement for Lemma~\ref{lemma:subst}, Lemma~\ref{lemma:lsubst} must be rephrased as follows
(indeed, $\dg{x}{M}, \dg{x}{V}$ are undefined when $M,V$ contain non linear resources).

\begin{lem}[Linear Substitution Lemma]
Let $M,L_1,\dots,L_k\in\FSet{\gto}, Q\in\FSet{\gt}$ and $P = [L_1,\ldots,L_k]$ (with $y \not\in \FV(P)$). Then we have:
\begin{enumerate}[\em(i)]
\item 
$((\seq a,b),\alpha)\in \Int{M\lsubst{y}{P}}_{\seq x,y}$ iff there exist
$(\seq a_i,\beta_i)\in\Int{L_i}_{\seq x}$ (for $i=1,\ldots,k$) and
$\seq a_0\in\Mfin{\cD}^n$ and $b \in\Mfin{\cD}$ such that
$((\seq a_0,[\beta_1,\ldots,\beta_k] \mcup b), \alpha)\in\Int{M}_{\seq x,y}$
and 
$\mcup_{i=0}^{k} \seq a_i=\seq a$.
\item 
$(\seq a,b) \in \Int{Q\lsubst{y}{P}}_{\seq x,y}$  iff there exist
$(\seq a_i,\beta_i)\in\Int{L_i}_{\seq x}$ (for $i=1,\ldots,k$) and
$\seq a_0\in\Mfin{\cD}^n$ and $b \in\Mfin{\cD}$ such that
$(\seq a_0,[\beta_1,\ldots,\beta_k]\mcup b)\in\Int{Q}_{\seq x,y}$ and 
$\mcup_{i=0}^{k} \seq a_i=\seq a$.\qed
\end{enumerate}
\end{lem}

\noindent From these lemmas it ensues that $\cD$ is also a model of the \dlam-calculus with tests.

\begin{thm}
$\cD$ is a  model of \dlam-calculus with tests.\qed
\end{thm}

\subsection{Head Reduction}\label{subsec:HeadRed}

We now provide a notion of \emph{head reduction} for the \dlam-calculus with tests.
Intuitively, head reduction is obtained by reducing a head redex, 
that is a redex occurring in head position in an expression $A$.
The main interest of introducing this reduction strategy is that 
it ``behaves well'' with respect to Taylor expansion 
in the sense of Proposition~\ref{prop:magic}.

The definition of term- and test-redexes is inherited from Definition~\ref{def:redex}.
Among these redexes we distinguish those that are in ``head'' position. 

\begin{defi}\label{def:head-redex}
A \emph{head redex} is defined inductively as follows: 
\begin{iteMize}{$\bullet$} 
\item[-]
        every test-redex $V$ is a head redex,
\item[-] 
	a term-redex $H$ is a head redex in both the term $\lam\seq y.H\seq P$ and the test $\gt[H\seq P]\paral V$.
\end{iteMize}
\end{defi}

\begin{defi} 
We say that $A\to \sB$ is a step of \emph{head reduction} if $\sB$ is obtained from $A$ by contracting a head redex.
If $A\to \sB$ is a step of head reduction then also $A+\sA\to \sB + \sA$ is.
\end{defi}

One-step head reduction is denoted by $\toh$, while $\mstoh$ indicates its reflexive and transitive closure.

\begin{rem}
Unlike in ordinary \lam-calculus, an expression $A$ may have more than one head redex, hence there may be more than one head reduction step starting from $A$.
\end{rem}

Head reduction induces a notion of head normal form on (sums of) expressions.

\begin{defi}
An expression $A$ \emph{is in head normal form} (\emph{hnf}, for short) if there is no $\sB$ such that $A\toh \sB$; 
a sum $\sA$ \emph{is in hnf} if each summand is in hnf.
\end{defi}

This notion of head normal form differs from that given by Pagani and Ronchi Della Rocca in \cite{PaganiR10}.
We keep this name because their definition captures the notion of ``outer-normal form'' rather than that of head normal form,
and in fact they changed terminology in \cite{PaganiR11}.

The following lemma gives a characterization of terms and tests in head normal form.

\begin{lem}\label{lemma:char-hnf}\hfill
\begin{enumerate}[\em(i)]
\item A term $M$ is in head normal form if and only if either $M\ass \lam\seq x.y\seq P$ or $M\ass \lam\seq x.\gto(V)$.
\item A test $V$ is in head normal form if and only if $V\ass \gt[x_1\seq P_1,\dots,x_n\seq P_n]$ for $n\ge 0$.
\end{enumerate}
\end{lem}

\begin{proof} By a simple inspection of the shape of head redexes (Definition~\ref{def:head-redex}).
\end{proof}

The following two lemmas concern reduction properties of \emph{promotion-free closed} tests.

\begin{lem}\label{lemma:RhasHR} 
Let $V\in\Set{\gt}$. If $V$ is closed and $V\neq \varepsilon$ then it has a head redex
(hence, $V\toh \sV$ for some $\sV$).
\end{lem}

\begin{proof} 
By structural induction on $V$.
It suffices to consider the case 
$V = \gt[M]$. 
We then proceed by cases on the structure of $M$ (which must be closed).
If $M = \lam x.N$ then $V$ head reduces using $(\gt)$.
If $M$ is an application then it must be written either as $M = (\lambda y.N)P_1\cdots P_k$ 
or as $M=\gto[W]P_1\cdots P_k$ (in both cases $k\geq 1$) and hence $V$ head reduces using either $(\gb)$ or $(\gto)$, respectively.
If $M = \gto(W)$ then $V$ head reduces using $(\gamma)$.
\end{proof}

\begin{lem}\label{lemma:RgoeiffRgohe} 
If $V\in\Set{\gt}$ is closed then $V\msto \varepsilon$ (resp. $V\msto 0$) if and only if $V\mstoh \varepsilon$ (resp.\ $V\mstoh 0$).
\end{lem}

\begin{proof} $(\Rightarrow)$ 
Suppose, by contradiction, that $V\msto\varepsilon$ but $V\not\mstoh\varepsilon$. 
By confluence (Theorem~\ref{thm:CR+SN}), we cannot have $V\mstoh 0$. 
Thus, since $V\in\Set{\gt}$ is strongly normalizing, the only way to have $V\not\mstoh\varepsilon$ is that $V\mstoh \sV$ where $\sV\neq\varepsilon,0$ is in hnf.
This is impossible by Lemma~\ref{lemma:RhasHR}.

An analogous proof shows that $V\msto 0$ entails $V\mstoh 0$.

$(\Leftarrow)$ Trivial since $\mstoh\ \subseteq\ \msto$.
\end{proof}

\begin{rem}
One should be careful when trying to extend the above result to terms $M\in\Set{\gto}$.
For instance, it is false that $M\msto 0$ if and only if $M\mstoh 0$ as shown by this easy counterexample:
the term $M \ass \lam x.x[\bold{I}[]]$ is in hnf but $M\to_\beta \lam x.x[0] = 0$.
\end{rem}

Head reduction will play an essential role in the next section.

\section{Full Abstraction via Taylor Expansion}\label{sec:TaylorExp}

In this section we are going to define the Taylor expansion of terms 
and tests of the \dlam-calculus with tests.
We will then use this expansion, combined with head-reduction, 
to generalize the full abstraction results obtained in Subsection~\ref{subsec:FAres1} 
to the framework of \dlam-calculus with tests.

\subsection{Taylor Expansion}\label{subsec:TaylorExp}

The \emph{(full) Taylor expansion} was first introduced in \cite{EhrhardR03,EhrhardR06}, in the context of \lam-calculus.
The Taylor expansion $\TE{M}$ of an ordinary \lam-term $M$ gives an infinite formal linear combination of terms
(equivalently, a set of terms) of the \dzlam-calculus. In the case of ordinary application it looks like:
$$
	\TE{(MN)} =\sum_{n=0}^\infty \frac{1}{n!} M[\underbrace{N,\ldots,N}_{n\textrm{ times}}]
$$ 
in accordance with the intended meaning and the denotational semantics
of application in the resource calculus.
In the syntax of Ehrhard-Regnier's differential \lam-calculus the above formula looks like 
$\sum_{n=0}^\infty \frac{1}{n!} M^{(n)}(0)(N,\ldots,N)$,
hence the connection with analytical Taylor expansion is evident.

Following \cite{ManzonettoP11}, we extend the definition of Taylor expansion 
from ordinary \lam-terms to expressions of the \dlam-calculus with tests.
Since in our context the sum is idempotent, the coefficients disappear and our Taylor expansion 
corresponds to the \emph{support} of the actual Taylor expansion.

As the set $\sums{e}_{\infty}$ of possibly infinite formal sums of expressions is isomorphic to $\Pow{\Set{e}}$, 
in the following we feel free of using sets instead of sums.

\begin{defi}\label{def:Taylor}
Let $\sA\in\Fsums{e}$. 
The \emph{(full) Taylor expansion of $\sA$} is the set $\TE{\sA}\subseteq\Set{e}$ 
which is defined (by structural induction on $\sA$) in Figure~\ref{fig:TE}.%
\begin{figure*}[!t]
$$
\begin{array}{rcl}
\TE{x} &=& \{x\}, \\
\TE{(\lam x.M)} &=& \{\lam x.M'\st M'\in\TE{M}\},\\
\TE{(MP)} &=& \{M'P'\st M'\in\TE{M},\ P'\in\TE{P} \},\\
\TE{(\gto(V))} &=& \{ \gto(V')\st V'\in\TE{V}\}\\
\TE{(\gt[M_1,\ldots,M_k])} &=& \{ \gt[M_1',\ldots,M_k']\st M_i'\in\TE{M_i},\textrm{ for } 1 \leq i \leq k\},\\
\TE{[L_1,\ldots,L_k,\sN^\bang]} &=& \{[L_1',\ldots,L_k']\mcup P \st L'_i\in\TE{L_i},\textrm{ for } 1\le i\le k,\ P\in \Mfin{\TE{\sN}} \},\\
\TE{(\Sigma_{i=1}^k A_i)} &=& \cup_{i=1}^k \TE{A_i}.\\
\end{array}
$$
\caption{\textrm{The \emph{Taylor expansion} $\TE{\sA}$ of $\sA\in\module{\FSet e}$.}}\label{fig:TE}
\end{figure*}
%

\end{defi}

The following are examples of Taylor expansion of terms and tests.

\begin{exa}\label{ex:TE}\
\begin{enumerate}[1.]
\item\label{ex:TE1} $\TE{(\lam x.x[x^{\bang}])} = \{ \lam x.x[\underbrace{x,\dots,x}_n] \st n\ge 0\}$,
\item\label{ex:TE2} $\TE{(\gt[\lam x.x[x,x^{\bang}]])} = \{ \gt[\lam x.x[\underbrace{x,\dots,x}_n]] \st n>0\}$,
\item\label{ex:TE3} $\TE{(x[(z[y^{\bang}])^\bang])} = \TE{(x[(z[] + z[y,y^\bang])^\bang])} = \{x[
\overbrace{
z[\underbrace{y,\dots,y}_{n_1}],
\dots,
z[\underbrace{y,\dots,y}_{n_k}]
}^{k}
] \st k,n_1,\dots,n_k\ge 0\}.$
\end{enumerate}
\end{exa}

In (\ref{ex:TE1}) and (\ref{ex:TE2}) we see that the Taylor expansion of an expression $A$ can be infinite.
In (\ref{ex:TE3}) we have an example of two different terms sharing the same Taylor expansion.

\begin{conv} To lighten the notations, we will adopt for infinite sets of expressions the same 
abbreviations as introduced for finite sums in Subsection~\ref{subs:Subst} (including those for substitutions).
\end{conv}
For instance, if $X,Y\subseteq \Set{\gto}$ then $\lam x.X$ denotes the set $\{\lam x.M'\st M'\in X\}$
and $X\lsubst{x}{Y}=\cup_{M\in X,N\in Y} M\lsubst{x}{N}$.

In \cite{Manzonetto10} it is proved that the Taylor formula holds in $\MRel$.
This property entails that Taylor expansion preserves the meaning of an
expression in $\cD$, as expressed in the next theorem.

\begin{thm}\label{thm:DmodelsTE}
$\Int{\sA}_{\seq x} = \cup_{A\in\TE{\sA}}\Int{A}_{\seq x}$, for all $\sA\in\Fsums{e}$.
\end{thm}

\begin{proof}
By adapting the proof in \cite{Manzonetto10} of the analogous theorem for the differential $\lam$-calculus.
\end{proof}

We now need the following technical lemma stating the commutation of 
Taylor expansion with respect to ordinary and linear substitutions.
The proof is lengthy but not difficult and is provided in Appendix~\ref{app:tech_app}.
For the sake of readability, in the next statements we use sums and unions interchangeably.

\begin{lem}\label{lemma:TEcommutesubsts} 
Let $A\in\FSet{e}$, $N\in\FSet{\gto}$ and $\sN\in\Fsums{\gto}$. Then, for $x\notin\FV(N)\cup\FV(\sN)$:
\begin{enumerate}[\em(i)]
\item\label{lemma:TEcommutesubsts1} 
	$\TE{(A\lsubst{x}{N})} = \TE{A}\lsubst{x}{\TE{N}}$,
\item\label{lemma:TEcommutesubsts2} 
	$\TE{(A\subst{x}{\sN})} = \bigcup_{P\in\Mfin{\TE{\sN}}}\TE{A}\lsubst{x}{P}\subst{x}{0}$.
\end{enumerate}
\end{lem}

The next proposition is devoted to show how Taylor expansion interacts with
head-reduction. To ease the formulation of the next proposition we assimilate $\Fsums{e}$ to $\Powf{\FSet{e}}$.

\begin{prop}\label{prop:magic}
Let $A\in\FSet{e}$ and let $A'\in\TE{A}$ be such that $A'\toh\sB'$, for some $\sB'$. 
Then there exists $\sB$ such that $A\toh\sB$ and $\sB'\subseteq\TE{\sB}$.
\end{prop} 

\begin{proof} 
The idea is that the syntactic tree of $A$ has the same structure as that of $A'$
and we can define a surjective mapping of the redexes of $A'$ into those of $A$.

We only treat the case $A' = \lam\seq x.H'P'_1\cdots P'_p$ where $H' = (\lam y.M')P'$ is a head-redex.
From $A'\in\TE{A}$ we get $A = \lam\seq x.HP_1\cdots P_p$ for some $H$ such that $H'\in\TE{H}$.
Hence, supposing wlog $P'=[\seq L',\seq N']$, we have that $H = (\lam y.M)[\seq L,\sN^\bang]$ where $M'\in\TE{M}$, 
the lengths of $\seq L'$ and $\seq L$ coincide, $L'_i\in\TE{L_i}$ for all $i$ and $[\seq N']\in\Mfin{\TE{\sN}}$.
We now know that $H'\toh M'\lsubst{y}{[\seq L']}\lsubst{y}{[\seq N']}\subst{y}{0}$ and
$H\toh M\lsubst{y}{[\seq L]}\subst{y}{\sN}$.
By Lemma~\ref{lemma:TEcommutesubsts}, 
$\TE{(M\lsubst{y}{[\seq L]}\subst{y}{\sN})} = \cup_{P\in\Mfin{\TE{\sN}}}\TE{M}\lsubst{y}{[\TE{\seq L}]}\lsubst{y}{P}\subst{y}{0}
\supseteq M\lsubst{y}{P'}\subst{y}{0}$.

We can conclude that 
$\lam\seq x.M'\lsubst{y}{P'}\subst{y}{0}P'_1\cdots P'_p \subseteq\TE{(\lam \seq x.M\lsubst{y}{[\seq L]}\subst{y}{\sN}P_1\cdots P_p)}$.

All other cases are simpler.
\end{proof}

\begin{rem}
The above proposition is false for regular $\beta$-reduction.
E.g., take $A\ass x[(\bold{I}[y])^\bang]$ and $A'\ass x[\bold{I}[y],\bold{I}[y]]\in\TE{A}$, then 
$A'\to_\beta x[y,\bold{I}[y]]$ and $A\to_\beta x[y^\bang]$ but $x[y,\bold{I}[y]]\notin\TE{( x[y^\bang])}$.
\end{rem}

\begin{cor}\label{cor:magic1} 
Let $\sA,\sB'\in\Fsums{e}$.
If $\sA'\subseteq\TE\sA$ and $\sA'\to_h\sB'$ then there exists $\sB$ 
such that $\sA\to_h\sB$ and $\sB'\subseteq\TE\sB$. \qed
\end{cor}

\begin{cor}\label{cor:magic} 
Let $V\in \FSet{\gt}$ be a closed test. 
If there exists a $V'\in\TE{V}$ such that $V'\msto\varepsilon$,
then $\cOnv{V}$.
\end{cor}

\begin{proof}
Suppose that there exists $V'\in\TE{V}$ such that $V' \msto \varepsilon$.
By Lemma~\ref{lemma:RgoeiffRgohe} there is a head-reduction chain of the form
$V' \to_h  \sV_1' \to_h \cdots \to_h \sV_n' = \varepsilon$.
By iterated application of Corollary~\ref{cor:magic1} there are tests 
$\sV_i$ (for $i=1,\ldots,n$) such that $V \to_h \sV_1 \to_h \cdots \to_h \sV_{n}$ 
with $\sV_i' \subseteq \TE{\sV_i}$. 
We conclude since $\varepsilon \in \TE{\sV_n}$ is only possible when $ \varepsilon\in\sV_n $. 
\end{proof}

\subsection{Full Abstraction for the \dlam-Calculus with Tests}

We are now going to prove that the relational model $\cD$ is inequationally fully abstract for 
the \dlam-calculus with tests.
\begin{lem}\label{lem:aboutTE}
Given $A\in\FSet{e}$ and $M\in\FSet{\gto}$ we have:
\begin{enumerate}[\em(i)]
\item\label{lem:aboutTE1}
	$\TE{(\cont{\alpha}{M})} = \cont{\alpha}{\TE{M}}$, for all $\alpha\in\cD$,
\item\label{lem:aboutTE2}
	$\TE{(A\lsubst{x}{\trm{a}})} = \TE{A}\lsubst{x}{\trm{a}}$, for all $a\in\Mfin{\cD}$.
\end{enumerate}	
\end{lem}

\begin{proof} 
Easy, as $\cont{\alpha}{\cdot}$ and $\trm{a}$ are $\bang$-free,
and the Taylor expansion $\TE{(\cdot)}$ behaves like the identity on $\bang$-free expressions.
\end{proof}

\begin{prop}\label{prop:Fcontext}
Let $M \in \FSet{\gto}$, $\seq{x}\supseteq\FV(M)$, $\alpha\in \cD$ and $\seq a\in\Mfin{\cD}$. 
Then the following statements are equivalent:
\begin{enumerate}[\em(i)]
\item\label{prop:Fcontext1} 
	$(\seq{a},\alpha) \in \Int{M}_{\seq{x}}$,
\item\label{prop:Fcontext2} 
	$\cOnv{\cont{\alpha}{M\lsubst{\seq{x}}{\trm{\seq{a}}}}}$.
\end{enumerate}
\end{prop}

\begin{proof}
(\ref{prop:Fcontext1} $\imp$ \ref{prop:Fcontext2})
Suppose $(\seq a,\alpha)\in\Int{M}_{\seq x}$, then by Theorem~\ref{thm:DmodelsTE} 
there is an $M'\in\TE{M}$ such that $(\seq a,\alpha)\in\Int{M'}_{\seq x}$.
Applying Proposition~\ref{prop:context} we know that $\cont{\alpha}{M'\lsubst{\seq x}{\trm{\seq{a}}}}\msto\varepsilon$.
Now, since $\cont{\alpha}{M'\lsubst{\seq x}{\trm{\seq{a}}}}\in\TE{(\cont{\alpha}{M\lsubst{\seq x}{\trm{\seq{a}}}})}$ (by Lemma~\ref{lem:aboutTE}),
we can apply Corollary~\ref{cor:magic} and get $\cOnv{\cont{\alpha}{M\lsubst{\seq x}{\trm{\seq{a}}}}}$.

(\ref{prop:Fcontext2} $\imp$ \ref{prop:Fcontext1}) 
Suppose that $\cont{\alpha}{M\lsubst{\seq{x}}{\trm{\seq{a}}}} \msto \varepsilon + \sV$, for some $\sV$; 
then $\Int{\cont{\alpha}{M\lsubst{\seq{x}}{\trm{\seq{a}}}}}_{\seq x}\neq\emptyset$.
Hence, by Theorem~\ref{thm:DmodelsTE}, there is a closed test $V\in\TE{(\cont{\alpha}{M\lsubst{\seq{x}}{\trm{\seq{a}}}})}$ 
such that $\Int{V}\neq\emptyset$.
By Lemma~\ref{lem:aboutTE} $V = \cont{\alpha}{M'\lsubst{\seq x}{\trm{\seq{a}}}}$ for some $M'\in\TE{M}$ and
since its interpretation is non-empty we have $V\msto\varepsilon$.
By applying Proposition~\ref{prop:context} we get $(\seq a,\alpha)\in\Int{M'}_{\seq x}\subseteq\Int{M}_{\seq x}$ 
(by Theorem~\ref{thm:DmodelsTE}).
\end{proof}

\begin{thm}\label{thm:main2} 
$\cD$ is inequationally fully abstract for the \dlam-calculus with tests (for all $M,N\in\FSet{\gto}$):
$$
\Int{M}_{\seq x}\subseteq \Int{N}_{\seq x} \iff M\Fobsle N.
$$
\end{thm}

\begin{proof} $(\imp)$ Suppose that $\Int{M}_{\seq x}\subseteq\Int{N}_{\seq x}$ 
and there is a test-context $C\hole{\cdot}$ (closing $M,N$) such that $\cOnv{C\hole{M}}$.
Since $C\hole{M}\msto\varepsilon + \sV$, for some $\sV$, we have $\Int{C\hole{M}}\neq\emptyset$.
Thus, by monotonicity of the interpretation we get 
$\Int{C\hole{M}}\subseteq\Int{C\hole{N}} = \Int{\TE{(C\hole{N})}}\neq\emptyset$.
By Corollary~\ref{cor:closed-tests-go-to-eps-zero} there is $V\in\TE{(C\hole{N})}$ such that $V\msto\varepsilon$ and we conclude 
that $\cOnv{C\hole{N}}$ by applying Proposition~\ref{prop:Fcontext}.

$(\Leftarrow)$ Suppose by contradiction that $M\Fobsle N$, but there is an $(\seq a,\alpha)\in\Int{M}_{\seq x} - \Int{N}_{\seq x}$.
By Proposition~\ref{prop:Fcontext} $\cOnv{\cont{\ga}{M\lsubst{\seq{x}}{\trm{\seq{a}}}}}$ and since $M\Fobsle N$ we 
have $\cOnv{\cont{\ga}{N\lsubst{\seq{x}}{\trm{\seq{a}}}}}$. 
Again, by Proposition~\ref{prop:Fcontext} $(\seq a,\alpha)\in\Int{N}_{\seq x}$. Contradiction.
\end{proof}

\begin{cor}
$\cD$ is equationally fully abstract for the \dlam-calculus with tests.
\end{cor}

\section{Conclusions and Further Works} \label{sec:conclusions}

In this paper we defined the interpretation of several resource calculi into the relational model $\cD$
and characterized the equality induced on the terms from an operational point of view.
The analogous question for untyped \lam-calculus was addressed in \cite{Manzonetto09}, where it is shown that the \lam-theory induced by $\cD$ is $\mathcal{H}^\star$,
therefore $\cD$ is fully abstract for \lam-calculus.

In the first result of our paper we proved that the model $\cD$ is also (in)equationally fully abstract for the \dzlam-calculus with tests. 
Such a proof is simplified by the absence of promotion in the calculus, which allows us to work in a strongly normalizing framework.
The interest of this proof is that it generalizes along two directions.

The first direction aims to get rid of the tests, while remaining in the promotion-free fragment of the calculus.
To extend this result to the \dzlam-calculus without tests we defined the test-expansion --- a translation from tests to terms
replacing every occurrence of a test operator $\gt,\gto$ by a suitable number of empty applications or dummy lambda abstraction. 
By applying this translation to a test-context separating two terms, we obtain a term-context having the same discriminatory power.
This is not surprising since everything is finite in the \dzlam-calculus (finite sums, finite reduction chains) therefore the infinitary nature of 
our test operators can be simulated by terms whose size is big enough.

The second direction aims to extend the full abstraction result to the \dlam-calculus with tests (and promotion available on resources).
The main contribution of the paper is to show that this generalization can be done just by combining the properties of the
head reduction and of the Taylor Expansion.

It is worth to notice that the test expansion method cannot be applied in presence of promotion because $\cD$ is not fully abstract 
for the \dlam-calculus; in other words the tests are \emph{necessary} to obtain the last result.
This has been recently showed by Breuvart \cite{Breuvart}, who exhibited two terms of the \dlam-calculus 
being observationally equivalent, but having different interpretations in $\cD$. 
The idea of the counterexample is to build, using fixpoint combinators, a term $M$ reducing (eventually) to an infinite sum of terms whose
head variable is preceded by an increasing number of lambda abstractions. 
This term is annihilated by the context $\gt[\hole{\cdot}[\gto(\varepsilon)]]$ because the operator $\gt$ ``eats'' all the lambda abstractions and 
substitutes the head-variable of each component of the sum by 0, while we know that the same context sends $\bold{I}$ to $\varepsilon$. 
The author then proved that no context of the \dlam-calculus can simulate this behaviour.

The following table summarizes all these results.
The definition of $\obsle^\oc$ is analogous to $\obsle$ with the definition of may-solvable given in \cite{PaganiR10};
the definition of $\preceq$ is the usual one given in \cite{MorrisTh}.

\begin{center}
\begin{tabular}{|c|c|c|}
\hline
Calculus               			& Operational Preorder & $\cD$ is fully abstract \\
\hline
\dlam-calculus with tests   			& $\Fobsle$ & yes (Thm.~\ref{thm:main2} ) \\
\dzlam-calculus with tests			& $\Tobsle$  & yes (Thm.~\ref{thm:main1}) \\
\dlam-calculus     				& $\obsle^\oc$ & no \cite{Breuvart} \\
\dzlam-calculus         			& $\obsle$ & yes (Thm.~\ref{thm:main1-1}) \\
\lam-calculus             			& $\preceq$ & yes (equationally) \cite{Manzonetto09}\\
\hline
\end{tabular}
\end{center}

Breuvart's counterexample raises the problem of finding a model that is actually fully abstract for the \dlam-calculus without tests.

\begin{qu}\label{Q1} Is there a fully abstract model of \dlam-calculus living in the relational semantics?
\end{qu}

It is known that the structure of the underlying Cartesian closed category may effect the theories of all models living in it.
For instance in \cite{Manzonetto10} it is shown that terms having the same Taylor expansion are equated in all models living in $\MRel$.
It is therefore possible that Question~\ref{Q1} admits a negative answer.
If this is the case, then the following question becomes interesting.

\begin{qu}\label{Q2} Is it possible to find a new comonad $T$, such that the (co)Kliesli $\bold{Rel}_T$ contains a fully abstract model of \dlam-calculus?
\end{qu}

Indeed, the comonad $\Mfin{-}$ of finite multisets is not the only one that leads to models of \dlam-calculus.
For instance it has been shown by Carraro, Ehrhard and Salibra in \cite{CarraroES10b} that one can consider exponential functors with infinite multiplicities. 
However, their models do not even validate the Taylor expansion, therefore are not suitable to solve Question~\ref{Q2}.
The challenge is to find other kinds of comonads.

\section*{Acknowledgement}
We wish to acknowledge fruitful discussions with Flavien Breuvart, Michele Pagani and Paolo Tranquilli.

\bibliographystyle{plain}
\bibliography{include/biblio}

\newpage
\appendix
\section{Technical Appendix}\label{app:tech_app}

This technical appendix is devoted to give the proofs of some results in the paper.
The following is an equivalent but slightly more compact version of the linear substitution lemma.

\begin{lem}[Lemma~\ref{lemma:lsubst}, Linear Substitution Lemma]
Let $M \in\Set{\gto}$, $Q\in\Set{\gt}$ and $P\in\Set{b}$. Then we have:
\begin{enumerate}[\em(i)]
\item $(\seq a,\alpha)\in \Int{M\lsubst{y}{P}}_{\seq x}$ if and only if 
	$\exists d \in \Mfin{\cD}$, $\exists \seq b,\seq c \in \Mfin{\cD}^n$ such that
        $((\seq b,d),\alpha) \in \Int{M}_{\seq x,y}$,
        $(\seq c,d) \in \Int{P}_{\seq x}$,  
        $\seq a = \seq b \mcup \seq c$, and 
        $\sharp d = \sharp P = \dg{y}{M}$.
\item $\seq a \in \Int{Q\lsubst{y}{P}}_{\seq x}$ if and only if 
	$\exists d \in \Mfin{\cD}$, $\exists \seq b,\seq c \in \Mfin{\cD}^n$ such that
        $(\seq b,d) \in \Int{Q}_{\seq x,y}$,
        $(\seq c,d) \in \Int{P}_{\seq x}$,  
        $\seq a = \seq b \mcup \seq c$, and 
        $\sharp d = \sharp P = \dg{y}{Q}$.
\end{enumerate}
\end{lem}

\proof The points (i) and (ii) are proved by mutual induction on $M$ and $Q$.

\begin{enumerate}[(i)] 
\item We only treat the case $M = N_0\Bag{N_1,\ldots,N_h}$.

($\Rightarrow$) 
First, let us call $\mathfrak{P}$ the set of all sequences $\seq {P'} = (P_0',\ldots,P_h')$ of bags such that $P_0' \mcup \cdots \mcup P_h' = P$ and $\sharp P_j' = \dg{y}{N_j}$ for all $j=0,\ldots,h$. 
Also, note that by definition of linear substitution we have 
\[(N_0[N_1,\ldots,N_h])\lsubst{y}{P} = \Sigma_{\seq {P'} \in
  \mathfrak{P}}
N_0\lsubst{y}{P_0'}[N_1\lsubst{y}{P_1'},\ldots,N_h\lsubst{y}{P_h'}].
\]
 Hence, by definition of interpretation, we have that $(\seq a, \alpha)\in \Int{M\lsubst{y}{P}}_{\seq x}$ iff there exist $\seq {P'} \in \mathfrak{P}$, $\alpha_1,\ldots,\alpha_h \in \cD$, $\seq a_0,\ldots,\seq a_h \in \Mfin{\cD}^n$ such that $(\seq a_0,[\alpha_1,\ldots,\alpha_h] \at \alpha)\in \Int{N_0\lsubst{y}{P_0'}}_{\seq x}$, $(\seq a_j,\alpha_j) \in \Int{N_j\lsubst{y}{P_j'}}_{\seq x}$ (for $1\le j\le h$), and $\seq{a} = \mcup_{i=0}^h \seq a_i$. Now by applying the induction hypothesis (i) we obtain that:
\begin{iteMize}{$\bullet$}
\item $\exists d_0 \in \Mfin{\cD}$, 
	 $\exists \seq b_{0},\seq c_{0} \in \Mfin{\cD}^n$ such that
        $((\seq c_0,d_0),[\alpha_1,\ldots,\alpha_h] \at \alpha)  \in \Int{N_0}_{\seq x,y}$,
        $(\seq b_0,d_0) \in \Int{P_0'}_{\seq x}$,  
        $\seq a_0 = \seq b_0 \mcup \seq c_0$, and
        $\sharp d_0 = \sharp P_0'$.
\item $\forall j=1,\ldots,k$, $\exists d_j \in \Mfin{\cD}$, 
	 $\exists \seq b_j,\seq c_j \in \Mfin{\cD}^n$ such that
        $((\seq c_j,d_j),\alpha_j) \in \Int{N_j}_{\seq x,y}$,
        $(\seq b_j,d_j) \in \Int{P_j'}_{\seq x}$, 
        $\seq a_j = \seq b_j \mcup \seq c_j$, and
        $\sharp d_j = \sharp P_j'$.
\end{iteMize}

\noindent Now let $\seq {c} = \mcup_{j=0}^{h}\seq c_j$, $\seq {b} =
\mcup_{j=0}^{h}\seq b_j$, and $d = \mcup_{j=0}^{h} d_j$. Clearly
 $((\seq {c},d),[\alpha_1,\ldots,\alpha_h] \at
\alpha)\in\Int{M}_{\seq x,y}$, $(\seq {b},d) \in \Int{P}_{\seq x}$,
and $\seq a = \seq {b} \mcup \seq {c}$. This concludes the proof of
the ($\Rightarrow$) implication.

($\Leftarrow$) Suppose that $\exists d \in \Mfin{\cD}$, $\exists \seq b,\seq c \in \Mfin{\cD}^n$ such that $((\seq c,d),\alpha) \in \Int{M}_{\seq x,y}$, $(\seq b,d) \in \Int{P}_{\seq x}$ and $\seq b \mcup \seq c = \seq a$. Now we observe that by the definition of interpretation
\begin{iteMize}{$\bullet$}
\item $\exists d_0,\ldots,d_h \in \Mfin{\cD}$, 
        $\exists \seq c_0,\ldots,\seq c_h \in \Mfin{\cD}^n$ such that
        $((\seq c_0,d_0),[\alpha_1,\ldots,\alpha_h] \at \alpha) \in \Int{N_0}_{\seq x,y}$,
        $((\seq c_j,d_j),\alpha_j) \in \Int{N_j}_{\seq x,y}$ (for $1\le j\le h$),
        $\mcup_{j=0}^{h} \seq c_j = \seq c$, and
        $\mcup_{j=0}^{h} d_j = d$, and
        $\sharp d_j = \dg{y}{N_j}$ (for $1\!\le\!j\!\le\! h$).
\item $\exists P_0',\ldots,P_h' \in \Set{b}$, 
        $\exists \seq b_0,\ldots,\seq b_h \in \Mfin{\cD}^n$ such that 
        $\mcup_{j=0}^{h} P_j' = P$, 
        $\mcup_{j=0}^{h} \seq b_j = \seq b$, 
        $\card{P_j'} = \dg{y}{N_j}$ (for $j=0,\ldots,h$), and
        $(\seq b_j,d_j) \in \Int{P_j'}_{\seq x}$ (for $j=0,\ldots,h$).
\end{iteMize}
Note that $\card{d_j} =\card{P_j'}$ (for $1\leq j\leq h$). Now let $\seq a_j = \seq b_j\mcup\seq c_j$ (for $j=0,\ldots,h$). Then by the induction hypo\-thesis~(i) we have that $(\seq a_0,[\alpha_1,\ldots,\alpha_h] \at \alpha)\in \Int{N_0\lsubst{y}{P_0'}}_{\seq x}$ and $(\seq a_j,\alpha_j) \in \Int{N_j\lsubst{y}{P_j'}}_{\seq x}$ (for $1\le j\le h$), and finally observing that $\seq a = \mcup_{j=0}^{h} \seq a_j$, we can conclude that $(\seq a,\alpha)\in \Int{N_0\lsubst{y}{P_0'}[N_1\lsubst{y}{P_1'},\ldots,N_h\lsubst{y}{P_h'}]}_{\seq x} \subseteq \Int{M\lsubst{y}{P}}_{\seq x}$.

\item We just consider the case $Q = \gt[N]$. By definition of
  interpretation we then have $\Int{\gt[N\lsubst{y}{P}]}_{\seq x} = \{\seq a \st (\seq a,*) \in \Int{N\lsubst{y}{P}}_{\seq x} \}$. Hence applying the induction hypothesis (i) and the fact that $\gt[N]\lsubst{y}{P} = \gt[N\lsubst{y}{P}]$ we conclude that $\Int{\gt[N]\lsubst{y}{P}}_{\seq x} = \{\seq a \mcup \seq b \st \exists d \in \Mfin{\cD},\ \card{d} = \card{P}=\dg{y}{Q},\ (\seq b,d)\in\Int{P}_{\seq x},\ ((\seq a,d)) \in \Int{\gt[N]}_{\seq x,y} \}$.\qed
\end{enumerate}

\begin{lem}[Lemma~\ref{lemma:subst}, Regular Substitution Lemma]
Let $M\in\Set{\gto}$, $Q\in\Set{\gt}$ and $\sN\in\sums{\gto}$. 
We have:
\begin{enumerate}[\em(i)]
\item $(\seq a, \alpha)\in \Int{M\subst{y}{\sN}}_{\seq x}$ iff 
	$\exists\beta_1,\ldots,\beta_k\in \cD$, $k=\dg{y}{M}$,
	$\exists\seq a_0,\ldots,\seq a_{k}\in\Mfin{\cD}^n$ such that
	$(\seq a_i,\beta_i)\in\Int{\sN}_{\seq x}$ (for $1\le i \le k$),
	$((\seq a_0,[\beta_1,\ldots,\beta_k]), \alpha)\in\Int{M}_{\seq x,y}$ and 
	$\seq a = \mcup_{j=0}^{k} \seq a_j$,
\item $\seq a\in \Int{Q\subst{y}{\sN}}_{\seq x}$ iff 
	$\exists\beta_1,\ldots,\beta_k\in \cD$, $k=\dg{y}{Q}$,
	$\exists\seq a_0,\ldots,\seq a_{k}\in\Mfin{\cD}^n$ such that
	$(\seq a_i,\beta_i)\in\Int{\sN}_{\seq x}$ (for $1\le i \le k$) and
	$(\seq a_0,[\beta_1,\ldots,\beta_k])\in\Int{Q}_{\seq x,y}$ and 
	$\seq a = \mcup_{j=0}^{k} \seq a_j$.
\end{enumerate}
\end{lem}

\begin{proof}The items (i) and (ii) are proved by mutual induction on $M$ and $Q$.

(i) We only treat the case $M = N_0\Bag{N_1,\ldots,N_h}$.

($\Rightarrow$) Suppose that $(\seq a,\alpha)\in \Int{M\subst{y}{\sN}}_{\seq x}$. By definition of classic substitution we have 
$(N_0[N_1,\ldots,N_h])\subst{y}{\sN} = N_0\subst{y}{\sN}[N_1\subst{y}{\sN},\ldots,N_h\subst{y}{\sN}]$. 
Hence, by definition, $(\seq a,\alpha)\in \Int{N_0\subst{y}{\sN}[N_1\subst{y}{\sN},\ldots,N_h\subst{y}{\sN}]}_{\seq x}$ if and only if 
there exist $\alpha_1,\ldots,\alpha_h \in \cD$, $\seq a_0,\ldots,\seq a_h \in \Mfin{\cD}^n$ such that $(\seq a_0,[\alpha_1,\ldots,\alpha_h] \at \alpha)\in \Int{N_0\subst{y}{\sN}}_{\seq x}$, $(\seq a_j,\alpha_j) \in \Int{N_j\subst{y}{\sN}}_{\seq x}$ for $1\le j\le h$, and $\seq{a} = \mcup_{i=0}^h \seq a_i$. 

By applying the induction hypothesis (i) we obtain that 
\begin{iteMize}{$\bullet$}
\item $\exists\delta_{0,1},\ldots,\delta_{0,\ell_0}\in \cD$ for $\ell_0=\dg{y}{N_0}$,
	 $\exists \seq b_{0,1},\ldots,\seq b_{0,\ell_0},\seq c_{0} \in
  \Mfin{\cD}^n$ such that \newline
        $((\seq c_0,[\delta_{0,1},\ldots,\delta_{0,\ell_0}]),[\alpha_1,\ldots,\alpha_h] \at \alpha)  \in \Int{N_0}_{\seq x,y},$
        $(\seq b_{0,i},\delta_{0,i}) \in \Int{\sN}_{\seq x}$, (for $1\leq i \leq \ell_0$),
        $(\mcup_{i=1}^{\ell_0} \seq b_{0,i}) \mcup \seq c_0 = \seq a_0$, and
\item $\forall j=1,\ldots,k$, $\exists \delta_{j,1},\ldots,\delta_{j,\ell_j} \in \Mfin{\cD}$, $\ell_j=\dg{y}{N_j}$,
	 $\exists \seq b_j,\seq c_j \in \Mfin{\cD}^n$ such that
        $((\seq c_j,[\delta_{j,1},\ldots,\delta_{j,\ell_j}]),\alpha_j) \in \Int{N_j}_{\seq x,y}$,
        $(\seq b_{j,i},\delta_{j,i}) \in \Int{\sN}_{\seq x}$, (for $1\leq i \leq \ell_j$),
        $(\mcup_{i=1}^{k_j} \seq b_{j,i}) \mcup \seq c_j = \seq a_j$.
\end{iteMize}

\noindent Now let $\seq {c} = \mcup_{j=0}^{h}\seq c_j$, $\seq {b} = \mcup_{j=0}^{h} \mcup_{i=1}^{\ell_j} \seq b_{j,i}$, 
$k = \Sigma_{j=0}^{h} \ell_j$ and $[\beta_1,\ldots,\beta_k] = \mcup_{j=0}^{h} [\delta_{j,1},\ldots,\delta_{j,\ell_j}]$. 
It is easy to see that $((\seq {c},[\beta_1,\ldots,\beta_k]),\alpha)\in\Int{M}_{\seq x,y}$ and $\seq a = \seq {b} \mcup \seq {c}$. 
This concludes the proof of the right implication.

($\Leftarrow$) Suppose that $\exists \beta_1,\ldots,\beta_k \in \cD$, $\exists \seq b_1,\ldots,\seq b_k,\seq c \in \Mfin{\cD}^n$, $k=\dg{y}{M}$, such that $((\seq c,[\beta_1,\ldots,\beta_k]),\alpha) \in \Int{M}_{\seq x,y}$, $(\seq b_i,\beta_i) \in \Int{\sN}_{\seq x}$ (for $1\le i\le k$), and $(\mcup_{i=1}^{k} b_i) \mcup \seq c = \seq a$. Now we observe that by definition of interpretation
\begin{iteMize}{$\bullet$}
\item $\exists\alpha_1,\ldots,\alpha_h \in \cD$,
        $\exists (\seq c_0,d_0),\ldots,(\seq c_h ,d_h)\in \Mfin{\cD}^{n+1}$ such that
        $((\seq c_0,d_0),[\alpha_1,\ldots,\alpha_h] \at \alpha) \in \Int{N_0}_{\seq x,y}$,
        $((\seq c_j,d_j),\alpha_j) \in \Int{N_j}_{\seq x,y}$ (for $1\le j\le h$),
        $\mcup_{j=0}^{h} (\seq c_j,d_j) = (\seq c,[\beta_1,\ldots,\beta_k])$, and
        $\sharp d_j = \dg{y}{N_j}$ (for $0\le j\le h$).
\end{iteMize}
We focus for a moment on the fact that $(\seq b_i,\beta_i) \in \Int{\sN}_{\seq x}$ (for $1\le i\le k$) and $\mcup_{j=0}^{h} d_j = [\beta_1,\ldots,\beta_k]$. Thus there exists a way of partitioning the set $\{1,\ldots,k\}$ into $h+1$ subsets $X_0, \ldots, X_h$ in such a way that for all $j =0,\ldots,h$ each  $i \in X_j$ is such that $\beta_i \in d_j$. Then we let $\seq e_j = \mcup_{i \in X_j} \seq b_i$. 

Now let $\seq a_j = \seq e_j\mcup\seq c_j$ (for $j=0,\ldots,h$). Then by induction hypothesis (i) we have that $(\seq a_0,[\alpha_1,\ldots,\alpha_h] \at \alpha)\in \Int{N_0\subst{y}{\sN}}_{\seq x}$ and $(\seq a_j,\alpha_j) \in \Int{N_j\subst{y}{\sN}}_{\seq x}$ (for $1\le j\le h$). 
Finally observing that $\seq a = \mcup_{j=0}^{h} \seq a_j$, we conclude $(\seq a,\alpha)\in \Int{N_0\subst{y}{\sN}[N_1\subst{y}{\sN},\ldots,N_h\subst{y}{\sN}]}_{\seq x} = \Int{M\subst{y}{\sN}}_{\seq x}$.

(ii) We just consider the case $Q = \gt[M]$. By definition of interpretation we have $\Int{\gt[M\subst{y}{\sN}]}_{\seq x} = \{\seq a \st (\seq a,*) \in \Int{M\subst{y}{\sN}}_{\seq x} \}$. Hence applying the induction hypothesis (i) and the fact that $\gt[M]\subst{y}{\sN} = \gt[M\subst{y}{\sN}]$ we conclude that $\Int{\gt[M]\subst{y}{\sN}}_{\seq x} = \{\seq a \mcup (\mcup_{i=1}^{k} \seq b_i) \st \exists\beta_1,\ldots,\beta_k\in \cD,\ k=\dg{y}{M},\ (\seq b_i,\beta_i)\in\Int{\sN}_{\seq x}\ (1\leq i\leq k),\ ((\seq a,[\beta_1,\ldots,\beta_k])) \in \Int{\gt[M]}_{\seq x,y} \}$.
\end{proof}

Let $\ell$ be a function from $\nat$ to $\nat$.
Given a natural number $k\in\nat$ we write $\ell[n:=k]$ for the map $\ell'$ which coincides with $\ell$, 
except on $n$, where $\ell'$ takes the value $k$.
We let $(\ell + k)$ denote the function $\ell'$ defined by $\ell'(x) = \ell(x) + k$.
 
In the following proofs we write $\sA\msto_h^n \sB$ if $\sA$ reduces to $\sB$ in $n$ steps of head reduction, which is
introduced in Section~\ref{subsec:HeadRed} for the full $\dlam$-calculus with tests.

\begin{lem}(Lemma~\ref{lemma:exists-forall})
Let $\sV\in\bool\langle(\Set{\gt})^\lab\rangle$ be a sum of labelled closed tests. 
If $\sV\msto\varepsilon$ then there exists a map $\ell:\nat\to\nat$ such that $\sV^{(\ell+k)}$ is solvable for all $k\in\nat$.
\end{lem}

\begin{proof} 
In the proof we use the characterization of solvable given in Theorem~\ref{thm:equivalence-solvable-normalizable}(\ref{thm:equivalence-solvable-normalizable2}).
We proceed by induction on the length $n$ of a head reduction $\sV\msto_h \varepsilon$ (by Lemma~\ref{lemma:RgoeiffRgohe}).
For the sake of simplicity we assume that in the sum $\sV$ we first reduce a component that head reduces to $\varepsilon$ 
(only when $\sV = \varepsilon + \sW$ we start reducing within $\sW$).

Case $n=0$. Then $\sV = \varepsilon$ and $\sV^\ell = \lam x.x[]$ independently from $\ell$.

Case $n > 0$. We have $\sV\to_h \sV'\msto^{n-1}_h \varepsilon$. The proof is divided into sub-cases depending on the redex that is contracted. 

Subcase $\sV = \gt[(\gto_j(\varepsilon))_i, (\seq L)_{\seq r}]$ and $\sV' = \gt[(\seq L)_{\seq r}]$. 
By induction hypothesis there is $\ell'$ such that $\sV'^{(\ell' +k)} = \lam z.z[\seq L^{(\ell' +k)}[]^{\sim\ell'(\seq r) + k}]$ is solvable for all $k$. 
From this it follows that $\seq L^{(\ell' +k)}[]^{\sim\ell'(\seq r) + k}$ are solvable for all $k$.
It is enough to take $\ell = \ell'[i:= 0][j:=0]$ to have $\sV^{(\ell + k)} = \lam z.z[(\lam x_1\ldots x_{k}.\lam y.y[]) []^{\sim k},\seq L^{(\ell +k)}[]^{\sim\ell'(\seq r) + k}]
\msto \lam z.z[\lam y.y[],\seq L^{(\ell +k)}[]^{\sim\ell(\seq r) + k}]$ that is solvable because $\seq L^{(\ell +k)} = \seq L^{(\ell' +k)}$
by Remark~\ref{rem:eq-dom-entails-eq}(\ref{rem:eq-dom-entails-eq2}) and $\ell(\seq r) = \ell'(\seq r)$.

Subcase $\sV = \gt[(\gto_j(V))_{i}, (\seq L)_{\seq r}]$ where $V = \gt[(\seq M)_{\seq m}]\neq\varepsilon$,
and $\sV' = \gt[(\seq M)_{\seq m}, (\seq L)_{\seq r}]$. 
By induction hypothesis there is $\ell'$ such that $\sV'^{(\ell' +k)} = 
\lam z.z[\seq M^{(\ell' +k)}[]^{\sim\ell'(\seq m)+k},\seq L^{(\ell' +k)}[]^{\sim\ell'(\seq r)+k}]$ is solvable for all $k$. 
From this it follows that $\seq M^{(\ell' +k)}[]^{\sim\ell'(\seq m)+k}$ and $\seq L^{(\ell' +k)}[]^{\sim\ell'(\seq r)+k}$ are solvable for all $k$.
It is enough to take $\ell = \ell'[i:= 0][j:=0]$ to have 
$$
\begin{array}{rl}
\sV^{(\ell + k)} = &\lam z.z[(\lam x_1\ldots x_{k}.\lam y.y[\seq M^{(\ell +k)}[]^{\sim\ell(\seq m) + k}]) []^{\sim k},\seq L^{(\ell +k)}[]^{\sim\ell(\seq r) +k}]\\
\msto &\lam z.z[\lam y.y[\seq M^{(\ell +k)}[]^{\sim\ell(\seq m) + k}],\seq L^{(\ell +k)}[]^{\sim\ell(\seq r) +k}]\\
\end{array}
$$ 
that is solvable because $\ell'(\seq m) = \ell(\seq m)$, $\ell'(\seq r) = \ell(\seq r)$ and, by Remark~\ref{rem:eq-dom-entails-eq}(\ref{rem:eq-dom-entails-eq2}), 
$\seq M^{(\ell +k)}=\seq M^{(\ell' +k)}$ and $\seq L^{(\ell +k)} = \seq L^{(\ell' +k)}$.

Subcase $\sV = \gt [(\lam x. M)_i, (\seq L)_{\seq r}]$ and $\sV' = \gt[(M\subst{x}{0})_i,(\seq L)_{\seq r}]$.
Since $\sV'$ converges, we have $x\notin\FV(M)$ and $M\subst{x}{0} = M$.
By induction hypothesis there is a map $\ell'$ such that, for all $k$,  $\sV'^{(\ell'+k)} = \lam z.z[ M^{(\ell'+k)}[]^{\sim\ell'(i)+k}, \seq L^{(\ell'+k)}[]^{\sim \ell'(\seq r) + k}]$ is solvable.
It is enough to take $\ell = \ell'[i:=\ell'(i) + 1]$ to have $\sV^{(\ell+k)} = \lam z.z[(\lam x. M^{(\ell+k)})[]^{\sim\ell(i)+k+1}, \seq L^{(\ell+k)}[]^{\sim\ell(\seq r)+k}]\to
\lam z.z[M^{(\ell+k)}[]^{\sim\ell(i)+k}, \seq L []^{\sim\ell(\seq r)+k}]$ that is solvable for all $k$ by Remark~\ref{rem:eq-dom-entails-eq}(\ref{rem:eq-dom-entails-eq2})
and $\ell(\seq r) = \ell'(\seq r)$.

Subcase $\sV = \gt[(\gto_j(V)[]\seq P)_i,(\seq L)_{\seq r}]$ and $\sV' = \gt[(\gto_j(V)\seq P)_{i},(\seq L)_{\seq r}]$.
By induction hypothesis there exists $\ell'$ such that,  for all $k$, the term 
$$ 
	\sV'^{(\ell'+k)} = \lam z.z[(\lam x_1\ldots x_{\ell'(j) + k}. V^{(\ell'+k)})\seq P^{(\ell'+k)}[]^{\sim\ell'(i)+k},\seq L^{(\ell'+k)}[]^{\sim\ell'(\seq r)+k}]\textrm{ is solvable.}
$$
For $\ell = \ell'[j:= \ell'(j)+1]$ we have 
$$
	\begin{array}{rl}
	\sV^{(\ell+k)} = &\lam z.z[(\lam x_1\ldots x_{\ell'(j) + k + 1}.V^{(\ell+k)})[]\seq P^{(\ell+k)}[]^{\sim\ell(i)+k},\seq L^{(\ell+k)}[]^{\ell(\seq r)+k}]\\
	\to_\beta&\lam z.z[(\lam x_2\ldots x_{\ell'(j) + k + 1}. V^{(\ell+k)})\seq P^{(\ell+k)}[]^{\sim\ell(i)+k},\seq L^{(\ell+k)}[]^{\ell(\seq r)+k}]\\
	=_{\alpha}&\lam z.z[(\lam x_1\ldots x_{\ell'(j) + k}. V^{(\ell+k)})\seq P^{(\ell+k)}[]^{\sim\ell(i)+k},\seq L^{(\ell+k)}[]^{\ell(\seq r)+k}]\\	
	\end{array}
$$
We conclude by Remark~\ref{rem:eq-dom-entails-eq}(\ref{rem:eq-dom-entails-eq2}) and since $\ell(\seq r) = \ell'(\seq r)$.

Subcase $\sV = \gt[((\lam x.M)Q\seq P)_i,(\seq L)_{\seq r}]$ and $\sV' = \gt[(M\lsubst{x}{Q}\subst{x}{0}\seq P)_i,(\seq L)_{\seq r}]$.
By IH, there is $\ell$ such that $\sV'^{(\ell+k)}
= \lam z.z[M^{(\ell+k)}\lsubst{x}{Q^{(\ell+k)}}\subst{x}{0}\seq P^{(\ell+k)}[]^{\sim\ell(i)+k},\seq L^{(\ell+k)}[]^{\sim\ell(\seq r)+k}]$ is solvable for all $k$. 
We conclude because 
$$
	\sV^{(\ell+k)} = \lam z.z[(\lam x.M^{(\ell+k)})Q^{(\ell+k)}\seq P^{(\ell+k)}[]^{\sim\ell(i)+k},\seq L^{(\ell+k)}[]^{\sim\ell(\seq r)+k}] \to_\beta \sV'^{(\ell+k)}.
$$

Subcase $\sV = V_1+\sW$ and $\sV' = \sV_1+\sW$ with $V_1\to_h \sV_1\msto_h^m\varepsilon$ for some $m<n$. 
By induction hypothesis there is $\ell'$ such that $\sV_1^{(\ell' + k)}$ is solvable for all $k$.
From the previous case analysis, there exists an $\ell$ such that,  for all $k$, $V_1^{(\ell + k)}$ is solvable if and only if $\sV_1^{(\ell' + k)}$ is solvable.
Then $\sV^{(\ell + k)}$ is solvable for all $k$, regardless the behaviour of $\sW^{(\ell+k)}$.

Subcase $\sV = \varepsilon+\sW$ and $\sV' = \varepsilon+\sW'$ with $\sW\to_h \sW'$. Trivial.
\end{proof}

\begin{lem}(Lemma~\ref{lemma:forall-exists})
Let $\sV\in\bool\langle(\Set{\gt})^\lab\rangle$ be a sum of labelled closed tests. 
If $\sV\msto 0$ then there exists a natural number $k$ such that $\sV^{(\ell+k)}\msto 0$ for all $\ell:\nat\to\nat$.
\end{lem}

\begin{proof}We proceed by induction on the length $n$ of a head reduction $\sV\msto_h 0$ (by Lemma~\ref{lemma:RgoeiffRgohe}).

Case $n=0$. Then $\sV = 0$ and $0^\ell = 0$ independently from $\ell$.

Case $n > 0$. We have $\sV\to_h \sV'\msto^{n-1}_h 0$. For the sake of simplicity we assume that in a test 
we always reduce a component head reducing to 0.
The proof is divided into subcases.

Subcase $\sV = \gt[(\gto_j(V))_i,(\seq L)_{\seq r}]$ where $V=\gt[(\seq M)_{\seq m}]$ and 
$\sV' = \gt[(\seq M)_{\seq m},(\seq L)_{\seq r}]$. 
By induction hypothesis there is $k$ such that $\sV'^{(\ell +k)} = 
\lam z.z[\seq M^{(\ell +k)}[]^{\sim\ell(\seq m) + k},\seq L^{(\ell +k)}[]^{\sim\ell(\seq r) + k}]\msto_\beta 0$ for all $\ell$. 
Then either one of the $M^{(\ell +k)}[]^{\sim\ell(\seq m) + k}$ or one of the $\seq L^{(\ell +k)}[]^{\sim\ell(\seq r) + k}$ reduces to 0.
From this it follows that, for all $\ell$, we have
$$
	\sV^{(\ell + k)} = \lam z.z[(\lam x_1\ldots x_{\ell(j) +k}.\lam y.y[\seq M^{(\ell +k)}[]^{\sim\ell(\seq m)+k}]) []^{\sim \ell(i) + k},\seq L^{(\ell +k)}[]^{\sim\ell(\seq r) +k}]\msto_\beta 0.
$$

Subcase $\sV = \gt[(\lam x. M)_i,(\seq L)_{\seq r}]$ and $\sV' = \gt[(M\subst{x}{0})_i,(\seq L)_{\seq r}]$.
By induction hypothesis there is $k$ such that $\sV'^{(\ell+k)} = \lam z.z[ (M\subst{x}{0})^{(\ell+k)}[]^{\sim\ell(i)+k}, \seq L^{(\ell+k)}[]^{\sim\ell(\seq r)+k}]\msto_\beta 0$ for all $\ell$.
For $k' = k+1$ we have 
$$
\begin{array}{rl}
\sV^{(\ell+k')} = &\lam z.z[(\lam x. M^{(\ell + k')})[]^{\sim\ell(i)+k'}, \seq L^{(\ell+k')}[]^{\sim \ell(\seq r) + k'}]\\
\to_\gb &\lam z.z[(M\subst{x}{0})^{(\ell+k')}[]^{\sim\ell(i)+k}, \seq L^{(\ell+k')}[]^{\sim \ell(\seq r) + k'}].\\
\end{array}
$$
If $x\in\FV(M)$ then $\sV^{(\ell+k')}\ass 0$ and we are done. 
Otherwise, when $x\notin\FV(M)$, we have 
$\sV^{(\ell+k')}=\lam z.z[M^{(\ell'+k)}[]^{\sim\ell'(i)+k}, \seq L^{(\ell'+k)}[]^{\sim\ell'(\seq r)+k}]$ where $\ell'=(\ell + 1)[i:=\ell(i)]$,
therefore $\sV^{(\ell+k')}\msto_\beta 0$ for all $\ell$.

Subcase $\sV = \gt[(\gto_j(V)P\seq P)_i,(\seq L)_{\seq r}]$ with $P\neq []$ and $\sV' = 0$.
For every $k>1$, we have 
$$
\begin{array}{rl}
\sV^{(\ell+k)} = &\lam z.z[(\lam x_1\dots x_{\ell(j)+k}.V^{(\ell+k)})P^{(\ell+k)}\seq P^{(\ell+k)}[]^{\sim \ell(i)+k},\seq L^{(\ell+k)}[]^{\sim\ell(\seq r)+k}]\\
\msto_\beta&\lam z.z[(\lam x_2\dots x_{\ell(j)+k}.V^{(\ell+k)}\lsubst{x}{P^{(\ell+k)}}\subst{x}{0})\seq P^{(\ell+k)}[]^{\sim \ell(i)+k},\seq L^{(\ell+k)}[]^{\sim\ell(\seq r)+k}].\\
\end{array}
$$ 
This is equal to 0 since, by definition of $(\ell+k)$-expansion, $x_1\notin\FV(V^{(\ell+k)})$.

Subcase $\sV = \gt[(\gto_j(V)[]\seq P)_i,(\seq L)_{\seq r}]$ and $\sV' = \gt[(\gto_j(V)\seq P)_i,(\seq L)_{\seq r}]$. 
By induction hypothesis there exists $k$ such that,  for all $\ell$, we have
$$ 
	\sV'^{(\ell+k)} = \lam z.z[(\lam x_1\ldots x_{\ell(j) + k}. V^{(\ell +k)})\seq P^{(\ell+k)}[]^{\sim\ell(i)+k},\seq L^{(\ell+k)}]\msto_\beta 0.
$$
For $k' = k + 1$ we have 
$$
	\begin{array}{rl}
	\sV^{(\ell + k')} = &\lam z.z[(\lam x_1\ldots x_{\ell(j) + k'}.V^{(\ell+k')})[]\seq P^{(\ell+k')}[]^{\sim\ell(i)+k'},\seq L^{(\ell+k')}[]^{\sim\ell(\seq r)+k'}]\\
	\msto_\beta&\lam z.z[(\lam x_2\ldots x_{\ell(j) + k'}. V^{(\ell+k')})\seq P^{(\ell+k')}[]^{\sim\ell(i)+k'},\seq L^{(\ell+k')}[]^{\sim\ell(\seq r)+k'}]\\
	=&\lam z.z[(\lam x_1\ldots x_{\ell'(j) + k}. V^{(\ell'+k)})\seq P^{(\ell'+k)}[]^{\sim\ell'(i)+k},\seq L^{(\ell'+k)}[]^{\sim\ell'(\seq r)+k}]\\
	\end{array}
$$
where $\ell' = (\ell+1)[j:=\ell(j)]$.

Subcase $\sV = \gt[((\lam x.M)P\seq P')_i,(\seq L)_{\seq r}]$ and $\sV' = \gt[(M\lsubst{x}{P}\subst{x}{0}\seq P')_i,(\seq L)_{\seq r}]$.
Now, if $\deg_{x}(M) \neq \card P$ then $\sV' = 0$ and we are done. Otherwise, by induction hypothesis, there is $k$ such that, for all $\ell$, we have
$$
\sV'^{(\ell+k)} = \lam z.z[M^{(\ell+k)}\lsubst{x}{P^{(\ell+k)}}\subst{x}{0}\seq P'^{(\ell+k)}[]^{\sim\ell(i)+k},\seq L^{(\ell+k)}[]^{\sim\ell(\seq r)+k}]\msto_\beta 0.
$$
This case follows since $\sV^{(\ell+k)} = \lam z.z[(\lam x.M^{(\ell+k)})P^{(\ell+k)}\seq P'^{(\ell+k)}[]^{\sim\ell(i)+k}, \seq L^{(\ell+k)}[]^{\sim\ell(\seq r)+k}]\to \sV'^{(\ell+k)}$.

Subcase $\sV = V_1+\sW$ and $\sV' = \sV_1+\sW$ with $V_1\to_h \sV'_1$. 
By induction hypothesis there is $k'$ such that $\sV'^{(\ell + k')} = \sV_1^{(\ell + k')}+\sW^{(\ell + k')}\msto_\beta 0$ for all $\ell$.
From the previous case analysis we get a $k\ge k'$ such that $V_1^{\ell+k}\msto_\beta 0$.
We conclude since $\sV^{\ell+k} = V_1^{\ell+k}+\sW^{\ell+k} = V_1^{\ell+k}+\sW^{(\ell+k-k')+k'}\msto_\beta 0$.
\end{proof}

\begin{lem} (Lemma~\ref{lemma:TEcommutesubsts})
Let $A\in\FSet{e}$, $N\in\FSet{\gto}$ and $\sN\in\Fsums{\gto}$. Then:
\begin{enumerate}[\em(i)]
\item
	$\TE{(A\lsubst{x}{N})} = \TE{A}\lsubst{x}{\TE{N}}$,
\item
	$\TE{(A\subst{x}{\sN})} = \bigcup_{P\in\Mfin{\TE{\sN}}}\TE{A}\lsubst{x}{P}\subst{x}{0}$.
\end{enumerate}
\end{lem}

\proof\hfill
\begin{enumerate}[(i)]
\item  By structural induction on $A$.
We only treat the case $A = M[\seq L,\sN^\bang]$. Observe that
$$ \TE{A} = \cup_{P' \in \Mfin{\TE{\sN}}} \TE{M}([\TE{\seq {L}}] \mcup P') $$
By definition of linear substitution we have 
$$
\hspace{-2pt}\begin{array}{ll}
& \ \ \ \ \TE{(A\lsubst{x}{N})} \\
& = \TE{(M\lsubst{x}{N}[\seq L,\sN^\bang])} \cup \\
& \ \ \ \ \cup_{i=1}^{k} \TE{(M[L_1,\pts,L_i\lsubst{x}{N},\pts,L_k,\sN^\bang])} \cup \\
& \ \ \ \ \cup\ \TE{(M[\seq L,\sN\lsubst{x}{N},\sN^\bang])} \\
& = \cup_{P\in \Mfin{\TE{\sN}}} \TE{(M\lsubst{x}{N})}([\TE{\seq {L}}]\mcup P) \cup \\
& \ \ \ \ \cup_{P'\in \Mfin{\TE{\sN}}} \cup_{i=1}^{k} \TE{M}([\TE{L_1},\pts,\TE{(L_i\lsubst{x}{N})},\pts,\TE{L_k}] \mcup P') \cup \\
& \ \ \ \ \cup_{P''\in \Mfin{\TE{\sN}}} \TE{M}([\TE{\seq {L}},\TE{(\sN\lsubst{x}{N})}] \mcup P'') \\
& = \cup_{P\in \Mfin{\TE{\sN}}} \TE{M}\lsubst{x}{\TE{N}}([\TE{\seq {L}}] \mcup P) \cup \\
& \ \ \ \ \cup_{P'\in \Mfin{\TE{\sN}}} \cup_{i=1}^{k} \TE{M}([\TE{L_1},\pts,\TE{L_i}\lsubst{x}{\TE{N}},\pts,\TE{L_k}] \mcup P') \cup \\
& \ \ \ \ \cup_{P''\in \Mfin{\TE{\sN}}} \TE{M}([\TE{\seq {L}},\TE{\sN}\lsubst{x}{\TE{N}}]\mcup P'') \\
& \ \ \ \ \textrm{by induction hypothesis,} \\
& = \cup_{P\in \Mfin{\TE{\sN}}} (\TE{M}([\TE{\seq {L}}] \mcup P))\lsubst{x}{\TE{N}} \\
& = \TE{A}\lsubst{x}{\TE{N}} \\
\end{array}
$$

\item By structural induction on $A$.
Also here we only treat one case, namely $A = M[\seq L,\sM^\bang]$ (where $\card{[\seq {L}]} = k$). In such a case we have
$$
\hspace{-2pt}
\begin{array}{ll}
& \ \ \ \ \cup_{P\in\Mfin{\TE{\sN}}} \TE{A}\lsubst{x}{P}\subst{x}{0} \\
& = \cup_{P' \in \Mfin{\TE{\sM}}} \cup_{P\in\Mfin{\TE{\sN}}} (\TE{M}([\seq {\TE{L}}] \mcup P'))\lsubst{x}{P}\subst{x}{0} \\
& = \cup_{P' \in \Mfin{\TE{\sM}}}\cup_{P_0,P_1,P_2 \in\Mfin{\TE{\sN}}} \\
& \ \ \ \ \TE{M}\lsubst{x}{P_0}\subst{x}{0}([\TE{\seq {L}}]\lsubst{x}{P_1}\subst{x}{0} \mcup P'\lsubst{x}{P_{2}}\subst{x}{0}) \\
& = \cup_{P' \in \Mfin{\TE{(\sM\subst{x}{\sN})}}} \TE{(M\subst{x}{\sN})}(\TE{([\seq {L}]\subst{x}{\sN})} \mcup P') \\
& \ \ \ \ \textrm{by induction hypothesis, using the fact that 	} \\
& \ \ \ \	\cup_{P' \in \Mfin{\TE{\sM}}}\cup_{P_2 \in\Mfin{\TE{\sN}}}P'\lsubst{x}{P_{2}}\subst{x}{0}\textrm{ is equal to }\\
& \ \ \ \	\Mfin{\cup_{P\in\Mfin{\TE{\sN}}}\TE{\sM}\lsubst{x}{P}\subst{x}{0}}\\
& = \TE{(M\subst{x}{\sN}[\seq {L}\subst{x}{\sN},\sM\subst{x}{\sN}^\bang])} \\
& = \TE{(A\subst{x}{\sN})}\rlap{\hbox to279 pt{\hfill\qEd}}
\end{array}
$$
\end{enumerate}

\end{document}